\documentclass[11pt]{amsart}
\usepackage{graphicx,amssymb,amsmath, latexsym,cite, amscd,stmaryrd}
\usepackage[foot]{amsaddr}  
\usepackage{amsthm,amsrefs} 
\usepackage{enumerate,lscape} 
\usepackage{float,xcolor}  
\usepackage{fancyhdr, lastpage}   
\usepackage{libertine}     
\usepackage{url}    
\usepackage{verbatim,lineno}  
\usepackage{mathtools}
\usepackage{amsfonts}  
\usepackage{a4wide} 
\usepackage{pdflscape}
\usepackage{arydshln}
 
\usepackage[a4paper, left=2.3cm, right=2.3cm, top=3cm, bottom=3cm]{geometry}

\usepackage{tikz}
\usetikzlibrary{arrows,decorations.markings}
\usetikzlibrary{matrix}
\usetikzlibrary{graphs}  
\usetikzlibrary{backgrounds}

\def\DynkinNodeSize{1.5mm}   
\def\DynkinArrowLength{1.5mm}
\tikzset{
dnode/.style={
circle,
inner sep=0pt,
minimum size=\DynkinNodeSize,
fill=white,
draw},
middlearrow/.style={
decoration={markings,
mark=at position 0.6 with
{\draw (0:0mm) -- +(+135:\DynkinArrowLength); \draw (0:0mm) -- +(-135:\DynkinArrowLength);},
},
postaction={decorate}
},
leftrightarrow/.style={
decoration={markings,
mark=at position 0.999 with
{
\draw (0:0mm) -- +(+135:\DynkinArrowLength); \draw (0:0mm) -- +(-135:\DynkinArrowLength);
},
mark=at position 0.001 with
{
\draw (0:0mm) -- +(+45:\DynkinArrowLength); \draw (0:0mm) -- +(-45:\DynkinArrowLength);
},
},
postaction={decorate}
},
sedge/.style={
},
dedge/.style={
middlearrow,
double distance=0.5mm,
},
tedge/.style={
middlearrow,
double distance=1.0mm+\pgflinewidth,
postaction={draw}, 
},
infedge/.style={
leftrightarrow,
double distance=0.5mm,
},
}

\theoremstyle{plain} \newtheorem{lemma}{Lemma}[section]
\theoremstyle{plain} 
\theoremstyle{plain} \newtheorem{theorem}[lemma]{Theorem}
\theoremstyle{plain} \newtheorem{corollary}[lemma]{Corollary}

\theoremstyle{definition}  

\newtheorem{definition}{Definition}
\newtheorem{remark}[lemma]{Remark}
\newtheorem{exa}[lemma]{Example}

\numberwithin{equation}{section}

\newcommand{\C}{\mathbb{C}}

\newcommand{\Z}{\mathbb{Z}}
\newcommand{\R}{\mathbb{R}}
\newcommand{\K}{\mathbb{K}}

\newcommand{\GL}{\mathrm{\mathop{GL}}}
\newcommand{\SL}{\mathrm{\mathop{SL}}}

\newcommand{\ad}{\mathrm{\mathop{ad}}}
\newcommand{\diag}{\mathrm{\mathop{diag}}}

\newcommand{\ssl}{\mathfrak{\mathop{sl}}}

\newcommand{\so}{\mathfrak{\mathop{so}}}

\newcommand{\g}{\mathfrak{g}}
\newcommand{\h}{\mathfrak{h}}

\newcommand{\z}{\mathfrak{z}}
\renewcommand{\a}{\mathfrak{a}}

\newcommand{\cc}{\mathfrak{c}}

\newcommand{\wG}{\widehat{G}}

\newcommand{\Dih}{{\rm Dih}}
\newcommand{\OO}{\mathcal{O}}
\newcommand{\CC}{\mathcal{C}}
\newcommand{\MM}{\mathcal{M}}
\newcommand{\QQ}{\mathcal{Q}}

\newcommand{\SmallMatrix}[1]{\text{\begin{tiny}${\arraycolsep=0.4\arraycolsep\ensuremath
{\begin{pmatrix}#1\end{pmatrix}}}$\end{tiny}}}

\newcommand{\mye}[1]{\text{\small$|#1\rangle$}}

\renewcommand{\leq}{\leqslant}

\renewcommand{\geq}{\geqslant}

\newcommand{\labelto}[1]{\xrightarrow{\makebox[1.5em]{\scriptsize ${#1}$}}}

\setpagewiselinenumbers
\reversemarginpar
\subjclass[2000]{}

\newcounter{ithmcount}
\newenvironment{iprf}{\begin{list}{{\rm
	\alph{ithmcount})}}{\usecounter{ithmcount}\labelwidth-5pt
      \leftmargin0pt \topsep3pt \itemsep1pt \parsep2pt}}{\qedhere\end{list}}

\newenvironment{ithm}{\begin{list}{{\rm \alph{ithmcount})}}{\usecounter{ithmcount}\labelwidth18pt
      \leftmargin18pt \topsep3pt \itemsep1pt \parsep2pt}}{\end{list}}

\begin{document}

\title{Classification of four-rebit states}
\subjclass[2000]{}
\author[H. Dietrich]{Heiko Dietrich}
\address[Dietrich, Origlia]{School of Mathematics, Monash University, Clayton VIC 3800, Australia}
\author[W. A. de Graaf]{Willem A.\ de Graaf}
\address[de Graaf]{Department of Mathematics, University of Trento, Povo (Trento), Italy}
\author[A. Marrani]{Alessio Marrani}
\address[Marrani]{IFT, Departamento de F{\'{i}}sica, Univ. de Murcia, Campus de Espinardo, E-30100 Murcia, Spain}
\author[M. Origlia]{Marcos Origlia}
\email{\rm heiko.dietrich@monash.edu, jazzphyzz@gmail.com, marcos.origlia@unc.edu.ar, willem.degraaf@unitn.it}
\thanks{The first, second, and fourth author were supported by an Australian Research Council grant, identifier DP190100317. The third author is supported by a `Maria Zambrano' distinguished researcher grant, financed by the European Union - NextGenerationEU program}
\keywords{}
\date{\today}

\begin{abstract}
We classify states of four rebits, that is, we classify the orbits of the group $\wG(\R) = \SL(2,\R)^4$ in the space $(\R^2)^{\otimes 4}$.
This is the real analogon of the well-known SLOCC operations in quantum
information theory. By constructing the $\wG(\R)$-module $(\R^2)^{\otimes 4}$
via a $\Z/2\Z$-grading of the simple split real Lie algebra of type D$_4$, the
orbits are divided into three groups: semisimple, nilpotent and mixed.
The nilpotent orbits have been classified in Dietrich et al.\ (2017),  
yielding applications in theoretical physics (extremal black holes in the STU model of $\mathcal{N}=2, D=4$ supergravity, see Ruggeri and Trigiante (2017)). 
Here we focus on the semisimple and mixed orbits which
we classify with recently developed methods based on Galois cohomology, see Borovoi et al.\ (2021).
These orbits are relevant to the
classification of non-extremal (or extremal over-rotating) and two-center extremal black hole
solutions in the STU model. 
\end{abstract}

\maketitle  

\section{Introduction} 

\noindent In a recent paper \cite{complex}, we obtained a complete and
irredundant classification of
the orbits of the group $\wG=\SL(2,\mathbb{C})^{4}$ acting on the space
$(\C^2)^{\otimes 4} = \C^2\otimes \C^2\otimes \C^2 \otimes \C^2$. This is
relevant to Quantum Information Theory because it amounts to the classification
of the entanglement states of four pure multipartite quantum bits
(\textit{qubits}) under the group $\wG$ of reversible Stochastic Local Quantum 
Operations assisted by Classical Communication (SLOCC). Here we
obtain the classification of the orbits of the real group
$\wG(\R)=\SL(2,\R)^4$ on the space $(\R^2)^{\otimes 4}$. This is relevant to real quantum mechanics, where the
elements of $(\R^2)^{\otimes 4}$ are called four-rebit states. Via the
\textit{``black hole /qubit correspondence"} our classification has also applications to high-energy
theoretical physics. We refer to Section \ref{secApp} for a short introduction
into rebits and their relevance to extremal black holes in string theory.

The main idea behind the complex classification is to construct the
representation of $\wG$ on $(\C^2)^{\otimes 4}$ using a $\Z/2\Z$-grading $\g=
\g_0\oplus \g_1$ of the simple Lie algebra $\g$ of type D$_4$ (see Section
\ref{secNotComp} for more details). In this construction the spaces
$(\C^2)^{\otimes 4}$ and
$\g_1$ are identified, yielding a Jordan decomposition of the elements of
$(\C^2)^{\otimes 4}$. This way the elements of (and hence the orbits in) the space
$(\C^2)^{\otimes 4}$ are divided into three groups: nilpotent, semisimple,
and mixed. The main result of \cite{complex} is a classification of the
semisimple and mixed elements;  the classification of the corresponding
nilpotent orbits was already completed a decade earlier by Borsten et al.\
\cite{nilp}.

Analogously to the complex case, the representation of $\wG(\R)$ on
$(\R^2)^{\otimes 4}$ can be
constructed using a $\Z/2\Z$-grading $\g(\R) = \g_0(\R)\oplus \g_1(\R)$ of the
split real form $\g(\R)$ of $\g$. Also here the nilpotent orbits have
been classified in previous work, see Dietrich et al.\ \cite{nilporb}: There are 145 nilpotent orbits, and 101 of these turned out to be relevant to
the study of (possibly multi-center) extremal black holes (BHs) in the STU model (see \cite{STU1,STU2}); this application was discussed in full detail in a subsequent paper by Ruggeri and Trigiante \cite{Trigiante2}. While in various papers, 
such as Bossart et al.\ \cite{15,20}, the classification of extremal BH solutions
had been essentially based on the \textit{complex} nilpotent $\wG$-orbits in $\mathfrak{g}_{1}$, a more intrinsic, accurate and 
detailed treatment was provided in \cite{nilporb,Trigiante2}. 

The present paper deals with the real $\wG(\R)$-orbits of semisimple and mixed elements in $(\R^2)^{\otimes 4}$. Such \emph{non-nilpotent} orbits are relevant for the classification of \textit{non-extremal} (or extremal \textit{over-rotating}) as well as of two-center extremal BH solutions in the STU model of $\mathcal{N}=2, D=4$ supergravity, see Section \ref{secApp}.
A detailed discussion of the application of our classification
to the study of BHs goes beyond the scope of the present investigation and we leave it
for  future work. From now on we focus on the mathematical side of this research.

\pagebreak

The methods that we use to classify the  $\wG(\R)$-orbits in
$(\R^2)^{\otimes 4}$ are based on \cite{trivectors,trivectors2} and employ the theory  of Galois cohomology. One of the main implications of this theory is the
following: Let $v\in (\R^2)^{\otimes 4}$ and consider its complex orbit
$\wG v\subset (\C^2)^{\otimes 4}$. Then the $\wG(\R)$-orbits contained in
$\wG v \cap (\R^2)^{\otimes 4}$ are in bijection with the Galois cohomology
set $H^1 (Z_{\wG}(v))$, where $Z_{\wG}(v) = \{ g\in \wG : g v = v\}$ is
the stabiliser of $v$ in $\wG$. So in principle the only thing one has to do is
to compute $H^1( Z_{\wG}(v) )$ for each $\wG$-orbit in $(\C^2)^{\otimes 4}$ that has
a real representative~$v$. This works well for the nilpotent orbits because
they are finite in number and all have real representatives (we do not discuss the nilpotent case here since a classification is already given in \cite{nilporb}). However, for the orbits of semisimple and mixed elements this is not straightforward: firstly, there is an infinite number of them and, secondly, it is a problem to decide  whether a given complex orbit has a real representative or not.

Our approach to classifying semisimple elements is described in Section
\ref{secSS} and analogous to the method developed in \cite{trivectors2}. However, the work in \cite{trivectors2} relies on some specific preliminary results that do not apply to our case; as a first step, we therefore need to establish the corresponding results for the situation discussed here.  Mixed elements are
considered in Section \ref{secMixed}; also here the methods are similar to
those in \cite{trivectors2}. The main difference is that in the case
treated in \cite{trivectors2}, the stabilisers of semisimple elements all
have trivial Galois cohomology. This is far from being the case in the situation discussed here, which requires significant amendments. For example, we will work with sets of 4-tuples $(p,h,e,f)$, which do not explicitly appear in~\cite{trivectors2}.

In the course of our research we have made frequent use of the computer
algebra system GAP4 \cite{gap4}. This system makes it possible to compute
with the simple Lie algebra $\g$ of type D$_4$. We have used additional GAP
programs of our own, for example to compute defining equations of the
stabilisers of elements in the group $\wG$. 

Our main result is summarised by the following theorem,  it is proved with Theorems \ref{thmRSE} and \ref{thmMTE}.

\begin{theorem}\label{thmMain} The following is established.
 \begin{ithm} 
\item  Up to $\wG(\R)$-conjugacy, the nonzero semisimple elements in $(\R^2)^{\otimes 4}$ are the elements in  Tables \ref{tabreps10}--\ref{tabreps}.
\item Up to $\wG(\R)$-conjugacy, the mixed elements in $(\R^2)^{\otimes 4}$ are the elements in Tables \ref{tabrepsME1}-- \ref{tabrepsME_NS_4_2}
\item Up to $\wG(\R)$-conjugacy, the nilpotent elements in $(\R^2)^{\otimes 4}$ are given in \cite[Table I]{nilporb}.  
 \end{ithm}
 The notation used in these tables is explained in Definition \ref{normalizer}.
\end{theorem}

\noindent{\bf Structure of this paper.} In Section \ref{secApp} we briefly comment on applications to real quantum mechanics and non-extremal black holes. In Section \ref{secNotComp} we introduce more notation and recall the known classifications over the complex field; these classifications are the starting point for the classifications over the real numbers. In Section \ref{secGH} we discuss some results from Galois cohomology that will be useful for splitting a known complex orbit into real orbits. Section~\ref{secSS} presents our classification of real semisimple elements; we prove our first main result Theorem~\ref{thmRSE}. In Section~\ref{secMixed} we prove Theorem \ref{thmMTE}, which completes the classification of the real mixed elements. The Appendix contains tables listing our classifications.

\section{Rebits and the black hole/qubit correspondence}\label{secApp}

\subsection{On rebits}

\textit{Real} quantum mechanics (that is, quantum mechanics defined over
real vector spaces) dates back to St\"{u}ckelberg \cite{Stueck}. It 
provides an interesting theory whose study may help to
discriminate among the aspects of \textit{quantum entanglement} which are
unique to standard quantum theory and those aspects which are more generic over
other physical theories endowed with this phenomenon \cite{real-QM}. Real
quantum mechanics is based on the \textit{rebit}, a quantum bit with \textit{%
real} coefficients for probability amplitudes of a two-state system, namely
a two-state quantum state that may be expressed as a real linear combination
of $\left\vert 0\right\rangle $ and $\left\vert 1\right\rangle $ (which can
also be considered as restricted states that are known to lie on a
longitudinal great circle of the Bloch sphere corresponding to only real
state vectors). In other words, the density matrix of the processed quantum
state $\rho $ is real; that is, at each point in the quantum computation, it
holds that $\left\langle x|\rho |y\right\rangle \in \mathbb{R}$  for all $%
\left\vert x\right\rangle $ and $\left\vert y\right\rangle $ in the
computational basis.

As discussed in \cite{real}, \cite{Levay-stringy}, and  \cite[Appendix B]{Borsten}, quantum computation based on rebits is qualitatively
different from the complex case. Following \cite{real-QM}, some entanglement
properties of two-rebit systems have been discussed in \cite{two-rebit},
also exploiting quaternionic quantum mechanics. Moreover, as recalled in
\cite{rebit-1}, rebits were shown in \cite{rebit-0} to be sufficient for
universal quantum computation; in that scheme, a quantum state of $n$
qubits 
\[
\left\vert \psi \right\rangle =\sum\nolimits_{\mathbf{v}\in \mathbb{Z}_{2}^{n}}r_{%
\mathbf{v}}e^{i\theta _{\mathbf{v}}}\left\vert \mathbf{v}\right\rangle \quad (r_{\mathbf{v}}\in \mathbb{R}^{+},\; \theta _{\mathbf{v}}\in \mathbb{R})
\]%
can be encoded into a state of $n+1$ rebits,%
\[
\overline{\left\vert \psi \right\rangle }=\sum\nolimits_{\mathbf{v}\in \mathbb{Z}%
_{2}^{n}}\left( r_{\mathbf{v}}\cos \theta _{\mathbf{v}}\left\vert \mathbf{v}%
\right\rangle \otimes \left\vert R\right\rangle +r_{\mathbf{v}}\sin \theta _{%
\mathbf{v}}\left\vert \mathbf{v}\right\rangle \otimes \left\vert
I\right\rangle \right) ,
\]%
where the additional rebit (which has been also named \textit{universal rebit%
} or \textit{ubit} \cite{rebit-2}), with basis states $\left\vert
R\right\rangle =\left\vert 0\right\rangle $ and $\left\vert I\right\rangle
=\left\vert 1\right\rangle $, allows one to keep track of the real and
imaginary parts of the unencoded $n$-qubit state.

It should also be remarked that in \cite{rebit-3} the three-tangle for three
rebits has been defined and evaluated, resulting to be expressed by the same
formula as in the complex case, but \textit{without} an overall absolute
value sign: thus, unlike the usual three-tangle, the rebit three-tangle can
be negative. In other words, by denoting the pure three rebits state as%
\[
\left\vert \phi \right\rangle =\sum\nolimits_{i,j,k\in \mathbb{Z}_{2}}a_{ijk}\left%
\vert ijk\right\rangle ,
\]%
where the binary indices $i,j,k$ correspond to rebits A, B, C, respectively,
the three-tangle is simply four times the \textit{Cayley's hyperdeterminant}
\cite{Cayley} of the cubic $2\times 2\times 2$ matrix $a_{ijk}$, see \cite{rebit-3}.

\subsection{Rebits and black holes}

In recent years, the relevance of rebits in high-energy theoretical physics
was highlighted by the determination of striking  relations between
the entanglement of pure states of two and three qubits and extremal BHs
holes in string theory. In this framework, which has been subsequently
dubbed as the ``black hole / qubit
correspondence''  (see for example \cite{Oct, BH-1, BH-2} for
reviews and references), rebits acquire the physical meaning of the
electric and magnetic charges of the extremal BH, and they linearly
transform under the generalised electric-magnetic duality group (named
U-duality group in string theory) $\mathcal{G}(\mathbb{R})$ of the
Maxwell-Einstein (super)gravity theory under consideration.\footnote{%
In supergravity, the approximation of \textit{real} (rather than integer)
electric and magnetic charges of the BH is often considered, thus
disregarding the charge quantization.} This development started with the seminal
paper \cite{Duff-Cayley}, in which Duff pointed out that the entropy of the
so-called extremal BPS STU BHs can be expressed in a very compact
way in terms of Cayley's hyperdeterminant \cite{Cayley},
which, as mentioned above, plays a prominent role as the three-tangle in
studies of three-qubit entanglement \cite{rebit-3}. Crucially, the electric
and magnetic charges of the extremal BH, which are conserved due to
the underlying Abelian gauge invariance, are forced to be real because they
are nothing but the fluxes of the two-form field strengths of the Abelian
potential one-forms, as well as of their dual forms, which are real. Later on, for example in \cite{Kallosh-Linde},
\cite{Levay-0, Levay-1}, \cite{Duff-Ferrara-1, Duff-Ferrara-2} and
subsequent developments, Duff's observation was generalised and extended to
non-BPS BHs (which thus break all supersymmetries), also in ($\mathcal{N}>2$)-extended supergravity theories in four and five space-time
dimensions. Further mathematical similarities were thoroughly investigated
by L\'{e}vay, which for instance showed that the frozen values of the moduli
in the calculation of the macroscopic, Bekenstein-Hawking BH entropy
in the STU model are related to finding the canonical form for a pure
three-qubit entangled state, whereas the extremisation of the BPS mass with
respect to the moduli is connected to the problem of finding the so-called
optimal local \textit{distillation protocol} \cite{Levay-2, Levay-3}.

Another application of rebits concerns extremal BHs with two-centers. Multi-center BHs are a natural generalisation of single-center
BHs. They
occur as solutions to Maxwell-Einstein equations in $4D$, regardless of the
presence of local supersymmetry, and they play a prominent role within the
dynamics of high-energy theories whose ultra-violet completion aims at
describing Quantum Gravity, such as $10D$ superstrings and $11D$ M-theory.
In multi-center BHs the \textit{attractor mechanism} \cite{AM1,AM1-1,AM1-2,AM1-3, AM2} is
generalised by the so-called \textit{split attractor flow} \cite{SP,SP1,SP2},
concerning the existence of a co-dimension-one region - the \textit{marginal
stability} wall - in the target space of scalar fields, where a stable
multi-center BH may decay into its various single-center constituents, whose
scalar flows then separately evolve according to the corresponding attractor
dynamics.

In this framework, the aforementioned real fluxes of the two-form Abelian field strengths and of
their duals, which are usually referred to as \textit{electric} and \textit{%
magnetic} charges of the BH, fit into a representation $\mathbf{R}$ of the $%
4D$ $U$-duality group $\mathcal{G}\left( \R\right) $. In the STU model of $\mathcal{N}=2, D=4$ supergravity, $\mathcal{G}\left( \R\right) =\SL(2,\R)^{3}$ and $\mathbf{R}=\left( \R^{2}\right) ^{\otimes 3}$,
and each $\SL(2,\R)^{3}$-orbit supports a \textit{unique} class of
single-center BH solutions. In general, in presence of a multi-center BH
solution with $p$ centers, the dimension $I_{p}$ of the
ring of $\mathcal{G}\left( \R\right) $-invariant homogeneous
polynomials constructed with $p$ distinct copies of the $\SL(2,\R%
)^{3}$-representation charge $\mathbf{R}$ is given by the general formula
\cite{FMOSY}%
\begin{equation}
p\dim _{\R}\mathbf{R}=\dim _{\R}\mathcal{O}_{p}+I_{p},
\label{counting}
\end{equation}%
where $\mathcal{O}_{p}=\mathcal{G}\left( \R\right) /\mathcal{H}%
_{p}\left( \R\right) $ is a generally non-symmetric coset describing
the generic, open $\mathcal{G}\left( \R\right) $-orbit, spanned by
the $p$ copies of the charge representation $\mathbf{R}$, each pertaining to
one center of the multi-center solution. A crucial feature of \textit{%
multi-center} ($p>1$) BHs is that the various ($I_{p}>1$) $\mathcal{G}\left(
\R\right) $-invariant polynomials arrange into multiplets of a
global, \textit{\textquotedblleft horizontal\textquotedblright } symmetry
group\footnote{%
Actually, the \textquotedblleft horizontal\textquotedblright\ symmetry group
is $\GL(p,\R)$, where the additional scale symmetry with respect to
$\SL(p,\R)$ is encoded by the homogeneity of the $\mathcal{G}(%
\R)$-invariant polynomials in the BH charges. The subscript
\textquotedblleft hor\textquotedblright\ stands for \textquotedblleft
horizontal\textquotedblright\ throughout.} $\SL_{\text{hor}}(p,\R)$
\cite{FMOSY}, encoding the combinatoric structure of the $p$-center
solutions of the theory, and commuting with $\mathcal{G}\left( \R%
\right) $ itself. Thus, by considering two-center BHs (that is,  $p=2$ -- an assumption which does not imply any loss of
generality due to tree structure of split attractor flows in the STU
model), it holds that $\dim _{\R}\mathbf{R}=8$ and the
stabiliser of $\mathcal{O}_{p=2}$ has trivial identity connected component.
The two-center version of formula \eqref{counting} in the STU model
yields%
\begin{equation}
{\rm STU}~\text{:~}I_{p=2}=2\dim _{\R}\left( \left( \R^{2}\right)
^{\otimes 3}\right) \mathbf{-}\dim _{\R}\left( \SL(2,\R%
)^{3}\right) =2\cdot 8-9=7,
\end{equation}%
implying that the ring of $\SL(2,\R)^{3}$-invariant homogeneous
polynomials built out of two copies of the tri-fun\-da\-mental representation $%
\left( \R^{2}\right) ^{\otimes 3}$ has dimension $7$. As firstly
discussed in \cite{FMOSY} and then investigated in \cite{ADFMT, ProcStora,
CFMY-Small, FMY-FI}, the seven $\SL(2,\R)^{3}$-invariant generators
of the aforementioned polynomial ring arrange into one quintuplet (in the
spin-$2$ irreducible representation $\mathbf{5}$) and two singlets $\mathbf{1%
}\oplus \mathbf{1}^{\prime }$ of the ``horizontal''  symmetry group $\SL_{\text{hor}}(2,\R)$:
\begin{equation}
I_{p=2}=7=\underset{\text{under~}\SL_{\text{hor}}(2,\R)}{\underset{%
\text{deg~}4}{\mathbf{5}}\oplus \underset{\text{deg~}2}{\mathbf{1}}\oplus
\underset{\text{deg~}4}{\mathbf{1}^{\prime }}},  \label{decomp}
\end{equation}%
where the degrees of each term (corresponding to one or more homogeneous
polynomials) has been reported. The overall semisimple global group
providing the action of the $U$-duality as well as of the ``horizontal'' symmetry on two-center BHs is
\begin{equation}
\SL_{\text{hor}}(2,\R)\otimes \SL(2,\R)^{3}\simeq \SL(2,%
\R)^{4}=\wG(\R),
\end{equation}%
acting on the $\SL_{\text{hor}}(2,\R)$-doublet of $G$-representations
$\mathbf{R}$'s, namely%
\begin{equation}
\R^{2}\otimes \mathbf{R}=\R^{2}\otimes \left( \R%
^{2}\right) ^{\otimes 3}\simeq \left( \R^{2}\right) ^{\otimes 4}.
\end{equation}

Since the \textquotedblleft horizontal\textquotedblright\ factor $\SL_{\text{%
hor}}$ stands on a different footing than the $U$-duality group $\SL(2,%
\R)^{3}$, only the discrete group Sym$_{3}$ of permutations of the
three tensor factors in $\mathbf{R}=\left( \R^{2}\right) ^{\otimes
3} $ should be taken into account when considering two-center BH solutions
in the STU model, to which a classification invariant under Sym$%
_{3}\ltimes \SL(2,\R)^{3}$ thus pertains. Clearly, the two singlets
in the right hand side of (\ref{decomp}) are invariant under the whole $\SL_{\text{hor}}(2,\R)\otimes \SL(2,\R)^{3}$; on the other hand, when
enforcing the symmetry also under the \textquotedblleft
horizontal\textquotedblright\ $\SL_{\text{hor}}(2,\R)$, one must
consider its non-transitive action on the quintuplet $\mathbf{5}$ occurring
in the right hand side of (\ref{decomp}). As explicitly computed (for example, in \cite{FMOSY})
and as known within the classical theory of invariants (see for example \cite%
{Vinberg-Weyl} as well as the Tables of \cite{Kac-80}), the spin-$2$ $\SL_{2}$%
-representation $\mathbf{5}$ has a two-dimensional ring of invariants,
\textit{finitely generated} by a \textit{quadratic} and a \textit{cubic}
homogeneous polynomial :%
\begin{equation}
I_{\text{spin-}2}=\dim _{\R}\left( \mathbf{5}\right) \mathbf{-}\dim
_{\R}\SL_{\text{hor}}(2,\R)=5-3=2=\underset{\text{under~}\SL_{%
\text{hor}}(2,\R)}{\underset{\text{deg~}2}{\mathbf{1}^{\prime \prime
}}\oplus \underset{\text{deg~}3}{\mathbf{1}^{\prime \prime \prime }}}.
\end{equation}%
This results into a four-dimensional basis of $(\rm{Sym}_{3}\ltimes \left( \SL_{%
\text{hor}}(2,\R)\otimes \SL(2,\R)^{3}\right))$-invariant
homogeneous polynomials, respectively of degree $2$, $4$, $8$ and $12$ in
the elements of the two-center BH charge representation space $\left(
\R^{2}\right) ^{\otimes 4}$. However, as discussed in \cite{FMOSY},
a lower degree invariant polynomial of degree $6$ can be introduced and
related to the degree-$12$ polynomial, giving rise to a
4-dimensional basis of $({\rm Sym}_{3}\ltimes \left( \SL_{\text{hor}}(2,\R%
)\otimes \SL(2,\R)^{3}\right))$-invariant homogeneous polynomials
with degrees $2$, $4$, $6$ and $8$, respectively, see \cite{FMOSY}.

We recall that the enforcement of the whole discrete
permutation symmetry Sym$_{4}$ (as done in Quantum Information Theory applications)
allows for the degrees of the four $({\rm Sym}_{4}\ltimes \left( \SL_{\text{hor}}(2,\R)\otimes \SL(2,\R)^{3}\right))$-invariant polynomial
generators to be further lowered down to $2$, $4$, $4$ and $6$; this is explicitly
computed  in \cite{Verstraete, LT} and then discussed in \cite{Levay} in
relation to two-center extremal BHs in the STU model. In all cases, the
lowest-order element of the invariant basis, namely the homogeneous
polynomial \textit{quadratic} in the BH charges, is nothing but the \textit{%
symplectic product} of the two copies of the single-center charge
representation $\mathbf{R}=\left( \R^{2}\right) ^{\otimes 3}$; such
a symplectic product is constrained to be non-vanishing in non-trivial and
regular two-center BH solutions with \textit{mutually non-local} centers
\cite{FMOSY}. This implies that regular two-center extremal BHs are related\enlargethispage{0.4cm}
to non-nilpotent orbits of the whole symmetry $\SL_{\text{hor}}(2,\R)\otimes \SL(2,\R)^{3}$ (with a discrete factor Sym$_{3}$
or Sym$_{4}$, as just specified) on $\left( \R^{2}\right) ^{\otimes 4}$. The application of the classification of such orbits (which are the object of interest in
this paper) to the study of two-center extremal BHs in the
prototypical STU model goes beyond the scope of the present investigation,
and we leave it for further future work.

\section{Notation and classifications over the complex field}\label{secNotComp}

\subsection{The grading} Let  $\g$ be the simple Lie algebra of type D$_4$ defined over the complex
numbers. Let $\Psi$ denote its root system with respect to a fixed Cartan
subalgebra $\mathfrak{t}$. Let $\gamma_1,\ldots,\gamma_4$ be a fixed choice of simple roots such that the Dynkin diagram of $\Psi$ is labelled as follows
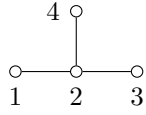
\begin{center}
\begin{tikzpicture}[x=0.8cm,y=0.8cm]
\node[dnode, label=below:{\small $1$}] (1) at (0,1) {};
\node[dnode, label=below:{\small $2$}] (2) at (1,1) {};
\node[dnode, label=below:{\small $3$}] (3) at (2,1) {};
\node[dnode, label=left:{\small $4$}] (4) at (1,2) {};
\path (1) edge[sedge] (2)
(2) edge[sedge] (3)
(2) edge[sedge] (4);
\end{tikzpicture}
\end{center}
We now construct a $\Z/2\Z$-grading of $\g$: let $\g_0$ be spanned by
$\mathfrak{t}$ along with the root spaces $\g_\gamma$, where $\gamma= \sum_i k_i
\gamma_i$ has $k_2$ even, and let $\g_1$ be spanned by those $\g_\gamma$
where $\gamma= \sum_i k_i \gamma_i$ has $k_2$ odd. Let $\gamma_0=\gamma_1+2\gamma_2+\gamma_3+\gamma_4$ be the highest root
of $\Psi$. The  root system of $\g_0$ is $\{\pm \gamma_0, \pm \gamma_1, \pm \gamma_3,\pm\gamma_4\}$, hence
$$\g_0\cong \ssl(2,\C)^4=\ssl(2,\C)\oplus \ssl(2,\C)\oplus\ssl(2,\C)\oplus\ssl(2,\C).$$
Taking $-\gamma_0,\gamma_1,\gamma_3,\gamma_4$ as basis of simple roots of
$\g_0$ we have that $-\gamma_2$ is the highest weight of the $\g_0$-module
$\g_1$, which therefore is isomorphic to $(\C^2)^{\otimes 4}$. We fix a basis
$\{e_0,e_1\}$ of $\C^2$ and denote the basis elements of
$(\C^2)^{\otimes 4}$ by
\[\mye{i_1i_2i_3i_4}= e_{i_1}\otimes e_{i_2}\otimes e_{i_3}\otimes e_{i_4}.\]
Mapping any nonzero root vector in $\g_{-\gamma_2}$ to $\mye{0000}$
extends uniquely to an isomorphism $\g_1 \to (\C^2)^{\otimes 4}$ of $\ssl(2,\C)^4$-modules. We denote  by $G$ the adjoint group of $\g$, and we write $G_0$ for the 
connected algebraic subgroup of $G$ with  Lie algebra $\ad_\g \g_0\cong
\mathfrak{sl}(2,\C)^4$. 
The isomorphism $\ssl(2,\C)^4 \to \g_0$ lifts to a surjective 
morphism $\pi\colon \widehat{G}\to G_0$  of algebraic groups, which makes
$\g_1$ into a $\wG$-module isomorphic to $(\C^2)^{\otimes 4}$.

In order to define a similar grading over $\R$ we take a basis of $\g$
consisting of root vectors and basis elements of $\mathfrak{t}$, whose
real span is a real Lie algebra (for example, we can take a Chevalley basis of
$\g$). We denote this real Lie algebra by $\g(\R)$. We set $\g_0(\R) =
\g_0\cap \g(\R)$ and $\g_1(\R) = \g_1\cap \g(\R)$, so that  \[\g(\R) =
\g_0(\R)\oplus \g_1(\R).\] If $G_0(\R)$ denotes the group of real points of $G_0$, then  $\pi$ restricts to a morphism $\pi \colon \wG(\R) \to G_0(\R)$ that makes $\g_1(\R)$ a $\wG(\R)$-module isomorphic to $(\R^2)^{\otimes 4}$.

A first consequence of these constructions is the existence of a Jordan
decomposition of the elements of the modules $(\C^2)^{\otimes 4}$ and 
$(\R^2)^{\otimes 4}$. Indeed, the Lie algebras $\g$ and $\g(\R)$ have such
decompositions as every element $x$ can be written uniquely as $x=s+n$
where $\ad s$ is semisimple, $\ad n$ is nilpotent, and  $[s,n]=0$. 
It is straightforward to see that if $x$ lies in $\g_1$ or $\g_1(\R)$, then the same holds for its  semisimple and nilpotent parts. Thus, the elements of
$(\C^2)^{\otimes 4}$ and $(\R^2)^{\otimes 4}$ are divided into three groups: semisimple,
nilpotent and mixed. Since the actions of $\wG$ and $\wG(\R)$ respect the Jordan
decomposition, also the orbits of these groups in their respective modules
are divided into the same three groups.

A second consequence is that we can consider $\ssl_2$-triples instead of
nilpotent elements; we use these in Section~\ref{secMixed} when considering  mixed elements: the classification of the orbits of mixed elements
with a fixed semisimple part $p$ reduces to the classification of
the nilpotent orbits in the centraliser of $p$, which in turn reduces to the
classification of orbits of certain $\ssl_2$-triples. We provide more details in Section~\ref{secMixed}.

\subsection{Notation}\label{secNot}
We now we recall the notation used in \cite{complex} to describe the
classification of $\wG$-orbits in $(\C^2)^{\otimes 4}$.

A \emph{Cartan subspace} of $\g_1$ is a maximal space consisting of commuting
semisimple elements. A Cartan subspace $\h$ of $\g_1$ (and in fact a Cartan subalgebra of $\g$) is spanned by  
$$u_1=\mye{0000}+\mye{1111},\;\; u_2=\mye{0110}+\mye{1001},\;\; u_3=\mye{0101}+\mye{1010},\;\; u_4=\mye{0011}+\mye{1100}.$$
We denote by $\Phi$  the corresponding root system with Weyl group $W$.
This group acts on $\Phi$ and $\h$ in the following way.
For $\alpha\in \Phi$ let $s_\alpha\in W$ be the corresponding reflection. If $\beta\in \Phi$ and $h\in \h$, then $s_\alpha(\beta)=\beta-\beta(h_\alpha)\alpha$ and $s_\alpha (h) = h - \alpha(h) h_\alpha$
where $h_\alpha$ is the unique element of $[\g_\alpha,\g_{-\alpha}]\leq \h$ with
$\alpha(h_\alpha)=2$. This defines a $W$-action on $\h$ and we write $W_p=\{\alpha\in W: \alpha(p)=p\}$ for the stabiliser of $p\in\h$ in $W$; the latter is generated by all $s_\alpha$ with $\alpha\in\Phi_p$ where $\Phi_p=\{\alpha\in\Phi:\alpha(p)=0\}$, see \cite[Lemma 2.4]{complex}. For a root subsystem $\Pi\subseteq \Phi$ define
\begin{align*}
  \h_\Pi&=\{p\in \h : \alpha(p)=0\text{ for all }\alpha\in\Pi\}, & 
   W_\Pi &=\langle s_\alpha : \alpha\in \Pi\rangle,\\
 \h_\Pi^\circ&=\{p\in \h_\Pi : \alpha(p)\ne0 \text{ for all }\alpha\in\Phi\setminus\Pi\},&  \Gamma_\Pi& = N_W(W_\Pi)/W_\Pi.
\end{align*}

Let   $\zeta$ be a fixed primitive $8$-th root of unity; for $A=\SmallMatrix{a&b\\c&d}$ and $u\in\C^\times$ we write
\begin{equation}\label{eqmats}
\begin{array}{rclrclrclrclrcl}
  A^\#&=&\SmallMatrix{d&c\\b&a},& D(u)&=&\SmallMatrix{u&0\\0&u^{-1}},& I&=&\SmallMatrix{1&0\\0&1},& J&=&\SmallMatrix{0&1\\-1&0},& K&=&\SmallMatrix{0&\imath \\\imath &0},\\[1.5ex] F&=&\SmallMatrix{1/2 & \imath /2\\\imath &1},&
  L&=&\SmallMatrix{\imath&0\\0&-\imath},&M&=&\SmallMatrix{\zeta^3&0\\0,&-\zeta},&N&=&\SmallMatrix{\zeta&0\\0&-\zeta^3}.
\end{array}
\end{equation}
Throughout this paper,  we freely identify the spaces $\g_1\cong (\C^2)^{\otimes 4}$ and
$\g_1(\R)\cong (\R^2)^{\otimes 4}$ and we write elements of $\wG$ and $\wG(\R)$  as 4-tuples
$(A,B,C,D)$, with  $A,B,C,D$ in $\SL(2,\C)$, respectively in $\SL(2,\R)$. 
 

\subsection{Complex classifications}
In \cite[Section 3.1]{complex} we have determined 11 subsystems $\Pi_1,\ldots,\Pi_{11}$  to classify the semisimple $\wG$-orbits in $\g_1$; these sets are also described in Table \ref{tabW}. The following result summarises \cite[Proposition 2.5, Lemma 2.9, Theorem~3.2, Lemma 3.5, Proposition 3.6]{complex}.

\begin{theorem}[Complex classification of semisimple elements \cite{complex}]\label{thmSE}
 ~{} 
 \begin{ithm} 
 \item Each semisimple $\wG$-orbit in $(\C^2)^{\otimes 4}$ intersects exactly
   one of the sets $\h_{\Pi_i}^\circ$ nontrivially. Two elements of
   $\h_{\Pi_i}^\circ$ are $\wG$-conjugate if and only if they are
   $\Gamma_{\Pi_i}$-conjugate. 
   Each $\Gamma_{\Pi_i}$ can be realised as complement subgroup to $W_{\Pi_i}$ in~$N_W(W_{\Pi_i})$, so as a matrix group relative to the
   basis $u_1,\ldots,u_4$ of $\h$.
   The group $\Gamma_{\Pi_2}\cong (\mathbb{Z}/2\mathbb{Z})^3$ is generated by all $4\times 4$ diagonal matrices that have two $1$s and two $-1$s on the diagonal; the groups $\Gamma_{\Pi_4},\Gamma_{\Pi_5},\Gamma_{\Pi_6}\cong \Dih_4$ are isomorphic to the dihedral group of order $8$ and defined as
\[\Gamma_{\Pi_4}=\langle \SmallMatrix{1&0&0&0\\0&1&0&0\\0&0&1&0\\0&0&0&-1}, \SmallMatrix{0&0&0&-1\\0&0&1&0\\0&1&0&0\\1&0&0&0}\rangle,\;
\Gamma_{\Pi_5}=\langle  \SmallMatrix{1&0&0&0\\0&1&0&0\\0&0&-1&0\\0&0&0&1},\SmallMatrix{0&0&-1&0\\0&0&0&-1\\1&0&0&0\\0&-1&0&0}\rangle,\;
\Gamma_{\Pi_6}=\langle  \SmallMatrix{-1&0&0&0\\0&1&0&0\\0&0&1&0\\0&0&0&1},\SmallMatrix{0&-1&0&0\\1&0&0&0\\0&0&0&-1\\0&0&-1&0}\rangle.
\]
Furthermore, $\Gamma_{\Pi_1}=W$ and the remaining $\Gamma_{\Pi_i}$ are equal to
$\{\pm 1\}$.
\item  If $x,y\in \h_{\Pi_i}^\circ$, then $Z_{\wG}(x)=Z_{\wG}(y)$, and the group
  $Z_{\wG}(x)$ is given in Row $i$ of Table \ref{tabZ}.
 \end{ithm}
 \end{theorem}
See \cite[Remark 3.3]{complex} for a comment on the $\Gamma_{\Pi_i}$-orbits in $\h_{\Pi_i}^\circ$; this yields a complete and irredundant classification of the semisimple $\wG$-orbits in $\g_1$, see Table \ref{tabW}. The next theorem is \cite[Theorem 3.7]{complex}.

\begin{theorem}[Complex classification of mixed elements \cite{complex}]\label{thmME}
  For $i\in\{2,\ldots,10\}$ let $\Sigma_i$ be a set of $\wG$-conjugacy representatives of semisimple elements in $\h_{\Pi_i}^\circ$ as specified in Table \ref{tabW}. Up to $\wG$-conjugacy, the mixed elements in $\g_1$ are the elements $s+n_{i,j}$ where $i\in\{2,\ldots,10\}$, $s\in \Sigma_i$, and $n_{i,j}$ as specified in Table \ref{tabnij}.
\end{theorem}

The nilpotent $\wG$-orbits in $\g_1$ and the nilpotent $\wG(\R)$-orbits in $\g_1(\R)$ are determined in \cite{nilp} and \cite{nilporb}, respectively, see also \cite[Table 7]{complex}; therefore we do not recall these classifications here. We  conclude this section by mentioning \cite[Remark 3.1]{complex}; the symmetries described in this remark allow us to simplify our classifications.

\begin{remark}\label{remSym}
  If $\sigma\in{\rm Sym}_4$, then the linear map $\pi_{\sigma}\colon \g_1\to\g_1$ that maps each $\mye{i_1i_2i_3i_4}$ to $\mye{i_{1^\sigma}i_{2^\sigma}i_{3^\sigma}i_{4^\sigma}}$ extends to a Lie algebra automorphism of $\g$ that preserves $\g_0$ and $\g_1$. The group generated by all these $\pi_\sigma$ fixes $u_1$ and permutes $\{u_2,u_3,u_4\}$ as ${\rm Sym}_3$.  Specifically, $\pi_{(2,3)}$ swaps $u_3$ and $u_4$, and $\pi_{(2,4)}$ swaps $u_2$ and $u_4$.
\end{remark}

\section{Galois cohomology}\label{secGH}
\noindent We describe some results from Galois cohomology that we use for determining the real orbits within a complex orbit; see \cite[Section 3]{trivectors} for a recent treatment of Galois cohomology in the context of orbit classifications.

In this section only we consider the following notation. Let $G$ be a group with \emph{conjugation}  $\sigma \colon G\to G$, that is, an automorphism of $G$ of order 2; often $\sigma$ 
is the complex conjugation of a complex group. An element $c\in G$ is a {\em cocycle} (with respect to $\sigma$) if
$c\sigma(c)=1$; write $Z^1(G,\sigma)$ for the set of all cocycles. Two cocycles
$c,c'$ are \emph{equivalent} if $c'=ac\sigma(a)^{-1}$ for some $a\in G$; the equivalence class of $c$ is denoted $[c]$, and the set of equivalence classes is denoted $H^1(G,\sigma)$. We also write $Z^1(G)$ and $H^1(G)$ if it is clear which conjugation is used; these definitions are an adaption of the definitions in \cite[I.\S 5.1]{serre} to the special case of an acting group $\langle \sigma\rangle$ of size $2$. We now list a few results that help to
determine $H^1(G,\sigma)$. In the following we write $G^\sigma = \{g\in G : \sigma(g)=g\}$.

Let $X$ be a set on which $G$ acts. We suppose that $X$ has a conjugation,
also denoted $\sigma$ (that is, a map $\sigma \colon X\to X$ with $\sigma^2={\rm Id}_X$),
such that  $\sigma(g x) = \sigma(g)\sigma(x)$ for all $x\in X$ and $g\in G$.
Let $\mathcal{O}$ be a $G$-orbit in $X$  that has a real point, that is, there is 
 $x_0\in \mathcal{O}$ with $\sigma(x_0)=x_0$. In this situation, $\mathcal{O}$ is stable under $\sigma$, and we are interested in listing the $G^\sigma$-orbits in $\mathcal{O}^\sigma= \{ y\in \mathcal{O}:
\sigma(y)=y\}$. For this we consider the stabiliser
$$Z_G(x_0) = \{ g\in G: g x_0=x_0\}$$
and the exact sequence
$$1\to Z_G(x_0) \labelto{i} G \labelto{j} \mathcal{O}\to 1$$
resulting from the orbit-stabiliser theorem; here  $j$ maps $g$ to $g x_0$. This sequence gives rise to the exact sequence
$$1\to (Z_G(x_0))^\sigma \labelto{i} G^\sigma\labelto{j}
\mathcal{O}^\sigma \labelto{\delta} H^1(Z_G(x_0),\sigma) \labelto{i^*} H^1(G,\sigma),$$
see \cite[Proposition 36]{serre}: the map $i^\ast$ sends the class defined by $g\in Z^1(Z_G(x_0),\sigma)$  to its class in $H^1(G,\sigma)$;  moreover,
$\delta(g x_0)$ is the class of the cocycle $g^{-1}\sigma(g)$. The following is one of the
main theorems in Galois cohomology, see
\cite[\S 5.4, Corollary 1 to Proposition 36]{serre}. 

\begin{theorem}\label{thm:H1}
The map $\delta$ induces a bijection between the orbits of $G^\sigma$ in
$\mathcal{O}^\sigma$ and the set $\ker i^*$.   
\end{theorem}  

\begin{remark}\label{H1torus}
  It is known for the usual complex conjugation $\bar *$ that $H^1(\GL(n,\C))=1=H^1(\SL(n,\C))$ for all $n\geq 1$, see for example \cite[Proposition III.8.24 and Corollary III.8.26]{berhuy}, in particular, $H^1(\C^\times)=1$. Moreover, if $\bar *$ acts entry-wise on a complex matrix group $G=X\times Y$, then $H^1(X\times Y)\cong H^1(X)\times H^1(Y)$. Since a torus $T\leq G$ is a direct product of copies of $\C^\times$, we have that $H^1(T,\bar *)=1$.
\end{remark}

\subsection{Cartan subspaces}
Recall that $\h$ is the fixed Cartan subspace spanned by $\{u_1,\ldots,u_4\}$. Semisimple elements that lie in $\h$ are represented as a linear combination of these basis elements, however, most of our real orbit representatives lie in a Cartan subspace different to $\h$. To simplify the notation in our classification tables, we classify all Cartan subspaces and then represent our semisimple orbit representatives with respect to fixed bases of these spaces.

It follows from Galois cohomology that the real Cartan subspaces in $\g_1$ are, up to $\wG$-conjugacy, in bijection
with $H^1(N)$ where $N = N_{\wG}(\h)$, see \cite[Theorem 4.4.9]{trivectors}.
The group $N$ fits into an exact sequence
$$ 1 \to Z_{\wG}(\h) \to N \to W \to 1.$$
Since $Z_{\wG}(\h)$ and $W$ are finite groups of orders $32$ and $192$,  respectively, $N$ is a finite group of order $32\cdot 192 = 6144$.
Because we know $Z_{\wG}(\h)$ and $W$
(for the former see the first line of Table \ref{tabZ}), we can determine
$N$. Since $N$ is finite, a brute force calculation determines $H^1(N)$, and we obtain $|H^1(N)| = 7$. 

For a fixed $[n]\in H^1(N)$ define $\tau \colon \h\to \h$ by $\tau(u) = n\bar{u}$. Since $\tau$ is an anti-involution of $\h$, the $\R$-dimension of the fixed space $\h^\tau = \{ u \in \h : \tau(u) = u\}$
equals the $\C$-dimension of $\h$. Let $g\in \wG$ be such that $g^{-1} \bar{g} = n$; if $u\in \h^\tau$, then $u=n\bar u$, and the element $\overline{gu}  = gn n^{-1} u=gu$ is real. Thus, the  real span of all $gu$ with $u\in \h^\tau$ gives a real Cartan subspace. Iterating this procedure for all $[n]\in H^1(N)$ gives all real Cartan subspaces up to $G(\R)$-conjugacy; we fix the notation in the following definition. 

\begin{definition}\label{normalizer}
 There are seven classes in $H^1(N)$ corresponding to cocycles $n_1^*,\ldots, n_7^*\in Z^1(N)$; for each $i\in\{1,\ldots,7\}$ choose $g_i^*\in \wG$ such that
$(g_i^*)^{-1} \bar{g_i^*} = n_i^*$ and $\cc_i = g_i^*(\h^\tau)$.  Specifically, using the notation introduced in \eqref{eqmats}, we choose
\[\begin{array}{llll}
	g^*_1=(I,I,I,I),& g^*_2=(L,I,I,I),& g^*_3=(D(\eta^5),D(\eta^5),-D(\eta^3),-D(\eta^7)),& g^*_4=(M,I,I,M)\\
	g^*_5=(I,M,I,M),& g^*_6=(I,I,M,M),& g^*_7=(D(\eta^5),D(\eta^5),D(\eta^5),D(\eta^5)),
\end{array}\] where $\eta$ is a primitive $16$-th root of unity with $\eta^2=\zeta$. Moreover, we fix the following bases for the  seven Cartan subspaces $\cc_1,\ldots, \cc_7$ constructed above:
\begin{equation*}
  \begin{array}{rclrclrclrcl}
 \{u_1\!\!\!\!&=&\!\!\!\!\mye{0000}+\mye{1111},& u_2\!\!\!\!&=&\!\!\!\!\mye{0110}+\mye{1001},& u_3\!\!\!\!&=&\!\!\!\!\mye{0101}+\mye{1010},& u_4\!\!\!\!&=&\!\!\!\!\mye{0011}+\mye{1100}\}\\
\{v_1\!\!\!\!&=&\!\!\!\!\mye{0000}-\mye{1111},& v_2\!\!\!\!&=&\!\!\!\!\mye{0110}-\mye{1001},& v_3\!\!\!\!&=&\!\!\!\!\mye{0101}-\mye{1010},& v_4\!\!\!\!&=&\!\!\!\!\mye{0011}-\mye{1100}\}\\
\{w_1\!\!\!\!&=&\!\!\!\!\mye{0000}-\mye{1111},& w_2\!\!\!\!&=&\!\!\!\!\mye{0110}-\mye{1001},& w_3\!\!\!\!&=&\!\!\!\!\mye{0101}-\mye{1010},& w_4\!\!\!\!&=&\!\!\!\!\mye{0011}+\mye{1100}\}\\
\{x_1\!\!\!\!&=&\!\!\!\!\mye{0000}-\mye{1111},& x_2\!\!\!\!&=&\!\!\!\!\mye{0110}-\mye{1001},& x_3\!\!\!\!&=&\!\!\!\!\mye{0101}+\mye{1010},& x_4\!\!\!\!&=&\!\!\!\!\mye{0011}+\mye{1100}\}\\
\{y_1\!\!\!\!&=&\!\!\!\!\mye{0000}-\mye{1111},& y_2\!\!\!\!&=&\!\!\!\!\mye{0110}+\mye{1001},& y_3\!\!\!\!&=&\!\!\!\!\mye{0101}-\mye{1010},& y_4\!\!\!\!&=&\!\!\!\!\mye{0011}+\mye{1100}\}\\
\{z_1\!\!\!\!&=&\!\!\!\!\mye{0000}-\mye{1111},& z_2\!\!\!\!&=&\!\!\!\!\mye{0110}+\mye{1001},& z_3\!\!\!\!&=&\!\!\!\!\mye{0101}+\mye{1010},& z_4\!\!\!\!&=&\!\!\!\!\mye{0011}-\mye{1100}\}\\  
\{t_1\!\!\!\!&=&\!\!\!\!\mye{0000}-\mye{1111},& t_2\!\!\!\!&=&\!\!\!\!\mye{0110}+\mye{1001},& t_3\!\!\!\!&=&\!\!\!\!\mye{0101}+\mye{1010},& t_4\!\!\!\!&=&\!\!\!\!\mye{0011}+\mye{1100}\}.
\end{array}
\end{equation*}
\end{definition}

\section{Real semisimple elements}\label{secSS}
\noindent Throughout this section, we fix one of the subsystems $\Pi=\Pi_i$ of Theorem \ref{thmSE} and abbreviate $\CC=\h_{\Pi}^\circ$. We fix a complex $\wG$-orbit $\OO=\wG t$ for some nonzero $t\in\CC$. We now discuss the following problems related to the orbit~$\OO$:\pagebreak
\begin{ithm} 
\item[1)] Decide whether $\OO\cap \g_1(\R)$ is nonempty, that is, whether $\OO$ has a real point.
\item[2)] If $\OO$ has real points, how can we find one?
\item[3)] Determine representatives of the real $\wG(\R)$-orbits contained in $\OO$.
\end{ithm}

\medskip

We prove a number of results that help to decide these questions. These results as well as the proofs are similar to material found in \cite{trivectors2}.
However, the results in \cite{trivectors2} concern a specific
$\Z/3\Z$-grading of the Lie algebra of type $E_8$. Since here we consider a
different situation, we have included the new proofs. 

In the following, the centraliser and normaliser of $\CC$ in $\wG$ are denoted by
\begin{eqnarray*}Z_{\wG}(\CC) &=& \{ g\in \wG : gx=x \text{ for all }
  x\in \CC\}\\ N_{\wG}(\CC)  &=& \{ g\in \wG : gx\in \CC \text{ for all }
  x\in \CC\}.
\end{eqnarray*}
  
\begin{lemma}\label{lem:gp1p2}
  Let $t_1,t_2\in \CC$. If $gt_1=t_2$ for some $g\in \wG$, then  $g\in N_{\wG}(\CC)$.
\end{lemma}

\begin{proof}
Theorem \ref{thmSE}a) shows that $w(t_1)=t_2$ for some $w\in N_W(W_\Pi)$. If   $\hat w\in N_{\wG}(\h)$ is a preimage of $w$, then $g^{-1}\hat w
  \in Z_{\wG}(t_1)$. Theorem \ref{thmSE}b) shows that  
  $g^{-1}\hat w \in Z_{\wG}(x)$ for every $x\in\CC$, so $gx=\hat w x=w(x)$ for all $x\in\CC$. Since $w\in N_W(W_\Pi)$, we have $\Pi=w\Pi$ and hence $w(\CC)=\CC$, see \cite[Lemma 2.3 and Proposition 2.5]{complex}.
\end{proof}
 
Now we define a map $\varphi \colon N_{\wG}(\CC) \to \Gamma_\Pi$: If $g\in N_{\wG}(\CC)$, then $gq=x\in \CC$, hence $w(q)=x$ for some $w\in N_W(W_\Pi)$, and we define $\varphi(g)=wW_\Pi\in \Gamma_\Pi$. This is well-defined: if $w'\in N_W(W_\Pi)$ satisfies $w'(q)=x$, then
$w^{-1}w'\in W_\Pi$, so $w'W_\Pi = wW_\Pi$. 

\begin{lemma}\label{lem:Np}
The map $\varphi\colon N_{\wG}(\CC)\to \Gamma_\Pi$ is a surjective group homomorphism with kernel $Z_{\wG}(\CC)$. Moreover, if $g\in N_{\wG}(\CC)$ and $x\in \CC$, then $gx = \varphi(g)x$.
\end{lemma} 

\begin{proof}
  We start with a preliminary observation. If  $g\in N_{\wG}(\CC)$ and $w\in N_W(W_\Pi)$ such that $gq=w(q)$, then $gy=w(y)$ for all $y\in\CC$: indeed, if $\hat w\in N_{\wG}(\h)$ is a preimage
  of $w$, then $\hat w^{-1} g\in Z_{\wG}(q)$ and $\hat w^{-1} g\in Z_{\wG}(y)$ by Theorem \ref{thmSE}b). Now let $g_1,g_2\in N_{\wG}(\CC)$ and let $w_1,w_2\in N_W(W_\Pi)$ be
  such that each $w_i(q) = g_iq$. By the made observation, $g_1g_2 q = w_1(w_2(q))$; this implies that that $\varphi$ is a group  homomorphism. If  $wW_\Pi\in \Gamma_\Pi$ with preimage  $\hat w\in N_{\wG}(\h)$, then $\hat w \in N_{\wG}(\CC)$ and $\varphi(\hat w) = wW_\Pi$, which shows that $\varphi$ is surjective. If $g\in \ker(\varphi)$, then $gq=q$ and  the
  first part of the proof shows that $g\in Z_{\wG}(\CC)$.
\end{proof}  

By abuse of notation, we also  write \[\varphi\colon N_{\wG}(\CC)/Z_{\wG}(\CC)\to \Gamma_\Pi\] for the induced isomorphism. The next theorem provides solutions to Problems 1) and 2); it is similar
to \cite[Proposition 5.2.4]{trivectors2}.
Recall that we fixed $\Pi=\Pi_i$ and $\OO=\wG t$ with $t\in\CC$.

\begin{theorem}\label{thmRP}
  Write $H^1(\Gamma_\Pi) = \{ [\gamma_1],\ldots, [\gamma_s]\}$.
  Suppose that for each  $\gamma_i\in Z^1(\Gamma_\Pi)$ there is
  $n_i\in Z^1 (N_{\wG}(\CC))$ with $\varphi(n_i)=\gamma_i$.
  Then $\OO$ has a real point if and only if
  there exist $q'\in \OO\cap \CC$ and $i\in\{1,\ldots,s\}$ with
  $\bar q' = \gamma_i^{-1} q'$. If the latter holds,
  then $gq'$ is a real point of $\OO$, where $g\in \wG$ is such that
  $g^{-1} \bar g = n_i$. 
\end{theorem}

\begin{proof}
  The elements $n_i$ exist by Lemma \ref{lem:Np}.   If $\OO$ has a real point, say $p=gt$ for some $g\in \wG$, then $\overline{gt}=gt$, and so $\bar t = n^{-1} t$ for  $n=g^{-1} \bar g$; note that $n$ is a cocycle since $n\bar n = 1$. Because $t,\bar t\in \CC$, Lemma~\ref{lem:gp1p2} shows that 
$n\in N_{\wG}(\CC)$, so we can define $\gamma = \varphi(n)$. Since  $\gamma\in Z^1(\Gamma_\Pi)$, there is $i\in\{1,\ldots,s\}$ and  $\beta\in \Gamma_\Pi$ with
$\gamma = \beta^{-1} \gamma_i \bar\beta$. Now Lemma~\ref{lem:Np} shows that   $\bar t = \gamma^{-1} t= \bar\beta^{-1} \gamma_i^{-1} \beta t$, so if we set $q' = \beta t$, then  $q'\in \OO\cap \CC$ and $\bar q' = \gamma_i^{-1} q'$, as claimed. Conversely, let $q'\in \OO\cap \CC$ and $\gamma_i$ be such that
$\bar q' = \gamma_i^{-1} q'$. By hypothesis there is  $n_i\in
Z^1(N_{\wG}(\CC))$  with $\varphi(n_i)=\gamma_i$, hence
$\bar q' = n_i^{-1} q'$ by Lemma \ref{lem:Np}. Because $n_i\bar n_i=1$ in $\wG$ and $H^1(\wG)=1$, there is a $g\in \wG$ with $n_i = g^{-1} \bar g$. Now $p=gq'$ is a real point of $\OO$, as can be seen from $\bar p = \bar g \bar q' = gn_in_i^{-1} q' =p$.
\end{proof} 

\begin{remark}\label{example_notinH}
The real point $gq'$ mentioned in Theorem \ref{thmRP} might lie in a Cartan
subspace different to $\h$. The real points corresponding to the class
$[\gamma_1]=[1]$ can be chosen to lie in the Cartan subspace $\h$.
\end{remark}

\begin{remark}
One of the hypotheses of the theorem is that for each $\gamma_i$ there is a
cocycle $n_i\in Z^1 (N_{\wG}(\CC))$ such that $\varphi(n_i) = \gamma_i$. We cannot
prove this a priori, but for the cases that are relevant to the classification
given in this paper we have verified it.
\end{remark}

Galois cohomology also comes in handy for a solution to Problem 3): the next
theorem follows from Theorem~\ref{thm:H1}, taking into account that 
$H^1(\wG)=1$, see Remark~\ref{H1torus}.

\begin{theorem}\label{thm1to1}
Let $p\in \OO$ be a real representative. There is a 1-to-1 correspondence between the elements of  $H^1(Z_{\wG}(p))$ and the $\wG(\R)$-orbits of semisimple elements in $\OO$ that are $\wG$-conjugate to $p$:  the real orbit corresponding to $[z]\in H^1(Z_{\wG}(p))$ has representative $bp$ where $b\in \wG$ is chosen with $z=b^{-1}\overline{b}$.
\end{theorem}
\subsection{Classification approach}\label{approach}
Now we explain the classification procedure in full detail. For $i\in\{1,\ldots,10\}$ we compute some information related to Row $i$ of Table
\ref{tabW}: First, we construct the cohomology sets $H^1(\Gamma_{\Pi_i})$;
the $\Gamma_{\Pi_i}$ are finite groups with
trivial conjugation, so the cohomology classes coincide with the
conjugacy classes of elements of order dividing~2. These can be computed brute-force. If the complex orbit of a semisimple element $q$ has a real point, then Theorem \ref{thmRP} shows that there is some $\gamma_j$ and some $q'$ in
the orbit of $q$ such that $\overline{q}' = \gamma_j^{-1} q'$; now $gq'$ is a real point in the orbit of $q$, where $g$ is defined in Theorem \ref{thmRP}. We therefore proceed by looking at each $[\gamma_j]\in H^1(\Gamma_{\Pi_i})$ and determining all $q'$ such that $\overline{q'}=\gamma_j^{-1} q'$; this
will eventually determine the orbits of elements in $\h_{\Pi_i}^\circ$ that have real points,
along with a real point in each such orbit.
The next lemma clarifies that the elements determined for different $j$ yield real orbit representatives of different orbits.

\begin{lemma}
 With the above notation, if $j\ne k$, then the real orbit representatives obtained for $[\gamma_j]$ are not $\wG$-conjugate to those representatives obtained for  $[\gamma_k]$.
\end{lemma}
\begin{proof} 
  Suppose in the above process we construct $q_j,q_k\in \h_{\Pi_i^\circ}$  such that $\bar q_j = \gamma_j^{-1} q_j$ and 
  $\bar q_k = \gamma_k^{-1}q_k$. If $q_j$ and $q_k$ lie in the same $\wG$-orbit, then Theorem \ref{thmSE} shows that $q_k=\beta q_j$ for some $\beta\in \Gamma_{\Pi_i}$, that is,   $\bar q_k = \gamma_k^{-1} \beta q_j$, and solving for $q_j$ yields $q_j = \beta^{-1}\gamma_k \bar \beta \bar q_j = \beta^{-1} \gamma_k \bar\beta\gamma_j^{-1}q_j.$
    Since $q_j,q_k\in \h_{\Pi_i}^\circ$, we have $W_{q_j}=W_{q_k}$, see Section~\ref{secNot}, and so $\beta^{-1} \gamma_k \bar \beta \gamma_j^{-1} \in W_{\Pi_i}$.  Since $\Gamma_{\Pi_i}=N_W(W_{\Pi_i})/W_{\Pi_i}$, it follows that   $[\gamma_j] = [\gamma_k]$ in $H^1(\Gamma_{\Pi_i})$ .
\end{proof}

Our algorithm now proceeds as follows;  recall that each of our $\Gamma_{\Pi_i}$ is realised as a subgroup of $W$:
\begin{iprf}
\item[{\bf (A)}] For each component $\CC=\h_{\Pi_i}^\circ$ and each cohomology class $[\gamma_j]\in H^1(\Gamma_{\Pi_i})$ with $\gamma_j\in W$, we determine all $q'\in \CC$ that satisfy 
$\bar q' = \gamma_j^{-1} q'$ (using Table \ref{tabCoc}, this condition on $q'$ is easily obtained).
  We then determine $g_j\in \wG$ such that $n_j = g_j^{-1} \bar g_j$ (using Table \ref{tabKA}), and set $p= g_j q'$.  Theorem \ref{thmRP} shows that $p$ is a real representative in the complex orbit $\wG q'$ of $q'$. We note that in this approach we do not fix a complex orbit $\OO$ and look for $q'\in \OO\cap\CC$ as in Theorem \ref{thmRP}, but we first look for suitable $q'\in \CC$ and then consider the reduction up to $\Gamma_{\Pi_i}$-conjugacy.

\item[{\bf (B)}] Next, we determine the real orbits contained in $\wG p=\wG q'$. Using Theorem \ref{thm1to1}, we need to consider $Z=Z_{\wG}(p)$ and determine $H^1(Z)$ with respect to the usual complex conjugation $\bar \ast$. Note that $Z$ is one of the centralisers in Table~\ref{tabZ} and $Z=g_j Z_{\wG}(q)g_j^{-1}$.
It will turn out that in most cases we can decompose $Z=\widetilde{Z}\times H$ where $H$ is abelian of finite order; then $H^1(Z)= H^1(\widetilde{Z})\times H^1(H)$ by Remark \ref{H1torus}, which is useful for determining $H^1(Z)$. We will see that the component group $\widetilde{Z}/Z^\circ$ is of order at most 2; if it is nontrivial, then is generated by the class of $(J,J,J,J)$ where $J$ is as in \eqref{eqmats}. The connected component $Z^\circ$ is in most cases  parametrised by a torus or by $\SL(2,\C)$, and  we show that $H_1(Z^\circ)$ is trivial (see Remark \ref{H1torus}). To compute $H^1(Z)$ it remains to consider the cohomology classes of elements of the form $u = w(J,J,J,J)\in\widetilde{Z}\setminus Z^\circ$ where $w\in Z^\circ$. We determine conditions on $w$ such that $u$ is a 1-cocycle, and then solve the equivalence problem. All these calculations can be done by hand, but we have also verified them computationally with the system GAP.

\item[{\bf (C)}]Finally, given $\gamma_j,n_j,g_j$ and $H^1(Z)$, Theorem \ref{thm1to1} shows that a real orbit representative corresponding to $[z_k]\in H^1(Z)$ is $b_kp$ where $b_k\in \wG$ is chosen such that $b_k^{-1}\overline{b_k}=z_k$ (using the elements $\varepsilon(M)$ given in Table~\ref{tabKA}); specifically, we obtain the following real orbits representatives, recall that $p=g_jq'$:
\begin{align}\label{orbitreps}\{(\varepsilon(A),\varepsilon(B),\varepsilon(C),\varepsilon(D))(g_jq') :  [(A,B,C,D)]\in H^1(Z) \;\text{and}\; q'\in \h_{\Pi_i}^\circ \;\text{with}\; \overline{q'}=\gamma_j^{-1}q' \}.
\end{align}
\end{iprf}

\begin{remark}\label{remRewrite}
  Some real points and some real orbit representatives computed in (A) and (C) lie outside our fixed Cartan subspaces  $\cc_1,\ldots,\cc_7$ as defined in Definition \ref{normalizer}. If this is the case, then we rewrite these elements by following steps (A') and (C') after (A) and (C), respectively:
  \begin{iprf}
  \item[{\bf (A')}] 
If the real point $p$ in (A) is not in one of our fixed Cartan spaces $\cc_1,\ldots,\cc_7$, then 
we search for $g\in g_k^\ast N_{\wG}(\h)$ for some $k\in\{1,\dots,7\}$ such that $n_j = g ^{-1}\bar{g}$; recall the definition of $g_k^*$ and $n_k^*$ from  Definition \ref{normalizer}; we then replace $g_j$ by $g$. This is indeed possible: by construction, there is  $g_0\in N_{\wG}(\h)$ and $k\in\{1,\ldots,7\}$ with $g_0n_j\bar{g}_0^{-1} = n_k^\ast = (g_k^\ast) ^{-1}\bar{g^\ast_k}$, so $g=g_k^*g_0\in g_k^* N_{\wG}(\h)$ is a suitable element, and $p=gq'$  lies in $\cc_k$.   
\item[{\bf (C')}] If one of the  $b_kp$ in Step (C) is not in our Cartan spaces $\cc_1,\ldots,\cc_7$, then we proceed as follows (using Definition~\ref{normalizer}). Note that  $b_kp$ is $\wG(\R)$-conjugate to one of our Cartan spaces, say $b_k'b_kp\in \cc_j$ for some $b_k'\in \wG(\R)$ and $j\in\{1,\dots,7\}$. Since $b_k'$ is real, $(b_k'b_k)^{-1} \overline{b_k'b_k} = z_k$, and it follows that $b_0=b_k'b_k\in \wG$ satisfies $b_0^{-1} \bar{b}_0=z_k$ and $b_0p\in \cc_j$, as required. Since $(g_j^\ast)^{-1}b_0 p \in \h$ is $\wG$-conjugate to $q\in \h$, there exists $w\in N_{\wG}(\h)$  such that  $(g_j^\ast)^{-1}b_0p = wq= wg^{-1}p$, where $p=gq$  as in (A'); we always succeed finding $b_0$ in $g_j^\ast N_{\wG}(\h) g^{-1}$.
  \end{iprf}
To simplify the exposition, in our proof below we do not comment on the rewriting process (A') and (C'), but only describe the results for (A), (B), and (C).
\end{remark}


\subsection{Classification results}\label{secProofRSE}
The  procedure detailed in Section \ref{approach} leads to the following result; we prove it in this section.

\begin{theorem}\label{thmRSE}
   Up to $\wG(\R)$-conjugacy, the nonzero semisimple elements in $\g_1(\R)$ are 
   the elements in Tables \ref{tabreps10}--\ref{tabreps} in Appendix \ref{appSSTab}; see Definition \ref{normalizer} for the notation used in these tables.
\end{theorem} 

It follows from Theorem \ref{thmSE} and Theorem \ref{thmRSE} that there are many complex semisimple orbits that have no real points. For example, consider the Case $i=10$ and let us fix $q=\lambda u_1$ with $\lambda=a+\imath b$ and $a,b\neq 0$. According to the description given in Case 10 of the proof of Theorem \ref{thmRSE}, there exists a real point in that orbit if and only if there is $q'$ with $\overline{q'}=\gamma_{10}^{-1} q'$; this is equivalent to requiring that $\lambda$ is either a real or a purely imaginary number, which is a contradiction; thus the complex semisimple orbit determined by $q$ has no real points.

We now prove  Theorem \ref{thmRSE} by considering each case $i\in\{1,2,3,4,7,10\}$ individually; due to Remark \ref{remSym}, the classifications for the cases $i\in\{5,6,8,9\}$ can be deduced from those for $i\in\{4,7\}$. For each $i$ we  comment on the classification steps (A), (B), (C) as explained in the previous section; we do not comment on the rewriting process (A') and (C'). In Case $i$ below we write $Z=Z_{\wG}(p_i)$ as in Table \ref{tabZ}. Throughout, we use the notation introduced in \eqref{eqmats}.

\medskip  

\noindent {\bf Case $\pmb{i=1}$.} There are $7$ equivalence classes of cocycles in $\Gamma_{\Pi_1}=W$, with representatives $\gamma_1=I$, $\gamma_2=-I$, $\gamma_3=\diag(-1,-1,-1,1)$, $\gamma_4=\diag(-1,-1,1,1)$, $\gamma_5=\diag(-1,1,-1,1)$, $\gamma_6=\diag(-1,1,1,-1)$, and $\gamma_7=\diag(-1,1,1,1)$. The centraliser $Z$ is finite so the cohomology can be easily computed, and  $H^1(Z)$ has 12 classes with representatives
\[\begin{array}{llll}
   z_1=(I,I,I,I),& z_2=(I,I,-I,-I),& z_3=(I,-I,I,-I),& z_4=(I,-I,-I,I), \\
   z_5=(K,K,K,K),& z_6=(K,K,-K,-K),& z_7=(K,-K,K,-K),& z_8=(K,-K,-K,K),\\
   z_9=(L,L,L,L),& z_{10}=(L,L,-L,-L),& z_{11}=(L,-L,L,-L),& z_{12}=(L,-L,-L,L).
\end{array}\]
We follow the procedure outlined in Section \ref{approach} and consider the various $[c]\in H^1(\Gamma_{\Pi_1})$. We note that in all cases $(1,j)$ below we start with a complex semisimple element $\lambda_1u_1+\ldots+\lambda_4u_4$ as in Table \ref{tabW}, with $(\lambda_1,\ldots,\lambda_4)$ reduced up to $\Gamma_{\Pi_1}$-conjugacy, where $\Gamma_{\Pi_1}=W$ is the Weyl group.

\medskip

\begin{iprf}
\item[{\bf $\pmb{(i,j)=(1,1)}$:}] Let $[c]=[\gamma_1]=[\diag(1,1,1,1)]$. We first determine  all $q'\in \h_{\Pi_1}^\circ$ with $\overline{q'}=q'$; by Theorem \ref{thmSE}, these are the elements $q'=\lambda_1 u_1+\ldots+\lambda_4 u_4\in \h_{\Pi_1}$ with $\lambda_1,\ldots,\lambda_4\in \R\setminus\{0\}$ and $\lambda_1\notin \{\pm\lambda_2\pm\lambda_3\pm\lambda_4\}$.
Since $\gamma_1=I$, we can choose $n_1=g_1=(I,I,I,I)$, and obtain $p=g_1q'=q'$ as real point in the complex $\wG$-orbit of $q'$. 
Since the first cohomology group of  $Z_{\wG}(q')=Z_{\wG}(p)$ has 12 elements, it follows from Theorem~\ref{thm1to1} that $\wG q'$ splits into 12 real orbits with representatives determined as in  \eqref{orbitreps}.
For $z_1=(I,I,I,I)$ we have $b_1=(I,I,I,I)$, with real orbit representative  $b_1p=b_1q'=q'$. For $z_2=(I,I,-I,-I)$ we choose $b_2=(I,I,L,L)$, with  real orbit representative  $b_2p=b_2q'=-\lambda_1 u_1+\lambda_2 u_2+\lambda_3 u_3-\lambda_4 u_4\in \h_{\Pi_1}$. In the same way we obtain the orbit representatives for $z_3,\ldots,z_{12}$, which we summarise in Table \ref{tabreps} in the block $j=1$.

\item[{\bf $\pmb{(i,j)=(1,2)}$:}] Now let $[c]=[\gamma_2]=[\diag(-1,-1,-1,-1)]$. 
As before we determine all $q'\in \h_{\Pi_1}^\circ$ with $\overline{q'}=-q'$; these are the elements  $q'=\lambda_1 u_1+\ldots+\lambda_4 u_4\in \h_{\Pi_1}$ with $\lambda_1,\ldots,\lambda_4\in \imath\R\setminus\{0\}$ and $\lambda_1\notin \{\pm\lambda_2\pm\lambda_3\pm\lambda_4\}$. Table~\ref{tabCoc} shows that $\diag(-1,-1,-1,-1)\in W$ is induced by $n_2=(-I,I,I,I)$. An element $g_2$ with $g_2^{-1}\overline{g_2}=n_2$ is $g_2=(L,I,I,I)$. Now $p=g_2q'$ is a real point in the complex orbit of $q'$. 
Since $Z_{\wG}(p)=g_2 Z_{\wG}(q')g_2^{-1}$ is finite, we can directly compute $H^1(Z_{\wG}(p))$ and obtain 12 classes with the following representatives
\[\begin{array}{llll}
z_1=(I,I,I,I), & z_2=(-I,-I,I,I), & z_3=(I,-I,I,-I), & z_4=(-I,I,I,-I),\\
z_5=(K,K,K,-K), & z_6=(-K,-K,-K,K), & z_7=(K,-K,K,K),  & z_8=(-K,K,K,K),\\
z_9 = (-L,-L,-L,-L), & z_{10}=(L,L,-L,-L), & z_{11}= (-L,L,-L,L),  & z_{12}= (L,-L,-L,L).
\end{array}\]
Real orbits representatives are now determined as in Equation \eqref{orbitreps}, see Table \ref{tabreps}.
Note that here every $\lambda_i$ is purely imaginary, but each product $(\varepsilon(A),\varepsilon(B),\varepsilon(C),\varepsilon(D))g_2(\lambda_1 u_1+\ldots+\lambda_4 u_4)$ is a real point.

\medskip

We repeat the same procedure for $\gamma_3,\cdots, \gamma_7$; for each case we only summarise the important data, and we refer to Table \ref{tabreps} for the list of real orbits representatives.

\medskip

\item[{\bf $\pmb{(i,j)=(1,3)}$:}] If $[c]=[\gamma_3]$, then  $q'\in \h_{\Pi_1}^\circ$
satisfies $\overline{q'}=\gamma_3^{-1}q'$ if and only if  $q'=\lambda_1 u_1+\ldots+\lambda_4 u_4$ with $\imath\lambda_1,\imath\lambda_2,\imath \lambda_3,\lambda_4\in \R\setminus\{0\}$ and $\lambda_1\notin \{\pm\lambda_2\pm\lambda_3\pm\lambda_4\}$; we have  $n_3=(M,M,-N,N)$ and
$g_3=(D(\eta^5),D(\eta^5),-D(\eta^3),-D(\eta^7))$;  the first cohomology of $Z_{\wG}(p)=g_3Z_{\wG}(q')g_3^{-1}$ consists of 
four classes defined by the representatives
\[\begin{array}{llll}
z_1=(-I,-I,-I,-I),& z_2=(-I,-I,I,I), & z_3=(-I,I,-I,I), & z_4=(-I,I,I,-I).
\end{array}\]

\item[{\bf $\pmb{(i,j)=(1,4)}$:}] If $[c]=[\gamma_4]$, then $q'\in \h_{\Pi_1}^\circ$ satisfies $\overline{q'}=\gamma_4^{-1}q'$ if and only if $q'=\lambda_1 u_1+\ldots+\lambda_4 u_4$ with $\imath\lambda_1,\imath\lambda_2,\lambda_3,\lambda_4\in \R\setminus\{0\}$ and $\lambda_1\notin \{\pm\lambda_2\pm\lambda_3\pm\lambda_4\}$. 
We have $n_4=(L,I,I,L)$, 
$g_4=(M,I,I,M),$ and the first cohomology of $Z_{\wG}(p)=g_4Z_{\wG}(q')g_4^{-1}$ consists of four classes defined by the representatives
\[\begin{array}{llll}
z_1=(I,I,I,I),& z_2=(I,I,-I,-I), & z_3=(L,L,L,L), & z_4=(L,L,-L,-L).
\end{array}\]

\item[{\bf $\pmb{(i,j)=(1,5)}$:}] If $[c]=[\gamma_5]$, then  $q'\in \h_{\Pi_1}^\circ$ satisfies $\overline{q'}=\gamma_5^{-1}q'$ if and only if $q'=\lambda_1 u_1+\ldots+\lambda_4 u_4$ with $\imath\lambda_1,\lambda_2,\imath\lambda_3,\lambda_4\in \R\setminus\{0\}$ and $\lambda_1\notin \{\pm\lambda_2\pm\lambda_3\pm\lambda_4\}$. 
We get $n_5=(I,L,I,L)$, 
$g_5=(I,M,I,M),$ and the first cohomology of  $Z_{\wG}(p)=g_5Z_{\wG}(q')g_5^{-1}$ consists of four classes defined by the representatives
\[\begin{array}{llll}
z_1=(I,I,I,I),& z_2=(I,I,-I,-I), & z_3=(L,L,L,L), & z_4=(L,L,-L,-L).
\end{array}\]

\item[{\bf $\pmb{(i,j)=(1,6)}$:}] If $[c]=[\gamma_6]$, then $q'\in \h_{\Pi_1}^\circ$ satisfies $\overline{q'}=\gamma_6^{-1}q'$ if and only if $q'=\lambda_1 u_1+\ldots+\lambda_4 u_4$ with $\imath\lambda_1,\lambda_2,\lambda_3,\imath\lambda_4\in \R\setminus\{0\}$ and $\lambda_1\notin \{\pm\lambda_2\pm\lambda_3\pm\lambda_4\}$. 
Now $n_6=(I,I,L,L)$, 
$g_6=(I,I,M,M)$, and the first cohomology of  $Z_{\wG}(p)=g_6Z_{\wG}(q')g_6^{-1}$ consists of four classes defined by the representatives
\[\begin{array}{llll}
z_1=(I,I,I,I),& z_2=(I,-I,I,-I), & z_3=(L,L,L,L), & z_4=(L,-L,L,-L).
\end{array}\]

\item[{\bf $\pmb{(i,j)=(1,7)}$:}] If $[c]=[\gamma_7]$, then  $q'\in \h_{\Pi_1}^\circ$ satisfies $\overline{q'}=\gamma_7^{-1}q'$ if and only if $q'=\lambda_1 u_1+\ldots+\lambda_4 u_4$ with $\imath\lambda_1,\lambda_2,\lambda_3,\lambda_4\in \R\setminus\{0\}$ and $\lambda_1\notin \{\pm\lambda_2\pm\lambda_3\pm\lambda_4\}$. 
We have $n_7=(M,M,M,M)$ and $g_7=(D(\eta^{5}),D(\eta^{5}),D(\eta^{5}),D(\eta^{5}))$; the first cohomology of $Z_{\wG}(p)=g_7Z_{\wG}(q')g_7^{-1}$ consists of four classes defined by the representatives
\[\begin{array}{llll}
z_1=(I,I,I,I),& z_2=(I,I,-I,-I), & z_3=(I,-I,I,-I), & z_4=(I,-I,-I,I).
\end{array}\]
\end{iprf}

\medskip

\noindent {\bf Case $\pmb{i=2}$.}
Since $\Gamma_{\Pi_2}$ is elementary abelian of order $8$, there are $8$ equivalence classes of cocycles in $H^1(\Gamma_{\Pi_2})$, with representatives
\[\text{\scalebox{0.97}{$\begin{array}{llll}
	\gamma_1=\diag(1,1,1,1),& \gamma_2=\diag(1,-1,-1,1),& \gamma_3=\diag(1,-1,1,-1),& \gamma_4=\diag(-1,-1,1,1), \\
	\gamma_5=\diag(1,1,-1,-1),& \gamma_6=\diag(-1,1,-1,1),& \gamma_7=\diag(-1,1,1,-1),& \gamma_8=\diag(-1,-1,-1,-1).
\end{array}$}}\]
The centraliser decomposes as $Z=\widetilde{Z}\times H$ where $H$ is abelian of order $4$,  generated by $(-I,-I,I,I)$ and $(-I,I,-I,I)$. Furthermore
$\widetilde{Z}/Z^\circ$ has size $2$, generated by the class of $(J,J,J,J)$, and
$Z^\circ$ is a $1$-dimensional torus consisting of elements
$T_1(a) = (D(a^{-1}), D(a^{-1}), D(a), D(a))$ with $a\in \C\setminus\{0\}.$
The main difference to the Case $i=1$ is that here $Z$ is not finite; we include some details to explain our computations. First,   a direct calculation shows that  $H^1(H)$
consists of four classes defined by the representatives
\[\begin{array}{llll}
	z_1=(I,I,I,I),& z_2=(-I,-I,I,I),& z_3=(-I,I,-I,I),& z_4=(I,-I,-I,I). 
\end{array}\]	
Next, we look at $H^1(\widetilde{Z})$. Since $Z^\circ$ is a $1$-dimensional torus, a direct computation (together with Remark \ref{H1torus}) shows that $H^1(Z^\circ)$ is trivial. It remains to consider the cohomology classes of elements in $\widetilde{Z}/Z^\circ$. Let  $u = wj^*$ where $w\in Z^\circ$ and  \[j^*=(J,J,J,J).\] A short calculation shows that $u$ is a 1-cocycle if and only if there is $a\in\R\setminus\{0\}$ with \[u=\SmallMatrix{0&-\imath a^{-1} \\ -\imath a &0} \times \SmallMatrix{0&-\imath a^{-1} \\ -\imath a &0} \times \SmallMatrix{0&\imath a \\ \imath a^{-1} &0} \times \SmallMatrix{0&\imath a \\ \imath a^{-1} &0}.\]
Moreover, every such $u$ is equivalent to $k=(-K,-K,K,K)$, thus  
$H^1(\tilde Z)=\{[1],[k]\}$. 
Indeed, we can verify (by a short calculation or with the help of GAP) that two $1$-cocycles $u$, $u'$ satisfying $u'=gu\overline g^{-1}$ where $g=(D(c^{-1}), D(c^{-1}), D(c), D(c))\in Z^\circ$ for some $c\in \C^\times$ if and only if $a'=a|c|^2$, thus
we can assume $a=\pm 1$. Now  $u'=gj^*u(\overline{gj^*})^{-1}$ for some $g\in Z^\circ$ if and only if $aa'=(c/|c|)^2$, then the 1-cocycles corresponding to $a=1$ and $a'=-1$ are equivalent.
Since $H^1(Z)=H^1(\widetilde{Z})\times H^1(H)$, representatives of the classes in $H^1(Z_{\wG}(p))$ are $z_1,\ldots,z_4$ and  
\[\begin{array}{ll}
	z_5=z_1k=(-K,-K,K,K),& z_6=z_2k=(K,K,K,K),\\ z_7=z_3k=(K,-K,-K,K),& z_8=z_4k=(-K,K,-K,K).
\end{array}\]
With the same approach we obtain $H^1(g_jZg_j^{-1})$ for all $\gamma_j$ for $j=1,\dots,8$. Representatives of the real orbits are listed in Table \ref{tabreps2}; below we only list the important data.

\begin{iprf}
	\item[{\bf $\pmb{(i,j)=(2,1)}$:}] If $[c]=[\gamma_1]$, then   $q'\in \h_{\Pi_2}$  satisfies  $\overline{q'}=q'$ if and only if  $\lambda_1,\lambda_2,\lambda_3\in \R$ are nonzero and $\lambda_1\notin \{\pm\lambda_2\pm\lambda_3\}$. We have $n_1=g_1=(I,I,I,I)$, and $H^1(Z)$ was computed above.

	\item[{\bf $\pmb{(i,j)=(2,2)}$:}] If $[c]=[\gamma_2]$, then $q'\in \h_{\Pi_2}$ satisfies $\overline{q'}=\gamma_2^{-1}q'$ if and only if  $\lambda_1,\imath\lambda_2,\imath\lambda_3\in\R\setminus\{0\}$ and $\lambda_1\notin \{\pm\lambda_2\pm\lambda_3\}$. We have $n_2=(I,I,L,-L)$ and $g_2=(I,I,M,D(\zeta))$, with real point $p=g_2q'$. A direct calculation shows that $u=g_2 wj^* g_2^{-1}$ with $w=(D(a)^{-1}, D(a)^{-1}, D(a), D(a))\in Z^\circ$ for some $a\in \C^\times$ is a $1$-cocycle if and only if $a=-\overline a$ and $a=\overline a$, which is a contradiction. This shows that  there is no $1$-cocycle with representative $u=g_2 wj^* g_2^{-1}$, so  $H^1(Z)=H^1(H)$.

\medskip

\noindent For $\gamma_3,\ldots,\gamma_7$ we also deduce that $H^1(Z)=H^1(H)$; the case of $\gamma_8$ is similar to $\gamma_1$.

            \medskip
	
	\item[{\bf $\pmb{(i,j)=(2,3)}$:}] If $[c]=[\gamma_3]$, then the condition on $q'=\lambda_1 u_1+\lambda_2u_2 +\lambda_3 u_3\in \h_{\Pi_2}$ is $\lambda_1,\imath\lambda_2,\lambda_3\in \R\setminus\{0\}$, and $\lambda_1\notin \{\pm\lambda_2\pm\lambda_3\}$. In this case $n_3=(I,L,I,-L)$ and $g_3=(I,M,I,D(\zeta))$. 
	
	\item[{\bf $\pmb{(i,j)=(2,4)}$:}] If $[c]=[\gamma_4]$, then the condition on $q'=\lambda_1 u_1+\lambda_2u_2 +\lambda_3 u_3\in \h_{\Pi_2}$ is $\imath\lambda_1,\imath\lambda_2,\lambda_3\in \R\setminus\{0\}$ and $\lambda_1\notin \{\pm\lambda_2\pm\lambda_3\}$. We have $n_4=(L,I,I,L)$ and $g_4=(M,I,I,M)$.

	\item[{\bf $\pmb{(i,j)=(2,5)}$:}] If $[c]=[\gamma_5]$, then the condition on  $q'=\lambda_1 u_1+\lambda_2u_2 +\lambda_3 u_3\in \h_{\Pi_2}$ is  $\lambda_1,\lambda_2,\imath\lambda_3\in \R\setminus\{0\}$ and $\lambda_1\notin \{\pm\lambda_2\pm\lambda_3\}$; we have $n_5=(L,I,I,-L)$ and $g_5=(M,I,I,N)$.

	\item[{\bf $\pmb{(i,j)=(2,6)}$:}] If $[c]=[\gamma_6]$, then the condition on $q'=\lambda_1 u_1+\lambda_2u_2 +\lambda_3 u_3\in \h_{\Pi_2}$ is $\imath\lambda_1,\lambda_2,\imath\lambda_3\in \R\setminus\{0\}$ and $\lambda_1\notin \{\pm\lambda_2\pm\lambda_3\}$; we have  $n_6=(I,L,I,L)$ and $g_6=(I,M,I,M)$.

	\item[{\bf $\pmb{(i,j)=(2,7)}$:}] If $[c]=[\gamma_7]$, then the condition on $q'=\lambda_1 u_1+\lambda_2u_2 +\lambda_3 u_3\in \h_{\Pi_2}$ is $\imath\lambda_1,\lambda_2,\lambda_3\in \R\setminus\{0\}$ and $\lambda_1\notin \{\pm\lambda_2\pm\lambda_3\}$; we have  $n_7=(I,I,L,L)$ and $g_7=(I,I,M,M)$.

	\item[{\bf $\pmb{(i,j)=(2,8)}$:}] If $[c]=[\gamma_8]$, then the condition on $q'=\lambda_1 u_1+\lambda_2u_2 +\lambda_3 u_3\in \h_{\Pi_2}$ is  $\imath\lambda_1,\imath\lambda_2,\imath\lambda_3\in\R\setminus\{0\}$ and $\lambda_1\notin \{\pm\lambda_2\pm\lambda_3\}$; we have $n_8=(-I,I,I,I)$ and $g_8= (L,I,I,I)$.	A short calculation shows that  $u = wj^*$ with $w=(D(a)^{-1}, D(a)^{-1}, D(a), D(a))\in Z^\circ$ is a $1$-cocycle if and only if $a$ is purely imaginary. Moreover, every such $u$ is equivalent to $k=(K,-K,K,K)$, thus  
	$|H^1(\tilde Z)|=2$.
	In conclusion, $H^1(Z)$ has $8$ classes with representatives $z_1,\ldots,z_4$ and 
	\[\begin{array}{ll}
	z_5=z_1k=(K,-K,K,K),& z_6=z_2k=(-K,K,K,K),\\ z_7=z_3k=(-K,-K,-K,K),& z_8=z_4k=(K,K,-K,K).
\end{array}\]	
\end{iprf}

\medskip

\noindent {\bf Case $\pmb{i=3}$.} Since $H^1(\Gamma_{\Pi_3})=\{ [I], [-I]\}$, there are $2$ equivalence classes of cocycles with representatives $\gamma_1=\diag(1,1,1,1)$ and $\gamma_2=\diag(-1,-1,-1,-1)$. We  decompose $Z=Z^\circ\times H$, where $H$ is the same as in the case $i=2$ and
\[Z^\circ =\{(A^\#,A^\#,A,A) : A\in\SL(2,\mathbb{C})\}.\] A short calculation and  Remark \ref{H1torus} show that $H^1(Z^\circ)=1$, so  $H^1( Z )= H^1( H)=\{[z_1],[z_2],[z_3],[z_4]\}$ as determined for  $i=2$. 
Representatives of the real orbits are listed in Table \ref{tabreps3}; below we give some details.

\medskip

\begin{iprf}
	\item[{\bf $\pmb{(i,j)=(3,1)}$:}] If $[c]=[\gamma_1]$, then $q'=\lambda_1 (u_1-u_2) + \lambda_2 (u_1-u_3)\in \h_{\Pi_3}^\circ$ satisfies $\overline{q'}=q'$ if and only if $\lambda_1,\lambda_2\in \R$ and $\lambda_1\lambda_2(\lambda_1+\lambda_2)\neq 0$; we have $n_1=g_1=(I,I,I,I)$.	
	
	\item[{\bf $\pmb{(i,j)=(3,2)}$:}] If $[c]=[\gamma_2]$, then the condition on $q'=\lambda_1 (u_1-u_2) + \lambda_2 (u_1-u_3)\in \h_{\Pi_3}$ is  $\lambda_1,\lambda_2\in \imath\R$ and $\lambda_1\lambda_2(\lambda_1+\lambda_2)\neq 0$; we have $n_2=(-I,I,I,I)$ and $g_2=(L,I,I,I)$. A short calculation shows that $H^1(g_2Z^\circ g_2^{-1})$ is in bijection with $H^1(\SL(2,\C))=1$; in conclusion,  $H^1(Z_{\wG}(p))=H^1(H)$. 
\end{iprf}		

\medskip

\noindent {\bf Case $\pmb{i=4}$.} This case is similar to  $i=2$. Here $H^1 (\Gamma_{\Pi_4})$ consists of the classes of 
\[\gamma_1=\diag(1,1,1,1),\quad \gamma_2=\diag(-1,1,1,1),\quad \gamma_3=\diag(-1,1,1,-1),\quad \gamma_4=\SmallMatrix{ 0&0&0&-1\\0&0&1&0\\0&1&0&0\\-1&0&0&0}.\]
We have $Z= \widetilde{Z}\times H$ where $H$ is abelian of order
2 and generated by $(-I,I,-I,I)$. 
Furthermore
$\widetilde{Z}/Z^\circ$ has order $2$, generated by the class of $(J,J,J,J)$, and  $Z^\circ$ is a $2$-dimensional torus consisting of elements $T_2(a,b)=\left\{ (D(a)^{-1}, D(a), D(b)^{-1}, D(b)) \;:\; a,b\in \C^\times\right \}$. First,
$H^1(H)$  
consists of 2 classes defined by the representatives
\[\begin{array}{llll}
	z_1=(I,I,I,I),& z_2=(-I,I,-I,I). 
\end{array}\]	
Since $Z^\circ$ is parametrised by a $2$-dimensional torus, a direct computation and Remark \ref{H1torus}) shows that $H^1(Z^\circ)$ is trivial. 
Now consider the cohomology classes of elements in $\widetilde{Z}/Z^\circ$, that is,  $u = wj^\ast$ where $w\in T_2(a,b)$. Computations similar to the ones in Case  $i=2$ show that $u$ is a $1$ cocycle if and only if $a,b\in\imath\R$. Moreover, such
a $1$-cocycle $u$ is equivalent to either $(-K,K,-K,K)$ or $(-K,K,K,-K)$, thus $|H^1 (\widetilde{Z})|=3$. In conclusion, $H^1(Z)$ has $6$ classes with representatives $z_1,z_2$ and  
\[z_3=(-K,K,-K,K),\quad z_4=(-K,K,K,-K),\quad  z_5=(K,K,K,K),\quad z_6=(K,K,-K,-K).\] Real orbit representatives are listed in Table \ref{tabreps4}; we summarise the important data below.

\begin{iprf}
\item[{\bf $\pmb{(i,j)=(4,1)}$:}] If $[c]=[\gamma_1]$, then  $q'=\lambda_1 u_1+\lambda_4u_4\in \h_{\Pi_4}$  satisfies  $\overline{q'}=q'$ if and only if  $\lambda_1,\lambda_4\in\R$ and $\lambda_1\lambda_4(\lambda_1+\lambda_4)(\lambda_1-\lambda_4)\neq0$. We have  $n_1=g_1=(I,I,I,I)$ and $p=g_1q'=q'$
  
	\item[{\bf $\pmb{(i,j)=(4,2)}$:}] If  $[c]=[\gamma_2]$, then the condition on  $q'=\lambda_1 u_1+\lambda_4u_4\in \h_{\Pi_4}$  is  $\imath\lambda_1,\lambda_4\in \R\setminus\{0\}$; we have  $n_2=(M,M,M,M)$ and $g_2=(D(\eta^5),D(\eta^5),D(\eta^5),D(\eta^5))$.	Let $u=g_2 wj^* g_2^{-1}$ with $w\in Z^\circ$. As in Case $i=2$, a short calculation shows that there is no $1$-cocycle with representative $u$, so $H^1(Z)=H^1(H)$ has size~2.

	\item[{\bf $\pmb{(i,j)=(4,3)}$:}] If $[c]=[\gamma_3]$, then the condition on $q'$ is $q'=\lambda_1 u_1+\lambda_4u_4\in \h_{\Pi_4}$ with $\lambda_1,\lambda_4\in\imath\R$ and $\lambda_1\lambda_4(\lambda_1+\lambda_4)(\lambda_1-\lambda_4)\neq0$; we have  $n_3=(I,I,L,L)$ and $g_3=(I,I,M,M)$. A direct calculation shows that a $1$-cocycle  $u=g_3 wj^* g_3^{-1}$ with $w \in Z^\circ$ is equivalent to $(-K,K,-K,-K)$ or $(-K,K,K,K)$. In conclusion $H^1(g_3Zg_3^{-1})$ consists of  6 classes with representatives $z_1,z_2$ and 
	\[\begin{array}{llll} 
		z_3=(-K,K,-K,-K),& z_4=(K,K,K,-K), &z_5=(-K,K,K,K),& z_6=(K,K,-K,K).
	\end{array}\]
	
	\item[{\bf $\pmb{(i,j)=(4,4)}$:}] Let $[c]=[\gamma_4]$, then the condition on $q'=\lambda_1 u_1+\lambda_4u_4\in \h_{\Pi_4}$ is $\overline\lambda_1=-\lambda_4$; we have  $\lambda_1-\lambda_4\in \R$ and $\lambda_1+\lambda_4\in \imath\R$. We have $g_4=(M,M,F,LF)$, and a direct calculation (assisted by GAP) shows that  $H^1(g_4Zg_4^{-1})$ has size $4$ with representatives $z_1,z_2$ and $z_3=(-K,-K,-L,-L)$ and $z_4=(K,-K,L,-L)$.
\end{iprf}

\medskip

\noindent {\bf Case $\pmb{i=7}$.}  Since $H^1 (\Gamma_{\Pi_7}) = \{ [I], [-I]\}$, we have the same $\gamma_1,\gamma_2$ as in the Case $i=3$. We decompose  $Z=Z^\circ\times H$ where $H$ has order 2, generated by $(-I,I,-I,I)$, and   $Z^\circ$ is parametrised by $\SL(2,\C)\times\SL(2,\C)$. As before,  $H^1(Z^\circ)=1$, so $H^1( Z) = H^1( H)$  consists of the classes of $z_1=(I,I,I,I)$ and  $z_2=(-I,I,-I,I)$. Table \ref{tabreps7} lists the real orbit representatives.

\begin{iprf}
	\item[{\bf $\pmb{(i,j)=(7,1)}$:}] If $[c]=[\gamma_1]$, then  $q'\in \h_{\Pi_7}^\circ$ satisfies $\overline{q'}=q'$ if and only if $q'=\lambda (u_1-u_4)$ with $\lambda\in \R$ and $\lambda\neq 0$; we have $n_1=g_1=(I,I,I,I)$.	
	\item[{\bf $\pmb{(i,j)=(7,2)}$:}] If $[c]=[\gamma_2]$, then the condition on $q'=\lambda (u_1-u_4)$ is  $\lambda\in \imath\R\setminus\{0\}$; we have $n_2=(-I,I,I,I)$ and $g_2=(L,I,I,I)$. A direct computation shows that  $H^1(g_2Z^\circ g_2^{-1})=1$, so  $H^1(g_2Zg_2^{-1})=H^1(H)$ determines 2 real orbits.
\end{iprf}

\medskip

\noindent {\bf Case $\pmb{i=10}$.} Here we have  $H^1 (\Gamma_{\Pi_{10}}) = \{ [1], [-1]\}$ and $Z= \widetilde{Z}$, where $Z^\circ$ is a 3-dimensional
torus consisting of elements $T_3(a,b,c)=\left\{ (D(abc)^{-1}, D(a), D(b), D(c)) \;:\; a,b,c\in \C^\times \right\}$ and $\widetilde{Z}/Z^\circ$ is of order 2, generated by the class of $(J,J,J,J)$. As before, $H^1(Z^\circ)=1$, and elements of the form $u = wj^*$ with $w=T_3(a,b,c)\in Z^\circ$ are $1$-cocycles if and only if $a,b,c\in\imath\R$. Moreover, every such 1-cocycle $u$ is equivalent to  $(K,K,K,K)$, $(-K,K,K,-K)$, $(-K,K,-K,K),$ or $(K,K,-K,-K)$, thus  $H^1(Z)$ has $5$ classes with representatives   
\[	z_1=(I,I,I,I),z_2=(K,K,K,K), z_3=(-K,K,K,-K), z_4=(-K,K,-K, K), z_5=(K,K,-K,-K).\]
Table \ref{tabreps10} lists the real orbits representatives.
\begin{iprf}
	\item[{\bf $\pmb{(i,j)=(10,1)}$:}] If $[c]=[\gamma_1]$, then the condition on  $q'=\lambda u_1$  is  $\lambda\in\R\setminus\{0\}$; we have $n_1=g_1=(I,I,I,I)$.

	\item[{\bf $\pmb{(i,j)=(10,2)}$:}] If $[c]=[\gamma_2]$, then the condition on  $q'=\lambda u_1\in \h_{\Pi_{10}}$ is  $\lambda\in \imath\R\setminus\{0\}$; we have $n_2=(-I,I,I,I)$ and $g_2=(L,I,I,I)$. Since $L$ commutes with diagonal matrices, $H^1(g_2Z^\circ g_2^{-1})=H^1(Z^\circ)$. On the other hand, every $1$-cocycle  $u = wj^*$ with $w\in Z^\circ$ is 
	equivalent to  $(K,K,K,K)$, $(-K,K,K,-K)$, $(-K,K,-K,K),$ or $(K,K,-K,-K)$; thus  $H^1(Z)$ has $5$ classes with representatives   
	\[z_1=(I,I,I,I), z_2=(-K,K,K,K), z_3=(K,K,K,-K),	z_4=(K,K,-K, K), z_5=(-K,K,-K,-K).\]
\end{iprf}

\medskip


\section{Real elements of mixed type}\label{secMixed}

\noindent An element of mixed type is of the form $p+e$ where $p$ is
semisimple, $e$ is nilpotent and $[p,e]=0$. From the uniqueness of the Jordan
decomposition it follows that two elements $p+e$ and $p'+e'$  of mixed type are
$\wG$-conjugate if and only if there is a $g\in \wG$ with $g p =p'$ and
$g e=e'$. So if we want to classify orbits of mixed type then we may assume
that the semisimple part is one of a fixed set of orbit representatives
of semisimple elements. For $\mathbb{K}\in \{\R,\C\}$  we define
$$\MM(\mathbb{K}) = \{ p+e \colon p+e\text{ is of mixed type in }
\g_1(\mathbb{K})\}.$$
We know the $\wG$-orbits in $\MM(\C)$ and we want to classify the
$\wG(\R)$-orbits in $\MM(\R)$. Applying the general Galois cohomology approach
will lead to additional challenges. To avoid these, instead of working with
$\mathcal{M}(\mathbb{K})$, we will consider 4-tuples $(p,h,e,f)$, where
$p+e$ is a mixed element and $(h,e,f)$ is a suitable $\ssl_2$-triple.
This has the
advantage that the stabiliser of such a 4-tuple in $\wG$ is smaller than
the stabiliser of $p+e$, and secondly it is reductive. This makes it easier
to compute the Galois cohomology sets. We now
explain the details of this approach.

Let $p\in \g_1(\K)$ be a semisimple element. The nilpotent parts of
mixed elements with semisimple part $p$ lie in the subalgebra
$$\a=\z_{\g(\K)}(p) = \{ x\in \g(\K) : [x,p]=0\}.$$
This subalgebra inherits the grading from $\g$, that is, if we set $\a_i = \a
\cap \g_i(\K)$ then $\a=\a_0\oplus \a_1$. Moreover, the possible nilpotent parts
of mixed elements with semisimple part $p$ correspond, up to $\wG(\K)$-conjugacy,
to the $Z_{\wG(\K)}(p)$-orbits of nilpotent elements in $\a_1$.
The latter are classified using {\em homogeneous $\ssl_2$-triples},
which are triples $(h,e,f)$ with $h\in \a_0$ and $e,f\in\a_1$ such that
$$[h,e]=2e,\quad [h,f]=-2f,\quad [e,f]=h.$$
By the Jacobson-Morozov Theorem (see
\cite[Proposition 4.2.1]{trivectors}), every nonzero nilpotent $e\in \a_1$ lies in some
homogeneous $\ssl_2$-triple. Moreover, if $e,e'\in \a_1$ lie in homogeneous
$\ssl_2$-triples $(h,e,f)$ and $(h',e',f')$, then $e$ and $e'$ are $Z_{\wG(\K)}(p)$-conjugate
if and only if the triples $(h,e,f)$ and $(h',e',f')$ are $Z_{\wG(\K)}(p)$-conjugate.
For this reason we consider the set of quadruples
$$\QQ(\mathbb{K}) = \{ (p,h,e,f) \colon p\in \g_1(\mathbb{K}) \text{ is
  semisimple and } (h,e,f) \text{ is a homogeneous $\ssl_2$-triple
  in } \z_{\g(\mathbb{K})}(p)\}.$$
We have just shown that there is a surjective map $\QQ(\K)\to \MM(\K)$,
$(p,h,e,f) \mapsto p+e$. By the next lemma, this map defines a bijection between the
$\wG(\K)$-orbits in the two sets.

\begin{lemma} Let $\mathbb{K}\in\{\R,\mathbb{C}\}$. Let  $p,\hat{p}\in \g_1(\mathbb{K})$ be semisimple and let  $(h,e,f)$ and $(\hat{h},\hat{e},\hat{f})$ be homogeneous $\ssl_2$-triples in
	$\z_{\g(\mathbb{K})}(p)$  and  $\z_{\g(\mathbb{K})}(\hat{p})$, respectively. Then $p+e$ and $\hat{p}+\hat{e}$ are $\wG(\mathbb{K})$-conjugate if and
  only if $(p,h,e,f)$ and $(\hat{p},\hat{h},\hat{e},\hat{f})$ are
  $\wG(\K)$-conjugate.
\end{lemma}

\begin{proof}
	Only one direction needs proof. If $g(p+e) =\hat{p}+\hat{e}$ with $g\in \wG(\mathbb{K})$, then $gp = \hat{p}$ and $ge=\hat{e}$ by uniqueness of the Jordan
	decomposition. Now  $(gh,\hat{e},gf)$ is a
	homogeneous $\ssl_2$-triple in $\z_{\g(\mathbb{K})}(\hat{p})$. By \cite[Proposition~4.2.1]{trivectors}, there is  $g_1\in Z_{\wG(\K)}(\hat{p})$
	such that $g_1(gh,\hat{e},gf) = (\hat{h},\hat{e},\hat{f})$, so  $(g_1g)(p,h,e,f) =(\hat{p},\hat{h},\hat{e},\hat{f})$.
\end{proof}
 Our approach now is to  classify the $\wG(\R)$-orbits in $\QQ(\R)$; the main tool for this is the following theorem which follows directly from 
Theorem \ref{thm:H1} and the fact that $\wG$ has trivial cohomology.

\begin{theorem}\label{thmMGal}
Let $(p',h',e',f')$ be a real point in the $\wG$-orbit of $(p,h,e,f)\in\QQ(\mathbb{C})$. There is a 1-to-1 correspondence between  $H^1(Z_{\wG}(p',h',e',f'))$ and the $\wG(\R)$-orbits in $\wG(p,h,e,f)$:  the orbit corresponding to the class $[z]\in H^1(Z_{\wG}(p',h',e',f'))$ has representative $b(p',h',e',f')$ where $b\in \wG$ satisfies $z=b^{-1}\overline{b}$.
\end{theorem}
 
The complex semisimple and mixed orbits are parametrised as follows.
For each $i\in\{1,\ldots,10\}$ let $\Sigma_i$ be a set of $\wG$-orbit
representatives of semisimple elements in $\h_{\Pi_i}^\circ$ as specified in
Table \ref{tabW}. By Theorem \ref{thmME}, up to $\wG$-conjugacy, the complex
elements in $\g_1$ of mixed type are $s+n$ where $s\in\Sigma_i$ for some $i$
and  $n=n_{i,r}$ for some $r$, as specified in Table \ref{tabnij}. In the following we write \[\Sigma=\Sigma_1\cup\ldots\cup \Sigma_{10}.\]
The first problem is to decide which orbits in $\QQ(\C)$ have real representatives, but 
we know already which semisimple orbits have real representative. So
let us consider $p\in \Sigma$ such that  $p'\in \g_1(\R)$ is a real element in its
$\wG$-orbit. We define  $\a = \z_{\g(\C)}(p')$ as above, with the induced grading $\a=\a_0\oplus \a_1$. It remains to determine which nilpotent
$Z_{\wG(\C)}(p')$-orbits in $\a_1$ have real representatives. In the case that the real point $p'$ also lies in $\Sigma$, this is straightforward; we discuss this case in Section \ref{sec:pinS}. We treat the case $p'\not\in\Sigma$ in Section \ref{sec:pnotinS}.

In conclusion, our efforts lead to the following theorem.

\begin{theorem}\label{thmMTE}
Up to $\wG(\R)$-conjugacy, the mixed  elements in $\g_1(\R)$ are the elements in Tables \ref{tabrepsME1}-- \ref{tabrepsME_NS_4_2} in Appendices~\ref{appMES} and \ref{appMENS}; see Definition \ref{normalizer} for the notation used in all tables.
\end{theorem}

\subsection{Classification for the case  $p'\in \Sigma$}\label{sec:pinS}
Here we suppose that the $\wG$-orbit of $p$ has a real point $p'$ in $\Sigma$.
As before set $\a = \z_{\g(\C)}(p')$. If $p'\in \Sigma_i$, then we can assume
that $p'$ corresponds to an element in the first row of block $j=1$ in the
table for Case $i$ (see Tables \ref{tabreps10}--\ref{tabreps2}). With this
assumption, it follows from Theorem \ref{thmME} that every  nilpotent
$Z_{\wG(\C)}(p')$-orbit in $\a_1$ has a real representative; in particular, we
can assume that $e'=n_{i,r}$ for some $r$, as specified in Table \ref{tabnij}.
This then yields a real  $4$-tuple $(p',h',e',f')\in\QQ(\R)$.

We start by computing the centralisers
$Z_{\wG}(p',h',e',f')$, similarly to how we computed $Z_{\wG}(p')$ before;
the result is listed in Table \ref{tabZME}. Due to Theorem \ref{thmSE}b),
we can always take one explicit element for our computations; for  example,
in Case $i=2$,  we can always take $p'=u_1+u_2+u_3$.

We now consider the different cases $i=2,\ldots,10$; for $i=1$ and $i=11$ there are no  elements of mixed type. As before, cases $i\in\{5,6\}$ and $i\in\{8,9\}$ follow from $i=4$ and $i=7$, respectively.  We compute the first cohomology of each centraliser by using the same approach as described in  Section \ref{secProofRSE}; all the centralisers in this section can be found in Table \ref{tabZME}.

\medskip

\noindent {\bf Case $\pmb{i=2}$.} We can assume $p'$ corresponds to the first element in block $j=1$ in Table \ref{tabreps2}. By Theorem~\ref{thmME}, there is only one nilpotent element $e'=\mye{0011}$ such that $p'+e'$ has  mixed type; this is the real point we use. First, we compute a real $\ssl_2$-triple associated to $e'$. A direct calculation shows that $H^1(Z)$ for  $Z=Z_{\wG}(p',h',e',f')$ has 8 classes with representatives
\begin{equation}\label{H1Z}
\begin{array}{llll}
z_1=(I,I,I,I), & z_2=(-I,-I,I,I), & z_3=(-I,I,-I,I), & z_4=(-I,I,I,-I),\\
z_5 = (L,L,L,L), & z_{6}=(L,L,-L,-L), & z_{7}= (-L,L,-L,L),  & z_{8}= (L,-L,-L,L).
\end{array}
\end{equation}
This shows that the complex orbit $\wG(p,h,e,f)$ splits into $8$ real orbits: each $[z]\in H^1(Z_i)$ determines some $b\in \wG$ with $z=b^{-1}\overline{b}$, and  then $b(p'+e')$ is the real representative of the mixed type orbit corresponding to $[z]$; the resulting orbit representatives are listed in  Table~\ref{tabrepsME1}.

\medskip

\noindent {\bf Case $\pmb{i=3}$.} We proceed as before and the resulting real orbit representatives are exhibited in Table \ref{tabrepsME1b}. Here we have to consider the nilpotent elements $n_{3,1}$ and $n_{3,2}$. The case $e'=n_{3,1}$ yields the same centraliser $Z_1=Z$ as in Case $i=2$. The centraliser $Z_2$ for the case $e'=n_{3,2}$ leads to a first cohomology  $H^1(Z_2)$ with $8$ classes, given by representatives
\begin{equation}\label{H1Z32}
\begin{array}{llll}
	z_1=(I,I,I,I), & z_2=(-I,-I,I,I), & z_3=(-I,I,-I,I), & z_4=(-I,I,I,-I),\\
	z_5 = (I,-I,-I,I), & z_{6}=(I,-I,I,-I), & z_{7}= (I,I,-I,-I),  & z_{8}= (-I,-I,-I,-I).
\end{array}
\end{equation}


\medskip

\noindent {\bf Case $\pmb{i=4}$.} Here we have four nilpotent elements $n_{4,r}$ with $r\in\{1,2,3,4\}$; the real orbit representatives for this case are given in Table \ref{tabrepsME2}. The centralisers for $r=1$ and $r=2$ coincide with $Z$ as in Case $i=2$; if $r\in\{3,4\}$, then the centraliser is infinite; its first cohomology has $4$ classes with representatives 
\begin{equation}\label{H1Z43}
	\begin{array}{llll}
		z_1=(I,I,I,I), & z_2=(-I,-I,I,I), & z_3=(-I,I,-I,I), & z_4=(-I,I,I,-I).
	\end{array}
\end{equation}


\medskip

\noindent {\bf Case $\pmb{i=7}$.} There are six nilpotent elements $n_{7,r}$ with $r\in\{1,\ldots,6\}$; the real orbit representatives are exhibited in Table \ref{tabrepsME3}. The centraliser for $r=1,2,3$ is the same as $Z_2$ as in Case $i=3$; the centraliser for $r=6$ is as in Case $i=4$ and $n_{4,4}$. It therefore remains to consider  $r\in\{4,5\}$; we use the approach described in Section \ref{approach}. For $r=4$, a short calculation shows that  $k=(-L,L,K,K)$ is the only $1$-cocycle arising from $(-L,L,-J,J)$; it follows that the first cohomology has $8$ classes with representatives $z_i$ as in \eqref{H1Z43} along with $z_i k$, that is
\begin{equation}\label{H1Z74}\small
	\begin{array}{llll}
		z_1=(I,I,I,I), & z_2=(-I,-I,I,I), & z_3=(-I,I,-I,I), & z_4=(-I,I,I,-I),\\
		z_5 = (-L,L,K,K), & z_{6}=(L,-L,K,K), & z_{7}= (L,L,-K,K),  & z_{8}= (L,L,K,-K).
	\end{array}
\end{equation}
 For $r=5$, the first cohomology has representatives $z_1,\ldots,z_4$ as above along with $z_i (K,K,-L,L)$:
{\small\begin{equation}\label{H1Z75}
	\begin{array}{llll}
		z_1=(I,I,I,I), & z_2=(-I,-I,I,I), & z_3=(-I,I,-I,I), & z_4=(-I,I,I,-I),\\
		z_5 = (K,K,-L,L), & z_{6}=(-K,-K,-L,L), & z_{7}= (-K,K,L,L),  & z_{8}= (-K,K,-L,-L).
	\end{array} 
\end{equation}}

\medskip

\noindent {\bf Case $\pmb{i=10}$.} There are $12$ nilpotent elements; the real orbit representatives are exhibited in Tables \ref{tabrepsME4} and~\ref{tabrepsME5}. Cases $r=1,3,7,9$ lead to centralisers that have the same first cohomology as in Case $i=2$; Cases $r=2,4,6,8,10,12$ yield the same first cohomology as Case $(i,r)=(4,4)$, see  \eqref{H1Z43}. For $r\in\{5,11\}$ we obtain the following cohomology representatives:
\begin{equation}\label{H1Z105}
	\begin{array}{llll}
		z_1=(I,I,I,I), & z_2=(-I,-I,I,I).
	\end{array}
\end{equation}
For $j=13$ the fist cohomology has 2 classes with representatives
\begin{equation}\label{H1Z1013}
	\begin{array}{llll}
		z_1=(I,I,I,I), & z_2=(-I,I,-I,I).
	\end{array}
\end{equation}

\begin{remark}
  The semisimple parts of the real orbit representatives arising from cocycles involving $K$ or $-K$ are not in our fixed Cartan subspaces $\cc_1,\ldots,\cc_7$, and we use Remark \ref{remRewrite} to replace these elements by elements in our spaces. For example, consider the cocycle $z = (-L,L,K,K)$ in Case $(i,r)=(7,4)$. There is $b\in \wG$ with $b^{-1}\bar{b} = z$, and we compute $bp'$, so that  $b(p'+e')$ is a real point; in this case, $bp'\notin \cc_1\cup\ldots\cup\cc_7$. However, $bp'$ is $\wG(\R)$-conjugate to
an element in some space $\cc_j$, so there is $b'\in \wG(\R)$
such that $(b'b)p'\in \cc_j$. Since $b'$ is real, $(b'b)^{-1} \overline{b'b} = z$. Thus, there is $b_0\in \wG$ with $b_0^{-1} \bar{b}_0=z$ and $b_0p'\in \cc_j$. Now $(b'b)e'$ is real  and  $b_0(p'+e')$ is a real point.
\end{remark}


\subsection{Classification for the case  $p'\notin \Sigma$}\label{sec:pnotinS}
As before, we consider a  complex semisimple element $p$. By Theorem~\ref{thmME}, we can assume that $p\in\Sigma_i$. Let $p'$ be a real point in $\wG p$ as  in Theorem \ref{thmRSE}. This time we consider the case $p'\notin \Sigma$, so if $\a=\z_{\g(\C)}(p')=\a_0\oplus\a_1$, then we do not know which nilpotent $Z_{\wG}(p')$-orbits in $\a_1$ have real points. We now discuss how decide this question.
The method that we describe is borrowed from \cite[Section~5.3]{trivectors2}.
However, some difficulties that occurred in the case considered in
\cite{trivectors2} do not appear here, see Remark \ref{rem:finde}.

Recall that our proof of Theorem~\ref{thmRSE} has exhibited an explicit  $g\in \wG$ such that $p'=gp$; this construction used  Theorem \ref{thm1to1} and a 1-cocycle $n=g^{-1}\overline{g}\in N_{\wG}(\h_p^\circ)$. In the following write \[U_{p'} = \z_{\g}(p') \cap \g_1\quad\text{and}\quad U_p= \z_{\g}(p) \cap \g_1.\] Since $Z_{\wG}(p')=gZ_{\wG}(p)g^{-1}$, the next lemma allows us to  determine the nilpotent $Z_{\wG}(p')$-orbits in $U_{p'}$ from the known $Z_{\wG}(p)$-orbits in $U_p$, cf.\ Theorem~\ref{thmME}.
\begin{lemma}
 The map $\varphi\colon U_p\to U_{p'}$, $x\mapsto gx$, is a bijection that maps $Z_{\wG}(p)$-orbits to $Z_{\wG}(p')$-orbits.
\end{lemma}
Having determined the $Z_{\wG}(p')$-orbits in $U_{p'}$, it remains to decide  when  such a complex orbit has a real point.
Note that if $e'$ is a real nilpotent element in $\a_1$
then  $e'=\varphi(x)$  for some $x\in U_p$ lying  in the $Z_{\wG}(p)$-orbit of some $e=n_{i,r}$.  Motivated by this observation, we proceed as follows: we fix $p\in\Sigma_i$ and $p'=gp$, and for each $e=n_{i,r}$ we look for $x$ in the complex
$Z_{\wG}(p)$-orbit of $e$ such that $\varphi(x)$ is real. Note that the condition that $\varphi(x)$ is real is equivalent to $n\overline{x}=x$ where $n=g^{-1}\overline{g}$  as above. Thus, we define
\begin{eqnarray}\label{eqmu}\mu\colon U_p \to U_p,\quad x\mapsto n\overline{x};
\end{eqnarray}
note that $\mu^2=1$ since $n\overline{n}=1$ and, by construction, $\varphi(x)$ is real if and only if $\mu(x)=x$.  The following lemma is analogous to \cite[Lemma 5.3.1]{trivectors}.
\begin{lemma}\label{lemYstuff}
	Let $Y=Z_{\wG}(p)e$. Then $\mu(Y)=Y$ if and only if $\mu(y)\in Y$ for some $y\in Y$.
\end{lemma}

\begin{proof}
	It follows from Theorem \ref{thmSE}b) 
	that $Z_{\wG}(p)=Z_{\wG}(\h_\Pi^\circ)$, where $\h_{\Pi}^\circ$ is the component containing $p$. 
	Suppose there is $y\in Y$ such that $\mu(y)\in Y$; we have to show $\mu(Y)=Y$. Write $y=he$ with $h\in Z_{\wG}(p)$.
        We know that $\mu(y) = n \overline{y} = n\overline{h} e = k e$ for some $k\in Z_{\wG}(p)$; note that $\overline{he}=\overline{h}e$ because $e$ is real.
	Since $\overline{h}\in Z_{\wG}(p)=Z_{\wG}(\h_\Pi^\circ)$ and the latter is normal in $N_{\wG}(\h_\Pi^\circ)$, we have $n\overline{h} = zn$ for some $z\in Z_{\wG}(p)$, and the previous equation yields $ne = z^{-1}ke$.
	Now let $w\in Y$, say $w=te$ with $t\in Z_{\wG}(p)$. There is some $s\in Z_{\wG}(p)$ such that $n\overline{t}=sn$. We have $\mu(w) = n\overline{t}e = sne=sz^{-1}ke$ and, since $sz^{-1}k\in Z_{\wG}(p)$, we deduce that  $\mu(w)\in Y$, so $\mu(Y)=Y$. The other implication is trivial.
\end{proof}

\begin{corollary}
If there is $x\in Z_{\wG}(p)e$ such that $\varphi(x)$ is real, then 
$\mu(e)$ is $Z_{\wG}(p)$-conjugate to $e$.
\end{corollary}

If $e=n_{i,r}$ does not satisfy $\mu(e)\in Z_{\wG}(p)e$, then we can discard the
pair $(p,e)$. If $\mu(e)$ is $Z_{\wG}(p)$-conjugate to $e$, then we attempt to
compute $x\in Z_{\wG}(p)e$ such that  $\mu(x)=x$; then $e'=\varphi(x)$ is a real
nilpotent element commuting with $p'$. Once this real point is found, we
construct a real $4$-tuple $(p',h',e',f')$ and apply Theorem~\ref{thmMGal} to
compute the $\wG(\R)$-orbits in $\wG(p',h',e',f')$.

\begin{remark}\label{rem:finde}
  In our classification, using ad hoc methods, we always found suitable elements $x$ as above. We note that if such ad hoc methods would not have worked, then we could have used a method described in \cite[Section~5.3]{trivectors2} for finding such elements (or for deciding that none exists). This method is based on computations with the second cohomology set $H^2(Z_{\wG}(p))$. 
\end{remark}


\medskip

\noindent {\bf Classification approach.} We summarise our   approach for classifying the real mixed  orbits in $\wG(p+e)$, where the real point $p'$ in $\wG p$ (as in Theorem \ref{thmRSE}) does not lie in $\Sigma$.

\begin{iprf}
\item[{\bf (1)}] Recall that Tables  \ref{tabreps10}--\ref{tabreps} list our real semisimple orbits; each table corresponds to a case $i\in\{1,\ldots,10\}$ and has subcases $j=1,2,\ldots$, where $j=1$ lists elements in $\Sigma$. For each $i\in\{2,\ldots,10\}$ and each $j>1$ listed in the corresponding table, choose the first real element $p'$ in the block labelled $j$. The proof of Theorem \ref{thmRSE} shows that  $p'=gp$ where $p\in \Sigma_i$ as given in Table \ref{tabW}; in particular, the element $g$ is determined  by our  classification, and we define  $n=g^{-1}\overline{g}$.
\enlargethispage{3ex}
\item[{\bf (2)}] For each $p\in \Sigma_i$ as determined in (1) we consider each $e=n_{i,r}$ ($r=1,2,\ldots$) such that the elements $p+n_{i,r}$ are the mixed orbit representatives as determined in Theorem \ref{thmME}. Using \eqref{eqmu}, we then define $\mu$ with respect to $n$, and check whether $\mu(e)\in Z_{\wG}(p)e$. If true, then we compute $x\in Z_{\wG}(p)e$ with $\mu(x)=x$ by computing the $1$-eigenspace $U_p^\mu = \{ u\in U_p : \mu(u)=u\}$ and looking for some $x\in U_p^\mu$ that is $Z_{\wG}(p)$-conjugate to $e$. Note that  $\dim_\R U_p^\mu=\dim_\C U_p$; this follows from the fact that  $\mu$ is an $\R$-linear map of order 2, so $U_p$ is the direct sum of the $\pm 1$-eigenspaces. Since multiplication by $\imath$ is a bijective $\R$-linear map swapping these eigenspaces,  they have equal dimension. In our classification, this search is always successful, and we set $e'=\varphi(x)$.

\item[{\bf (3)}] We determine a real 4-tuple $(p',h',e',f')$ and
apply Theorem \ref{thmMGal} to find the real $\wG(\R)$-orbits in the $\wG$-orbit
of this 4-tuple. If $(p'',h'',e'',f'')$ is a representative of such an
orbit then $p''+e''$ is the corresponding element of mixed type. This is a
representative of a $\wG(\R)$-orbit contained in $\wG(p+e)$. 
\end{iprf}


We now discuss the individual cases in detail. The following case distinction determines the relevant real points $n_{i,j,r}$ for Case $i$ and the cohomology class $[\gamma_j]$ of $H^1(\Gamma_{p_i})$ with $\gamma_j$ as determined in the proof of Theorem~\ref{thmRSE}; note that we do not have to  consider the trivial class $[\gamma_1]$ because this class produces elements in~$\Sigma$.

\medskip

\noindent {\bf Case $\pmb{i=2}$.} We have to consider cocycles $\gamma_2,\ldots,\gamma_8$, with corresponding elements $n_j$ given by (see Section $6$)
\[\begin{array}{llll}
	n_2=(I,I,L,-L), & n_3=(I,L,I,-L), & n_4=(L,I,I,L),& \\
	n_5 = (L,I,I,-L), & n_{6}=(I,L,I,L), & n_{7}= (I,I,L,L),  & n_{8}= (-I,I,I,I).
\end{array}
\]
For this case there is only one nilpotent element $e=n_{2,1}$. After a short calculation we can see that $\mu(e)=n_je$ is conjugate to $e$ for all $j=2,\dots,8$.  
For $n_2=(I,I,L,-L)$ we obtain $g_2=(I,I,M,D(\zeta))$; following Remark~\ref{remRewrite}, we replace $g_2$ by  $g_2=(-I,-J,-M,-M J)$. We then compute $U_p^\mu$ and verify that $e\in U_p^\mu$, thus $x=e$ is the element we are looking for, and therefore we set $e'=\varphi(x)=g_2e=-\mye{0110}$; we denote the latter by $n_{2,2,1}$. For $n_3$ we obtain  $g_3=(I,M,I,D(\zeta))$. Again, to get elements in one of the seven Cartan subspaces, it is necessary to replace it by $g_3=(-I,-M,-J,-M J)$. After computing $U_q^\mu$ we have that $x=\imath e\in U_p^\mu$, which is conjugate to $e$ under the action of $Z_{\wG}(p)$ via $g=(D(\imath),D(\imath),D(\imath^{-1}),D(\imath^{-1}))$. Therefore we set $e'=\varphi(x)=g_3e=\mye{0000}$;
the latter is denoted by $n_{2,3,1}$. The other cases $j>3$ are computed along the same way.

\medskip

\noindent {\bf Case $\pmb{i=3}$.} We have to consider the cocycle $\gamma_2$ with $n_2=(-I,I,I,I)$. There are two nilpotent elements $n_{3,1}$ and $n_{3,2}$. 
We compute $U_p^\mu$ and see that $x=\imath\mye{0011}$ is a nilpotent element in $U_p^\mu$ conjugate to $n_{3,1}$ via $(D(a)^{-1},D(a)^{-1},D(a),D(a))$ with $a^{-4}=\imath$. Thus, the real nilpotent element is $e'=\varphi(x)=-\mye{0011}$, denoted $n_{3,2,1}$. Similarly,  we find $x\in U_p^\mu$ such that $x$ is conjugate to $n_{3,2}$, see $\varphi(x)=n_{3,2,2}$ in Table \ref{tabNilpElements}.

\medskip

\noindent {\bf Case $\pmb{i=4}$.} We have to consider cocycles $\gamma_2,\ldots,\gamma_4$, with corresponding elements
\[\begin{array}{llll}
	n_2=(M,M,M,M), & n_3=(I,I,L,L), & n_4=(L,L,-K,K).
\end{array}
\]
For each of them, we have to consider four nilpotent elements $n_{4,1},\ldots,n_{4,4}$. After computing $U_p^\mu$ for $n_2$, we observe that each $n_{4,\ell}\in U_p^\mu$, thus the real representatives are $n_{4,2,1},\ldots,n_{4,2,4}$, defined as $n_{4,2,\ell}=g_2n_{4,\ell}=n_{4,\ell}$  with $g_2=(D(\eta^5),D(\eta^5),D(\eta^5),D(\eta^5))$. The case $n_3$ is similar, so now let us consider $n_4$. We obtain that  $\mu(n_{4,\ell})$ is not conjugate to $n_{4,\ell}$ for $\ell\in\{3,4\}$, thus we have to only consider $n_{4,1}$ and $n_{4,2}$. For $n_{4,1}$ there is no $x\in U_q^\mu$ such that $x$ is conjugate to $n_{4,1}$. This is seen by acting on
$n_{4,1}$ by a general element $g$ of $Z_{\wG}(p)$ which is
$g=h(D(a^{-1}),D(a),D(b^{-1}),D(b))$ where $a,b\in \C^*$ and $h$ lies in the
component group (see Table \ref{tabZ}). From the expressions obtained it is
straightforward to see that the image can never lie in $U_p^\mu$. 

On the other hand, we observe that $n_{4,2}$ lies in $U_p^\mu$  and therefore
we set $e'=\varphi(n_{4,2})=g_4n_{4,2}=\frac12(-\mye{1110}-\mye{1101}+\mye{1010}+\mye{1001}+\mye{0110}+\mye{0101}-\mye{0010}-\mye{0001})$. In this case it is  necessary to replace $g_4$ in order to get elements in one of the seven Cartan subspaces.


\medskip

\noindent {\bf Cases $\pmb{i\in\{7,10\}}$.} The nilpotent elements for the remaining cases $i=7$ and $i=10$ are computed analogously.

\bigskip

The next step is to compute the centralisers $Z=Z_{\wG}(p',h',e',f')$ of the $4$-tuples $(p',h',e',f')$, and then the first cohomology $H^1(Z)$ as we did in Table \ref{tabZME}. The difference here is that $p'\notin \Sigma$ and $e'$ is a nilpotent element from Table \ref{tabNilpElements}, with corresponding $\ssl_2$-triple 
$(h',e',f')$ in $\z_\g(p')$. We exemplify the details  with a few examples:

\begin{exa}
For $(i,j)=(2,2)$ let $p'$ be the first element in the second block of  Table \ref{tabreps2} and let  $e'=n_{2,2,1}=-\mye{0011}$.  We compute an $\ssl_2$-triple with nilpotent element $e'$ in the whole Lie algebra and then check that it centralises $p'$. Since the centralisers of semisimple elements in the same component are equal, in order to compute $Z_{\wG}(p')$ we can assume that the parameters defining $p'$ are $\lambda_1=\lambda_3=1$ and $\lambda_2=\imath$. Using Groebner basis techniques, we determine that $Z_{\wG}(p',h',e',f')$ is the same as in the first row of Table \ref{tabZME}, hence $H^1(Z)$ has $8$ classes with representatives given by \eqref{H1Z}. The results are listed in the first block of Table \ref{tabrepsME_NS_2_1}.
\end{exa}

In fact, in most of the cases, the centraliser $Z=Z_{\wG}(p',h',e',f')$  for  Case $i$, cohomology class $[\gamma_j]$, and  nilpotent element $e'=n_{i,j,r}$ is exactly the same given in Table \ref{tabZME} for parameters $i$ and $r$. For example, if $i=2$ and $e'=n_{2,j,1}$ with  $j\in\{2,3,4,5,6,7,8\}$, then $Z_{\wG}(p',h',e',f')$ is always given in the first row of Table \ref{tabZME}. The only exceptions where this behaviour was not observed are the following two cases.

\begin{exa}
For $(i,j)=(4,4)$ let $p'$ be the first element in the second block of Table \ref{tabreps4} and let $e'=n_{4,4,1}$ as in Table \ref{tabNilpElements}. Here the centraliser $Z=Z_{\wG}(p',h',e',f')$ is generated by $(-I,-I,I,I)$, $(-I,I,-I,I)$, $(K,K,J,J)$, therefore different from the centraliser given in Table \ref{tabZME} for parameters $(i,r)=(4,1)$. The first cohomology  $H^1(Z)$ has $4$ classes with representatives 
\begin{equation}\label{H1ZnoS44}
	\begin{array}{llll}
		z_1=(I,I,I,I), & z_2=(-I,-I,I,I), & z_3=(-I,I,-I,I), & z_4=(-I,I,I,-I).
	\end{array}
\end{equation}
For $(i,j)=(7,2)$ and $e'=n_{7,2,5}$ the centraliser is generated by  $(-I,I,-I,I)$, $(-I,I,I,-I)$, $(-J,J,L,L)$, and $\left\{ ( D(a)^{-1}, D(a),I,I) \;:\; a\in \C^\times \right\}$, and therefore is different from the centraliser given in Table \ref{tabZME} for parameters $(i,r)=(7,5)$; its first cohomology has  $8$ classes with representatives
\begin{equation}\label{H1ZnoS75}
		\begin{array}{llll}
	z_1=(I,I,I,I), & z_2=(-I,I,-I,I), & z_3=(-I,I,I,-I), & z_4=(I,I,-I,-I),\\
	z_5 = (K,K,L,L), & z_{6}=(-K,K,-L,L), & z_{7}= (-K,K,L,-L),  & z_{8}= (K,K,-L,-L).
\end{array} 
\end{equation}
\end{exa}

Having found all real points and having determined the corresponding cohomology class representatives, the last step is to employ the usual Galois theory approach to obtain a real representative for each of the mixed orbits obtained above. We exhibit the results in Tables \ref{tabrepsME1}-- \ref{tabrepsME_NS_4_2}. This proves Theorem \ref{thmMTE}.

\appendix 

\newpage

\section{Tables}
\subsection{Complex classification}\label{appComp}
~{}

\begin{table}[H]
\scalebox{0.92}{\begin{tabular}{r|r|r|c|c|c|c}
  {\bf $\pmb{i}$} & {\bf type of $\Pi_i$} & {\bf roots of $\Pi_i$} & {\bf elements of $\pmb{\h_{\Pi_i}}$} & {\bf condition for being in $\pmb{\h_{\Pi_i}^\circ}$}& $\pmb{\Gamma_{\Pi_i}}$ & $\pmb{\z_\g(p_i)'}$\\
  \hline
  1& $\emptyset$     &  & $\lambda_1u_1+\cdots+\lambda_4u_4$ &
  $\lambda_i\neq 0$ and  $\lambda_1\notin \{\pm \lambda_2\pm \lambda_3\pm \lambda_4\}$ & $W$ & $0$\\
  2 & ${\rm A}_1$ & $\alpha_4$ & $\lambda_1u_1+\lambda_2u_2+\lambda_3u_3$ &
  $\lambda_i\neq 0$ and  $\lambda_1\notin\{ \pm\lambda_2\pm\lambda_3\}$& $(\Z/2\Z)^3$ & $\ssl(2,\C)$\\
  3 & ${\rm A}_2$ & $\alpha_2,\alpha_4$ & $\lambda_1 (u_1-u_2) + \lambda_2 (u_1-u_3)$ &
   $\lambda_u\neq 0$ and $\lambda_1\ne -\lambda_2$& $\langle-I_4\rangle$ & $\ssl(3,\C)$\\
  4 & $2{\rm A}_1$ & $\alpha_1,\alpha_3$ & $\lambda_1u_1+\lambda_2u_4$ &
  $\lambda_i\neq 0$ and $\lambda_1\notin \{\pm\lambda_2\}$& $\Dih_4$ & $\ssl(2,\C)^2$\\
  5 & $2{\rm A}_1$ & $\alpha_1,\alpha_4$ & $\lambda_1u_1+\lambda_2u_3$ &
  $\lambda_i\neq 0$ and $\lambda_1\notin \{\pm\lambda_2\}$& $\Dih_4$ & $\ssl(2,\C)^2$\\
  6 & $2{\rm A}_1$ & $\alpha_3,\alpha_4$ & $\lambda_1u_1+\lambda_2u_2$ &
    $\lambda_i\neq 0$ and $\lambda_1\notin \{\pm\lambda_2\}$& $\Dih_4$ & $\ssl(2,\C)^2$\\
  7 & ${\rm A}_3$ & $\alpha_1,\alpha_2,\alpha_3$ & $\lambda_1(u_1-u_4)$ & $\lambda_1\neq 0$& $\langle-I_4\rangle$ & $\ssl(4,\C)$\\
  8 & ${\rm A}_3$ & $\alpha_1,\alpha_2,\alpha_4$ & $\lambda_1(u_1-u_3)$ & $\lambda_1\neq 0$& $\langle-I_4\rangle$ & $\ssl(4,\C)$\\
  9 & ${\rm A}_3$ & $\alpha_2,\alpha_3,\alpha_4$ & $\lambda_1(u_1-u_2)$ & $\lambda_1\neq 0$& $\langle-I_4\rangle$ & $\ssl(4,\C)$\\
  10 & $3{\rm A}_1$ & $\alpha_1,\alpha_3,\alpha_4$ & $\lambda_1 u_1$ & $\lambda_1\neq 0$& $\langle-I_4\rangle$ & $\ssl(2,\C)^3$\\
  11 & ${\rm D}_4$ & $\alpha_1,\ldots,\alpha_4$  & $0$ & $0$& $1$ & $\so(4,\C)$
\end{tabular}}\\[1ex]\caption{This is \cite[Table 2]{complex}:  Complete root subsystems $\Pi_i$ of $\Phi$ and related data, with parameters $\lambda_1,\ldots,\lambda_4\in\C$; the last column displays
the derived algebra of the centraliser $\z_\g(p_i)$ for $p_i \in
\h_{\Pi_i}^\circ$. Elements in $\h_{\Pi_i}^\circ$ are not $\wG$-conjugate to elements in $\h_{\Pi_j}^\circ$ if $i\ne j$; two elements in the same component $\h_{\Pi_i}^\circ$ are $\wG$-conjugate if and only if they are $\Gamma_{\Pi_i}$-conjugate.}\label{tabW}
\end{table}


\begin{table}[H]\renewcommand\arraystretch{1.3}
 \scalebox{0.92}{
\begin{tabular}{ccc}
    {\bf $\pmb{i}$}  & {\bf identity component $\pmb{ Z_{\widehat{G}}(s)^\circ}$} & {\bf preimages of generators of $\pmb{Z_{\widehat{G}}(s)/Z_{\widehat{G}}(s)^\circ}$} \\\hline
	 1 & 1 & $(J,J,J,J),(-I,-I,I,I),(-I,I,-I,I),(K,K,K,K)$ \\\hline
	 
	 2 & $\left\{ (D(a)^{-1}, D(a)^{-1}, D(a), D(a)) \;:\; a\in \C^\times \right\}$ &
$ (-I,-I,I,I),(-I,I,-I,I),(J,J,J,J)$         \\\hline
	 
	 3 & $\left\{(A^\#, A^\#, A, A) \;:\;
	 A \in \SL(2,\C)\right\}$ & $(-I,-I,I,I),(-I,I,-I,I)$ \\\hline
	 
	 4 & $\left\{ (D(a)^{-1}, D(a), D(b)^{-1}, D(b)) \;:\; a,b\in \C^\times\right \}$ & $(-I,I,-I,I),(J,J,J,J)$ \\\hline
         5 & $\left\{ (D(a)^{-1}, D(b)^{-1}, D(a), D(b)) \;:\; a,b\in \C^\times\right \}$ & $(-I,-I,I,I),(J,J,J,J)$ \\\hline
         6 & $\left\{ (D(a)^{-1}, D(b), D(b)^{-1}, D(a)) \;:\; a,b\in \C^\times\right \}$ & $(-I,I,-I,I),(J,J,J,J)$ \\\hline

	 7 & $\left\{ (A^\#, A, B^\#, B) : A,B\in \SL(2,\C)\right\}$ & $(-I,I,-I,I)$ \\\hline
	 8 & $\left\{ (A^\#, B^\#, A, B) : A,B\in \SL(2,\C)\right\}$ & $(-I,-I,I,I)$ \\\hline
         9 & $\left\{ (A^\#, B, B^\#, A) : A,B\in \SL(2,\C)\right\}$ & $(-I,I,-I,I)$ \\\hline
         
	 10 & $\left\{ (D(abc)^{-1}, D(a), D(b), D(c)) \;:\; a,b,c\in \C^\times \right\}$ &  $(J,J,J,J)$ \\[2ex]
  \end{tabular}}\caption{This is \cite[Table 3]{complex}; the groups $Z_{\widehat{G}}(s)$: the entry $i$ is the label of the canonical semisimple set $\h_{\Pi_i}^\circ$ that contains $s$, as in Table \ref{tabW}; the notation is  explained in \eqref{eqmats}.}\label{tabZ} 
\end{table}


\begin{table}[H]\renewcommand\arraystretch{1.1}
  \scalebox{0.95}{\begin{tabular}{ll}
 ${\pmb{i}}$ & {\bf nilpotent elements}\\\hline     
$2$& $n_{2,1}=\mye{0011}$\\
$3$& $n_{3,1}=\mye{0011}$,\; $n_{3,2}=\mye{0111}+\mye{1011}+\mye{0010}+\mye{0001}$\\
$4$& $n_{4,1}=\mye{0110}+\mye{1010}$,\; $n_{4,2}=\mye{0110}+\mye{0101}$,\; $n_{4,3}=\mye{0110}$,\; $n_{4,4}=\mye{0101}$\\
$5$& $n_{5,1}=\mye{0110}+\mye{1100}$,\; $n_{5,2}=\mye{0110}+\mye{0011}$,\; $n_{5,3}=\mye{0110}$,\; $n_{5,4}=\mye{0011}$\\
$6$& $n_{6,1}=\mye{0011}+\mye{1010}$,\; $n_{6,2}=\mye{0011}+\mye{0101}$,\; $n_{6,3}=\mye{0011}$,\; $n_{6,4}=\mye{0101}$\\
$7$& $n_{7,1}=\mye{1101}+\mye{1011}+\mye{1000}+\mye{0001}$,\; $n_{7,2}=\mye{1101}+\mye{1010}+\mye{0001}$,\\& $n_{7,3}=\mye{1011}+\mye{1000}+\mye{0101}$,$n_{7,4}=\mye{1011}+\mye{1000}$,\; $n_{7,5}=\mye{1101}+\mye{0001}$,\; $n_{7,6}=\mye{1001}$\\
$8$& $n_{8,1}=\mye{1011}+\mye{1101}+\mye{1000}+\mye{0001}$,\; $n_{8,2}=\mye{1011}+\mye{1100}+\mye{0001}$,\\& $n_{8,3}=\mye{1101}+\mye{1000}+\mye{0011}$,\; $n_{8,4}=\mye{1101}+\mye{1000}$,\; $n_{8,5}=\mye{1011}+\mye{0001}$,\; $n_{8,6}=\mye{1001}$\\
$9$& $n_{9,1}=\mye{1101}+\mye{1110}+\mye{1000}+\mye{0100}$,\; $n_{9,2}=\mye{1101}+\mye{1010}+\mye{0100}$,\\& $n_{9,3}=\mye{1110}+\mye{1000}+\mye{0101}$,\; $n_{9,4}=\mye{1110}+\mye{1000}$,\;  $n_{9,5}=\mye{1101}+\mye{0100}$,\; $n_{9,6}=\mye{1100}$\\
$10$&$n_{10,1}=\mye{1100}+\mye{1010}+\mye{0110}$,\; $n_{10,2}=\mye{1010}+\mye{0110}$,\; $n_{10,3}=\mye{1010}+\mye{0110}+\mye{0011}$,\\&$n_{10,4}=\mye{1100}+\mye{0110}$,\;  $n_{10,5}=\mye{0110}$,\; $n_{10,6}=\mye{0110}+\mye{0011}$,\; $n_{10,7}=\mye{1100}+\mye{0110}+\mye{0101}$,\\&$n_{10,8}=\mye{0110}+\mye{0101}$,\; $n_{10,9}=\mye{0110}+\mye{0101}+\mye{0011}$,\;  $n_{10,10}=\mye{1100}+\mye{1010}$,\\& $n_{10,11}=\mye{1010}$,\; $n_{10,12}=\mye{1010}+\mye{0011}$,\; $n_{10,13}=\mye{0011}$.\\[1ex]
\end{tabular}}\caption{This is from \cite[Theorem 3.7]{complex}: The nilpotent elements $n_{i,j}$ used in Theorem \ref{thmME}.}\label{tabnij} 
\end{table}


\subsection{Semisimple elements}\label{appSSTab}
~{}
 
\begin{table}[H]\renewcommand\arraystretch{1.2}
  \hspace*{-0.2cm}\scalebox{0.9}{\begin{tabular}[t]{c|c}
    $\pmb{w\in W}$ & {\bf cocycle in $\pmb{\widehat{G}}$}\\
    \hline
    $\diag(-1,-1,-1,-1)$ & $(-I,I,I,I)$ \\
    $\diag(-1,-1,-1,1)$ & $(M,M,-N,N)$\\
    $\diag(-1,-1,1,1)$ & $(L,I,I,L)$\\
    $\diag(-1,1,-1,1)$ & $(I,L,I,L)$\\
    $\diag(-1,1,1,-1)$ & $(I,I,L,L)$\\
    $\diag(1,-1,-1,1)$ & $(I,I,L,-L)$\\
    $\diag(1,-1,1,-1)$ & $(I,L,I,-L)$\\
    $\diag(1,1,-1,-1)$ & $(L,I,I,-L)$\\ 
    $\diag(-1,1,1,1)$ & $(M,M,M,M)$\\ 
    $\diag(1,-1,1,1)$ & $(N,M,M,N)$\\ 
    $\diag(1,1,-1,1)$ & $(M,N,M,N)$\\ 
    $\diag(1,1,1,-1)$ & $(M,M,N,N)$\\ 
    $\diag(1,-1,-1,-1)$ & $(-N,N,N,N)$\\ 
\end{tabular}}\hspace*{0.8cm}\scalebox{0.9}{\begin{tabular}[t]{c|c}
    $\pmb{w\in W}$ & {\bf cocycle in $\pmb{\widehat{G}}$}\\
    \hline
    $\SmallMatrix{0&0&0&-1\\0&0&1&0\\0&1&0&0\\-1&0&0&0}$ &
    $(L,L,-K,K)$\\
    $\SmallMatrix{0&0&-1&0\\0&0&0&-1\\-1&0&0&0\\0&-1&0&0}$ &
    $(I,K,I,K)$\\
    $\SmallMatrix{0&-1&0&0\\-1&0&0&0\\0&0&0&-1\\0&0&-1&0}$ &
    $(K,I,I,K)$
  \end{tabular}}\\[1ex]\caption{Cocycles in $\widehat{G}$ that induce the cocycles in $W$
whose equivalence classes form the various sets $H^1(\Gamma_{\Pi_i})$: these will be the elements $n_i$ that map under $\varphi$ to $\gamma_i$ as described in Theorem~\ref{thmRP}.  Matrices $I,J,K,L,M,N$ are from \eqref{eqmats} and elements in $H^1(\Gamma_{\Pi_i})$ are considered as elements in $W$.}\label{tabCoc}
\end{table}

\begin{table}[H]\small \renewcommand\arraystretch{1.3}
  \hspace*{-0.2cm}\begin{tabular}{c|ccccccccc}
$\pmb{A}$ & $-I$ & $M$ & $-M$ & $N$ & $-N$ & $L$ & $-L$ & $K$ & $-K$ \\\hline
    $\pmb{\varepsilon(A)}$ & $L$  &  $D(\eta^5)$ &$D(\eta)$ &$D(\eta^7)$ & $-D(\eta^3)$ & $M$ &  $D(\zeta)$ & $LF$ & $F$    
    \end{tabular}\\[1ex]\caption{Matrices $F,I,J,K,L,M,N$ are from \eqref{eqmats}, and $\eta$ is  a primitive $16$-th root of unity with $\eta^2=\zeta$; if $A\in \SL(2,\C)$, then $\varepsilon(A)\in\SL(2,\C)$ satisfies $\varepsilon(A)^{-1}\overline{\varepsilon(A)}=A$.}\label{tabKA}
\end{table}

\medskip


        
	\begin{table}[H]
		\centering\small
			\begin{tabular}[H]{c|c|lllll}
			  $\pmb{j}$ & $\pmb{k}$  & \multicolumn{5}{l}{{\bf real orbit representatives (semisimple)}}\\   
			  \hline
$1$ & $1,2$ & $\lambda_1 (1,0,0,0)$
& $(-2,2,2,2)/\lambda_1$\\ \cline{2-2}
&\multicolumn{1}{l}{\bf Coefficients:}   & \multicolumn{4}{l}{with respect to the basis $\{u_1,u_2,u_3,u_4\}$}\\
& \multicolumn{1}{l}{\bf Conditions:}   &  \multicolumn{5}{l}{$\lambda_1$ up to $\Gamma_{\Pi_{10}}$-conjugacy, $\lambda_1\in\R^\times$}  \\ \hline

$2$ & $1,2$ &$\imath\lambda_1 (1,0,0,0)$
& $(-2,2,2,2)/(\imath\lambda_1)$ \\ \cline{2-2}
                          &\multicolumn{1}{l}{\bf Coefficients:}   & \multicolumn{4}{l}{with respect to the basis $\{v_1,v_2,v_3,v_4\}$}\\
				& \multicolumn{1}{l}{\bf Conditions:}   &  \multicolumn{5}{l}{$\lambda_1$ up to $\Gamma_{\Pi_{10}}$-conjugacy, $\lambda_1\in\imath\R^\times$ }
	\end{tabular}\\[1ex]\caption{{\bf Case $\pmb{i=10}$:}  The table lists real orbit representatives corresponding to $\gamma_j$ and $[z_k]\in H^1(Z)$. }
	\label{tabreps10}
        \end{table}


\renewcommand\arraystretch{1.1}
	\begin{table}[H]
		\centering\small
			\begin{tabular}[H]{c|c|llll}
			  $\pmb{j}$ & $\pmb{k}$  & \multicolumn{4}{l}{{\bf real orbit representatives (semisimple)}}\\   
			  \hline
			 $1$ & $1,2$ &  $\lambda_1(1,-1)$
			             &  $-\lambda_1(1,1)$\\  \cline{2-2}
                          &\multicolumn{1}{l}{\bf Coefficients:}   & \multicolumn{4}{l}{with respect to the basis $\{u_1,u_4\}$}\\
				& \multicolumn{1}{l}{\bf Conditions:}   &  \multicolumn{4}{l}{$\lambda_1$ up to $\Gamma_{\Pi_7}$-conjugacy, $\lambda_1\in\R^\times$}\\ \hline
			
			$2$ & $1,2$ &   $\imath \lambda_1(1,-1)$
		                    & $-\imath\lambda_1(1,1)$ \\ \cline{2-2}
                          &\multicolumn{1}{l}{\bf Coefficients:}   & \multicolumn{4}{l}{with respect to the basis $\{v_1,v_4\}$}\\
				& \multicolumn{1}{l}{\bf Conditions:}   &  \multicolumn{4}{l}{$\lambda_1$ up to $\Gamma_{\Pi_7}$-conjugacy, $\lambda_1\in\imath \R^\times$}
			
		\end{tabular}\\[1ex]
\caption{{\bf Cases $\pmb{i=7,8,9}$:} The table lists real orbit representatives corresponding to $\gamma_j$ and $[z_k]\in H^1(Z)$ for $\h_{\Pi_7}^\circ$; the real representatives for Cases $i\in\{8,9\}$ are obtained via Remark \ref{remSym}}\label{tabreps7}
	\end{table}


	\begin{table}[H]
			\centering\small
			\begin{tabular}[H]{c|c|llll}
			  $\pmb{j}$ & $\pmb{k}$  & \multicolumn{4}{l}{{\bf real orbit representatives (semisimple)}}\\   
			  \hline

				$1$ & $1,2$ & $(\lambda_1,0,0,\lambda_4)$ & $(-\lambda_1,0,0,\lambda_4)$\\ 
				& $3$ & \multicolumn{4}{l}{$( -\lambda_1+\lambda_4, -\lambda_1-\lambda_4, \lambda_1+\lambda_4, -\lambda_1+\lambda_4)/2$} \\
				& $4$ & \multicolumn{4}{l}{$( -\lambda_1+\lambda_4,  \lambda_1+\lambda_4, -\lambda_1-\lambda_4, -\lambda_1+\lambda_4)/2$} \\
				& $5$ & \multicolumn{4}{l}{$( -\lambda_1-\lambda_4, \lambda_1-\lambda_4,  \lambda_1-\lambda_4, \lambda_1+\lambda_4)/2$} \\
			        & $6$ & \multicolumn{4}{l}{$( -\lambda_1-\lambda_4, -\lambda_1+\lambda_4, -\lambda_1+\lambda_4, \lambda_1+\lambda_4)/2$} \\\cline{2-2}
                          &\multicolumn{1}{l}{\bf Coefficients:}   & \multicolumn{4}{l}{with respect to the basis $\{u_1,u_2,u_3,u_4\}$}\\
				& \multicolumn{1}{l}{\bf Conditions:}   &  \multicolumn{4}{l}{$(\lambda_1,\lambda_4)$ up to $\Gamma_{\Pi_4}$-conjugacy, $\lambda_1,\lambda_4\in\R^\times$, $\lambda_1\notin \{\pm \lambda_4\}$}  \\
\end{tabular}\\[1ex]\caption{{\bf Cases $\pmb{i=4,5,6}$:} The table lists real orbit representatives corresponding to $\gamma_j$ and $[z_k]\in H^1(Z)$ for $\h_{\Pi_4}^\circ$; the real representatives for Cases $i\in\{5,6\}$ are obtained via Remark \ref{remSym}.}
		\label{tabreps4}
		
	\end{table}
 

	\begin{table}[H]
		\centering\small
			\begin{tabular}[H]{c|c|llll}
			  $\pmb{j}$ & $\pmb{k}$  & \multicolumn{4}{l}{{\bf real orbit representatives (semisimple)}}\\   
			  \hline		
			$1$ & $1,\ldots,4$ &  $(\lambda_1+\lambda_2,-\lambda_1,-\lambda_2)$
			 &  $ (-\lambda_1-\lambda_2,-\lambda_1,-\lambda_2)$
			 & $(-\lambda_1-\lambda_2,-\lambda_1,\lambda_2)$
			  & $(-\lambda_1-\lambda_2,\lambda_1,-\lambda_2)$\\ \cline{2-2}
 &\multicolumn{1}{l}{\bf Coefficients:}   & \multicolumn{4}{l}{with respect to the basis $\{u_1,u_2,u_3\}$}\\
				& \multicolumn{1}{l}{\bf Conditions:}   &\multicolumn{4}{l}{$(\lambda_1,\lambda_2)$ up to $\Gamma_{\Pi_3}$-conjugacy, each $\lambda_i\in \R^\times$, $\lambda_1\ne -\lambda_2$}
			\\ \hline
			
			$2$ & $1,\ldots,4$ &  $\imath (\lambda_1+\lambda_2,-\lambda_1,-\lambda_2)$
			& $-\imath(\lambda_1+\lambda_2,\lambda_1,\lambda_2) $	
		        & $-\imath(\lambda_1+\lambda_2,\lambda_1,-\lambda_2) $
			& $-\imath(\lambda_1+\lambda_2,-\lambda_1,\lambda_2) $\\ \cline{2-2}
&\multicolumn{1}{l}{\bf Coefficients:}   & \multicolumn{4}{l}{with respect to the basis $\{v_1,v_2,v_3\}$}\\
				& \multicolumn{1}{l}{\bf Conditions:}   &  \multicolumn{4}{l}{$(\lambda_1,\lambda_2)$ up to $\Gamma_{\Pi_3}$-conjugacy, each $\lambda_i\in \imath \R^\times$, $\lambda_1\ne -\lambda_2$} 
		\end{tabular}\\[1ex]\caption{{\bf Case $\pmb{i=3}$:} The table lists real orbit representatives corresponding to  $\gamma_j$ and $[z_k]\in H^1(Z)$.  } 		
		\label{tabreps3}
	\end{table}


\renewcommand\arraystretch{1.1}
\begin{table}[H]{\centering\small
			\begin{tabular}[H]{c|c|llll}
			  $\pmb{j}$ & $\pmb{k}$  & \multicolumn{4}{l}{{\bf real orbit representatives (semisimple)}}\\   
			\hline
			$1$ & $1,\ldots,4$ & $(\lambda_1,\lambda_2,\lambda_3,0)$
			 & $(-\lambda_1,\lambda_2,\lambda_3,0)$
			 & $(-\lambda_1,\lambda_2,-\lambda_3,0)$
			 & $(-\lambda_1,-\lambda_2,\lambda_3,0)$\\ 
			& $5$ & \multicolumn{4}{l}{$( -\lambda_1+\lambda_2+\lambda_3,
			 -\lambda_1+\lambda_2-\lambda_3, -\lambda_1-\lambda_2+\lambda_3,
			 \lambda_1+\lambda_2+\lambda_3)/2$}\\
			& $6$ & \multicolumn{4}{l}{$(-\lambda_1-\lambda_2-\lambda_3,
			 \lambda_1+\lambda_2-\lambda_3, \lambda_1-\lambda_2+\lambda_3,
			 \lambda_1-\lambda_2-\lambda_3)/2$}\\
			& $7$ & \multicolumn{4}{l}{$(-\lambda_1-\lambda_2+\lambda_3,
			 \lambda_1+\lambda_2+\lambda_3, -\lambda_1+\lambda_2+\lambda_3,
			 -\lambda_1+\lambda_2-\lambda_3)/2$}\\
			& $8$ & \multicolumn{4}{l}{$(-\lambda_1+\lambda_2-\lambda_3,
			 -\lambda_1+\lambda_2+\lambda_3,\lambda_1+\lambda_2+\lambda_3,
			 -\lambda_1-\lambda_2+\lambda_3)/2$}\\\cline{2-2}
                         &\multicolumn{1}{l}{\bf Coefficients:}   & \multicolumn{4}{l}{with respect to the basis $\{u_1,u_2,u_3,u_4\}$}\\
			& \multicolumn{1}{l}{\bf Conditions:}   &  \multicolumn{4}{l}{$(\lambda_1,\lambda_2,\lambda_3)$ up to $\Gamma_{\Pi_2}$-conjugacy, each $\lambda_i\in\R^\times$ and $\lambda_1\notin \{\pm\lambda_2\pm\lambda_3\}$} \\
\end{tabular}}\\[1ex]\caption{{\bf Case $\pmb{i=2}$:} The table lists real orbit representatives corresponding to $\gamma_j$ and $[z_k]\in H^1(Z)$.} 
	\label{tabreps2}
\end{table}


\newpage
         
\renewcommand\arraystretch{1.2}
\begin{table}[H]\resizebox{\columnwidth}{!}{\centering
		\begin{tabular}[H]{c|c|llll}
			$\pmb{j}$ & $\pmb{k}$  & \multicolumn{4}{l}{{\bf real orbit representatives (semisimple)}}\\
			\hline
			$1$ & $1,\ldots,4$ & $(\lambda_1,\lambda_2,\lambda_3,\lambda_4)$
	  	        & $(-\lambda_1,\lambda_2,\lambda_3,-\lambda_4)$ 
		        & $(-\lambda_1,\lambda_2,-\lambda_3,\lambda_4)$ 
		        & $(-\lambda_1,-\lambda_2,\lambda_3,\lambda_4)$ \\ 
			& $5$ & \multicolumn{4}{l}{$( -\lambda_1-\lambda_2-\lambda_3-\lambda_4, 
			\lambda_1+\lambda_2-\lambda_3-\lambda_4,
			\lambda_1-\lambda_2+\lambda_3-\lambda_4,
			\lambda_1-\lambda_2-\lambda_3+\lambda_4)/2$}\\
			& $6$ & \multicolumn{4}{l}{$(-\lambda_1+\lambda_2+\lambda_3-\lambda_4,
			-\lambda_1+\lambda_2-\lambda_3+\lambda_4,
			-\lambda_1-\lambda_2+\lambda_3+\lambda_4,
			\lambda_1+\lambda_2+\lambda_3+\lambda_4)/2$}\\
			& $7$ & \multicolumn{4}{l}{$( -\lambda_1+\lambda_2-\lambda_3+\lambda_4,
			-\lambda_1+\lambda_2+\lambda_3-\lambda_4,
			\lambda_1+\lambda_2+\lambda_3+\lambda_4,
			-\lambda_1-\lambda_2+\lambda_3+\lambda_4)/2$}\\                                
			& $8$ & \multicolumn{4}{l}{$( -\lambda_1-\lambda_2+\lambda_3+\lambda_4,
			\lambda_1+\lambda_2+\lambda_3+\lambda_4,
			-\lambda_1+\lambda_2+\lambda_3-\lambda_4,
			-\lambda_1+\lambda_2-\lambda_3+\lambda_4)/2$}\\
			& $9,\ldots,12$ & $(-\lambda_1,\lambda_2,\lambda_3,\lambda_4)$
                         & $(\lambda_1,\lambda_2,\lambda_3,-\lambda_4)$  
		         & $(\lambda_1,\lambda_2,-\lambda_3,\lambda_4)$ 
		        & $(\lambda_1,-\lambda_2,\lambda_3,\lambda_4)$ \\\cline{2-2}
			& \multicolumn{1}{l}{\bf Coefficients:} & \multicolumn{4}{l}{with respect to the basis $\{u_1,u_2,u_3,u_4\}$}\\
\end{tabular}}
	\\[1ex]\caption{{\bf Case $\pmb{i=1}$:} The table lists real orbit representatives corresponding to $\gamma_j$ and $[z_k]\in H^1(Z)$. 
	 } \label{tabreps}
\end{table}

        
\newpage
\subsection{Mixed elements: Centralisers}\label{appMEc} 
~{}
\begin{table}[H]\small  \renewcommand\arraystretch{1.4}
	\hspace*{0cm}\begin{tabular}[H]{ccccc}
		{\bf $\pmb{i}$}  & {\bf $\pmb{r}$} & {\bf identity component $\pmb{ Z^\circ}$} & {\bf preimages of generators of $\pmb{Z/Z^\circ}$} & \textbf{$\pmb{H^1}$ ref.} \\\hline
		2 & 1 & 1 & $ (-I,-I,I,I),(-I,I,-I,I),(-L,-L,L,L)$ &  \eqref{H1Z} \\\hline
		
		3 & 1 & 1 & $ (-I,-I,I,I),(-I,I,-I,I),(-L,-L,L,L)$  & \eqref{H1Z} \\\hdashline

		3 & 2 & 1 & $ (-I,-I,I,I),(-I,I,-I,I),(-I,-I,-I,-I)$ & \eqref{H1Z32} \\\hline
		
		4 & 1 & 1 & $ (-I,-I,I,I),(-I,I,-I,I),(-L,-L,L,L)$  & \eqref{H1Z} \\\hdashline

		4 & 2 & 1 & $ (-I,-I,I,I),(-I,I,-I,I),(-L,-L,L,L)$  & \eqref{H1Z} \\\hdashline

		4 & 3 & $\left\{ (D(a)^{-1}, D(a), D(a)^{-1}, D(a)) \;:\; a\in \C^\times \right\}$ & $ (-I,-I,I,I),(-I,I,-I,I)$  & \eqref{H1Z43} \\\hdashline

		4 & 4 & $\left\{ (D(a), D(a)^{-1}, D(a)^{-1}, D(a)) \;:\; a\in \C^\times \right\}$ & $ (-I,-I,I,I),(-I,I,-I,I)$  & \eqref{H1Z43} \\
\end{tabular}\caption{Centralisers $Z=Z_{\wG}(p,h,e,f)$ for each $p\in\Sigma_i$ and $e=n_{i,r}$ as in Theorem~\ref{thmME}. The last column lists reference labels for the equations describing the representatives of the classes in $H^1(Z)$; the notation used in the third and fourth columns is from \eqref{eqmats}.
	}\label{tabZME}   
\end{table}

\newpage

\subsection{Mixed elements, case $p\in\Sigma$}\label{appMES} 
~{}
	\begin{table}[H]
		\centering\small
		\begin{tabular}{c|c|c|p{5.5cm}|p{5.5cm}}
			$\pmb{i}$&$\pmb{r}$&$\pmb{k}$  & {\bf semisimple part } & {\bf nilpotent part }\\
			\hline

			$2$&$1$& $1$ &  $(\lambda_1, \lambda_2, \lambda_3, 0)$ & $\mye{0011}$\\ 
			&&$2$ & $( -\lambda_1, \lambda_2, \lambda_3, 0 )$ & $-\mye{0011}$\\
			&&$3$ & $( -\lambda_1, \lambda_2, -\lambda_3, 0)$ & $\mye{0011}$\\ 
			&&$4$ & $(-\lambda_1, -\lambda_2, \lambda_3, 0 )$ & $\mye{0011}$\\ 
			&&$5$ & $(-\lambda_1, \lambda_2, \lambda_3, 0)$ & $\mye{0011}$\\ 		
			&&$6$ & $(\lambda_1, \lambda_2, \lambda_3, 0)$ & $-\mye{0011}$\\ 
			&&$7$ & $(\lambda_1, \lambda_2, -\lambda_3, 0)$ & $\mye{0011}$\\ 
			&&$8$ &$(\lambda_1, -\lambda_2, \lambda_3, 0)$ & $\mye{0011}$ \\ \cline{4-5}
			\multicolumn{5}{l}{ {\bf Coefficients:} with respect to the basis $\{u_1,u_2,u_3,u_4\}$}\\

			\multicolumn{5}{l}{ {\bf Conditions:} $(\lambda_1,\lambda_2,\lambda_3)$ up to $\Gamma_{\Pi_2}$-conjugacy, each $\lambda_i\in\R^\times$ and $\lambda_1\notin \{\pm\lambda_2\pm\lambda_3\}$ } \\\hline	
			\end{tabular}\caption{ {\bf Case $\pmb{i=2}$ and $\pmb{p\in \Sigma}$:} Real representatives corresponding to $[z_k]\in H^1(Z)$ and $n_{i,r}$.  }		
		\label{tabrepsME1}
	\end{table}

	\begin{table}[H]
		\centering\small
		\begin{tabular}{c|c|c|p{5.5cm}|p{5.5cm}}
			$\pmb{i}$&$\pmb{r}$&$\pmb{k}$  & {\bf semisimple part } & {\bf nilpotent part }\\
			\hline

			$3$&$1$& $1$ &  $(\lambda_1+\lambda_2, -\lambda_1, -\lambda_2, 0)$ & $\mye{0011}$\\ 
			&& $2$ & $(-\lambda_1-\lambda_2, -\lambda_1, -\lambda_2, 0)$ & $-\mye{0011}$\\
			&& $3$ & $(-\lambda_1-\lambda_2, -\lambda_1, \lambda_2, 0)$ & $\mye{0011}$\\ 
			&& $4$ & $(-\lambda_1-\lambda_2, \lambda_1, -\lambda_2, 0)$ & $\mye{0011}$\\ 
			&& $5$ & $(-\lambda_1-\lambda_2, -\lambda_1, -\lambda_2, 0)$ & $\mye{0011}$\\ 		
			&& 	$6$ & $(\lambda_1+\lambda_2, -\lambda_1, -\lambda_2, 0)$ & $-\mye{0011}$\\ 
			&& $7$ & $(\lambda_1+\lambda_2, -\lambda_1, \lambda_2, 0)$ & $\mye{0011}$\\
\hline
			
			\end{tabular}\caption{ {\bf Case $\pmb{i=3}$ and $\pmb{p\in \Sigma}$:} Real representatives corresponding to $[z_k]\in H^1(Z)$ and  $n_{i,r}$.  }		
		\label{tabrepsME1b}
	\end{table}

	\begin{table}[H]
	\centering\small
	\begin{tabular}{c|c|c|p{5.5cm}|p{5.5cm}}
		$\pmb{i}$&$\pmb{r}$&$\pmb{k}$  & {\bf semisimple part } & {\bf nilpotent part }\\
		\hline
			
			$4$	& $1$ &	$1$ &  $(\lambda_1, 0, 0, \lambda_4)$ 
			& $\mye{1010}+\mye{0110}$\\ 
			& &	$2$ & $(-\lambda_1, 0, 0, -\lambda_4)$
			& $\mye{1010}+\mye{0110}$\\ 
			& &	$3$ & $(-\lambda_1, 0, 0, \lambda_4)$
			&  $-\mye{1010}+\mye{0110}$     \\  	
			& &	$4$ & $(-\lambda_1, 0, 0, \lambda_4)$
			&   $\mye{1010}-\mye{0110}$    \\  
			& &	$5$ &$(-\lambda_1, 0, 0, \lambda_4)$
			&    $\mye{1010}+\mye{0110}$   \\    
			& &	$6$ & $(\lambda_1, 0, 0, -\lambda_4 )$
			&    $\mye{1010}+\mye{0110}$   \\  
			& &	$7$ &$(\lambda_1, 0, 0, \lambda_4)$
			&    $-\mye{1010}+\mye{0110}$   \\
\hline
		\end{tabular}\caption{
		{\bf Case $\pmb{i=4}$ and $\pmb{p\in \Sigma}$:}
		 Real representatives corresponding to $[z_k]\in H^1(Z)$ and $n_{i,r}$.  }		
		\label{tabrepsME2}
	\end{table}

	\begin{table}[H]\resizebox{\columnwidth-1.5cm}{!}{%
		\centering
		\begin{tabular}{c|c|c|p{3cm}|p{11cm}}
			$\pmb{i}$	& $\pmb{r}$	&	$\pmb{k}$  & {\bf semisimple part } & {\bf nilpotent part }\\
			\hline
			
			$7$	& $1$ &	$1$ & $\lambda_1( 1, 0, 0, -1 )$ 
			&  $\mye{1101}+\mye{1011}+\mye{1000}+\mye{0001}$ \\ 
			& &	$2$ & $\lambda_1(-1, 0, 0, 1)$
			& $-\mye{1101}+\mye{1011}+\mye{1000}-\mye{0001}
			$\\ 
			& &	$3$ & $\lambda_1(-1, 0, 0, -1)$
			&  $\mye{1101}-\mye{1011}+\mye{1000}-\mye{0001}$     \\  	
			& &	$4$ & $\lambda_1(-1, 0, 0, -1)$
			&   $-\mye{1101}-\mye{1011}+\mye{1000}+\mye{0001}$    \\  
			& &	$5$ & $\lambda_1(-1, 0, 0, -1)$
			&    $\mye{1101}+\mye{1011}-\mye{1000}-\mye{0001}$   \\    
			& &	$6$ & $\lambda_1(-1, 0, 0, -1)$
			&    $-\mye{1101}+\mye{1011}-\mye{1000}+\mye{0001}$   \\  
			& &	$7$ & $\lambda_1(-1, 0, 0, 1)$
			&    $\mye{1101}-\mye{1011}-\mye{1000}+\mye{0001}$   \\
\hline

		\end{tabular}}\caption{{\bf Case $\pmb{i=7}$ and $\pmb{p\in \Sigma}$:} Real representatives corresponding to $[z_k]\in H^1(Z)$ and $n_{i,r}$. }		
		\label{tabrepsME3}
	\end{table}

	\begin{table}[H]
		\centering
		\begin{tabular}{c|c|c|p{5cm}|p{5cm}}
			$\pmb{i}$	& $\pmb{r}$	&	$\pmb{k}$  & {\bf  semisimple part } & {\bf nilpotent part }\\
			\hline
			
			$10$	& $1$ &	$1$ & $(\lambda_1, 0, 0, 0 )$ 
			& $\mye{1100}+\mye{1010}+\mye{0110}$\\ 
			& &	$2$ & $(-\lambda_1, 0, 0, 0 )$
			& $-\mye{1100}+\mye{1010}+\mye{0110}$\\ 
			& &	$3$ & $(-\lambda_1, 0, 0, 0 )$
			&  $\mye{1100}-\mye{1010}+\mye{0110}$     \\  	
			& &	$4$ & $(-\lambda_1, 0, 0, 0 )$
			&   $\mye{1100}+\mye{1010}-\mye{0110}$    \\  
			& &	$5$ & $(-\lambda_1, 0, 0, 0 )$
			&    $\mye{1100}+\mye{1010}+\mye{0110}$   \\    
			& &	$6$ & $(\lambda_1, 0, 0, 0 )$
			&    $-\mye{1100}+\mye{1010}+\mye{0110}$   \\  
			& &	$7$ & $(\lambda_1, 0, 0, 0 )$
			&    $\mye{1100}-\mye{1010}+\mye{0110}$   \\
\hline
			
		\end{tabular}\caption{{\bf Case $\pmb{i=10}$ and $\pmb{p\in \Sigma}$ (Part I):} Real representatives corresponding to $[z_k]\in H^1(Z)$ and  $n_{i,r}$. }		
		\label{tabrepsME4}
	\end{table}

	\begin{table}[H]
		\centering
		\begin{tabular}{c|c|c|c|c}
			$\pmb{i}$	& $\pmb{r}$	&	$\pmb{k}$  & {\bf semisimple part } & {\bf nilpotent part }\\
			\hline
			
			$10$	& $2$ &	$1$ & $(\lambda_1, 0, 0, 0 )$ 
			& $\mye{1010}+\mye{0110}$\\ 
			& &	$2$ & $(-\lambda_1, 0, 0, 0 )$ 
			& $\mye{1010}+\mye{0110}$\\ 
			& &	$3$ & $(-\lambda_1, 0, 0, 0 )$ 
			&  $-\mye{1010}+\mye{0110}$     \\  	
			& &	$4$ & $(-\lambda_1, 0, 0, 0 )$ 
			&   $\mye{1010}-\mye{0110}$    \\ \hline
			
			$10$	& $4$ &	$1$ & $(\lambda_1, 0, 0, 0 )$ 
			& $\mye{1100}+\mye{0110}$\\ 
			& &	$2$ & $(-\lambda_1, 0, 0, 0 )$ 
			& $-\mye{1100}+\mye{0110}$\\ 
			& &	$3$ & $(-\lambda_1, 0, 0, 0 )$ 
			&  $\mye{1100}+\mye{0110}$     \\
\hline

		\end{tabular}\caption{{\bf Case $\pmb{i=10}$ and $\pmb{p\in \Sigma}$ (Part II):} Real representatives corresponding to $[z_k]\in H^1(Z)$ and  $n_{i,r}$. }		
		\label{tabrepsME5}
	\end{table}

\newpage 
\subsection{Mixed elements, case $p\notin\Sigma$}\label{appMENS} 
~{}

\begin{table}[H]\renewcommand\arraystretch{1.1}
  \scalebox{1}{\begin{tabular}{lll}
    $\pmb{i}$ & $\pmb{j}$ & {\bf nilpotent elements} \\\hline  
			$2$& $2$&  $n_{2,2,1}=-\mye{0110}$\\
			     & $4,6,7,8$ & $n_{2,j,1}=\mye{0011}$\\
			     & $3,5$ & $n_{2,j,1}=\mye{0000}$\\ \hline

			$3$& $2$& $n_{3,2,1}=-\mye{0011}$,\;\; $n_{3,2,2}=-\mye{0111}+\mye{1011}-\mye{0010}-\mye{0001}$\\\hline

			$4$& $2,3$& $n_{4,j,1}=\mye{0110}+\mye{1010}$,\;\; $n_{4,j,2}=\mye{0110}+\mye{0101}$,\;\; $n_{4,j,3}=\mye{0110}$,\; $n_{4,j,4}=\mye{0101}$\\
			& $4$& $n_{4,4,1}=-\frac12(-\mye{1110}-\mye{1101}+\mye{1010}+\mye{1001}+\mye{0110}+\mye{0101}-\mye{0010}-\mye{0001})$\\\hline
			 		 
			$7$& $2$& $n_{7,2,1}=-\mye{1101}-\mye{1011}-\mye{1000}+\mye{0001}$,\;\; $n_{7,2,2}=-\mye{1101}-\mye{1010}+\mye{0001}$,\\
			&&$n_{7,2,3}=-\mye{1011}-\mye{1000}+\mye{0101}$,\;\; $n_{7,2,4}=-\mye{1011}-\mye{1000}$,\\ 
			&& $n_{7,2,5}=-\mye{1101}+\mye{0001}$,\;\; $n_{7,2,6}=-\mye{1001}$\\\hline

			$10$& $2$& $n_{10,2,1}=-\mye{1100}-\mye{1010}+\mye{0110}$,\;\; $n_{10,2,2}=-\mye{1010}+\mye{0110}$,\\ &&$n_{10,2,3}=-\mye{1010}+\mye{0110}+\mye{0011}$,\;\;
			$n_{10,2,4}=-\mye{1100}+\mye{0110}$,\;\;  $n_{10,2,5}=\mye{0110}$,\\  
			&& $n_{10,2,6}=\mye{0110}+\mye{0011}$,\;\;
			$n_{10,2,7}=-\mye{1100}+\mye{0110}+\mye{0101}$,\;\;
			$n_{10,2,8}=\mye{0110}+\mye{0101}$,\\ 
			&& $n_{10,2,9}=\mye{0110}-\mye{0101}-\mye{0011}$,\;\; 
			 $n_{10,2,10}=-\mye{1100}-\mye{1010}$,\;\;
			 $n_{10,2,11}=-\mye{1010}$, \\
			&&$n_{10,2,12}=-\mye{1010}+\mye{0011}$,\;\; $n_{10,2,13}=\mye{0011}$.\\[1ex]
			
	\end{tabular}}\caption{The nilpotent elements $n_{i,j,r}$ used in the classification of mixed  elements with semisimple part $p'\notin\Sigma$, sorted by  Case $i$  and cohomology class $[\gamma_j]$.  }
\label{tabNilpElements}
\end{table}

	\begin{table}[H]
	\centering
	\begin{tabular}{c|c|c|p{8.5cm}|p{5.5cm}}
		$\pmb{i}$&$\pmb{j}$&$\pmb{k}$  & {\bf semisimple part } & {\bf nilpotent part }\\
		\hline

		$2$&$2$& $1$ & $(-\imath\lambda_3, 0, \lambda_1, -\imath\lambda_2)$ & $-\mye{0110}$\\ 
		&&$2$ & $(\imath\lambda_3, 0, \lambda_1, \imath\lambda_2)$ & $-\mye{0110}$\\
		&&$3$ & $(\imath\lambda_3, 0, -\lambda_1, -\imath\lambda_2)$ & $-\mye{0110}$\\ 
		&&$4$ & $(\imath\lambda_3, 0, \lambda_1, -\imath\lambda_2)$ & $\mye{0110}$\\ 
		&&$5$ & $(\imath\lambda_3, 0, \lambda_1, -\imath\lambda_2)$ & $-\mye{0110}$\\ 		
		&&$6$ & $(-\imath\lambda_3, 0, \lambda_1, \imath\lambda_2)$ & $-\mye{0110}$\\ 
		&&$7$ & $(-\imath\lambda_3, 0, -\lambda_1, -\imath\lambda_2)$ & $-\mye{0110}$\\
\hline
	\end{tabular}
	\caption{ {\bf Cases $\pmb{i=2}$ and $\pmb{p\notin S}$ (Part I):} Mixed real representatives corresponding to $\gamma_j$, $[z_k]\in H^1(Z)$, and  $n_{2,j,1}$.}		
	\label{tabrepsME_NS_2_1}
    \end{table}

	\begin{table}[H]
	\centering
	\begin{tabular}{c|c|c|p{8.5cm}|p{5.5cm}}
		$\pmb{i}$&$\pmb{j}$&$\pmb{k}$  & {\bf semisimple part } & {\bf nilpotent part }\\
		\hline
		
		$2$&$5$& $1$ & $(0, \imath\lambda_3, -\lambda_2, \lambda_1)$ & $\mye{0000}$\\
		&&$2$ & $(0, \imath\lambda_3, -\lambda_2, -\lambda_1)$ & $-\mye{0000}$\\
		&&$3$ & $(0, \imath\lambda_3, \lambda_2, \lambda_1)$ & $-\mye{0000}$\\
		&&$4$ & $(0, -\imath\lambda_3, -\lambda_2, \lambda_1)$ & $-\mye{0000}$\\
		&&$5$ & $(0, \imath\lambda_3, -\lambda_2, \lambda_1)$ & $-\mye{0000}$\\
		&&$6$ & $(0, \imath\lambda_3, -\lambda_2, -\lambda_1)$ & $\mye{0000}$\\
		&&$7$ & $(0, \imath\lambda_3, \lambda_2, \lambda_1)$ & $\mye{0000}$\\
\hline 	

	\end{tabular}	
	\caption{ {\bf Cases $\pmb{i=2}$ and $\pmb{\notin S}$ (Part II):} Mixed real representatives corresponding to $\gamma_j$, $[z_k]\in H^1(Z)$, and $n_{2,j,1}$.}		
	\label{tabrepsME_NS_2_2}
	\end{table}

	\begin{table}[H]
	\centering
	\begin{tabular}{c|c|c|p{8.5cm}|p{5.5cm}}
		$\pmb{i}$&$\pmb{r}$&$\pmb{k}$  & {\bf semisimple part } & {\bf nilpotent part }\\
		\hline

		$3$&$1$& $1$ &  $(\imath\lambda_1+\imath\lambda_2, -\imath\lambda_1, -\imath\lambda_2, 0)$ & $\mye{0011}$\\ 
		&& $2$ & $(-\imath\lambda_1-\imath\lambda_2, -\imath\lambda_1, -\imath\lambda_2, 0)$ & $-\mye{0011}$\\
		&& $3$ & $(-\imath\lambda_1-\imath\lambda_2, -\imath\lambda_1, \imath\lambda_2, 0)$ & $\mye{0011}$\\ 
		&& $4$ & $(-\imath\lambda_1-\imath\lambda_2, \imath\lambda_1, -\imath\lambda_2, 0) $ & $\mye{0011}$\\ 
		&& $5$ & $(-\imath\lambda_1-\imath\lambda_2, -\imath\lambda_1, -\imath\lambda_2, 0)$ & $\mye{0011}$\\ 		
		&& 	$6$ & $(\imath\lambda_1+\imath\lambda_2, -\imath\lambda_1, -\imath\lambda_2, 0)$ & $-\mye{0011}$\\ 
		&& $7$ & $(\imath\lambda_1+\imath\lambda_2, -\imath\lambda_1, \imath\lambda_2, 0)$ & $\mye{0011}$\\
\hline

	\end{tabular}
	\caption{ {\bf Cases $\pmb{i=3}$ and $\pmb{p\notin S}$}: Mixed real representatives corresponding to $\gamma_2$, $[z_k]\in H^1(Z)$, and $n_{3,2,r}$.}		
\label{tabrepsME_NS_3}		
\end{table}

	\begin{table}[H]\resizebox{\columnwidth}{!}{
	\centering
	\begin{tabular}{c|c|c|p{4.5cm}|p{11.5cm}}
		$\pmb{i}$&$\pmb{r}$&$\pmb{k}$  & {\bf semisimple part } & {\bf nilpotent part }\\
		\hline

		$7$&$1$& $1$ &  $(\imath\lambda, 0, 0, -\imath\lambda)$& $-\mye{1101}-\mye{1011}-\mye{1000}+\mye{0001}$\\ 
		&& $2$  &  $(-\imath\lambda, 0, 0, \imath\lambda)$ & $\mye{1101}-\mye{1011}-\mye{1000}-\mye{0001}$\\
		&& $3$  &  $(-\imath\lambda, 0, 0, -\imath\lambda)$ & $-\mye{1101}+\mye{1011}-\mye{1000}-\mye{0001}$\\ 
		&& $4$  &  $(-\imath\lambda, 0, 0, -\imath\lambda)$ & $\mye{1101}+\mye{1011}-\mye{1000}+\mye{0001}$\\
		&& $5$  &  $(-\imath\lambda, 0, 0, -\imath\lambda)$ & $-\mye{1101}-\mye{1011}+\mye{1000}-\mye{0001}$\\		
		&& $6$  &  $(-\imath\lambda, 0, 0, -\imath\lambda)$ & $\mye{1101}-\mye{1011}+\mye{1000}+\mye{0001}$\\
		&& $7$  &  $(-\imath\lambda, 0, 0, \imath\lambda)$ & $-\mye{1101}+\mye{1011}+\mye{1000}+\mye{0001}$\\
\hline
	
	\end{tabular}} 
	\caption{ {\bf Cases $\pmb{i=7}$ and $\pmb{p\notin S}$}: Mixed real representatives corresponding to $\gamma_2$, $[z_k]\in H^1(Z)$, and  $n_{7,2,r}$.}		
\label{tabrepsME_NS_7}		
\end{table}

	\begin{table}[H]
	\centering
	\begin{tabular}{c|c|c|c|c}
		$\pmb{i}$&$\pmb{r}$&$\pmb{k}$  & {\bf semisimple part } & {\bf nilpotent part }\\
		\hline

		$10$& $1$& $1$ & $(\imath\lambda, 0, 0, 0)$ & $-\mye{1100}-\mye{1010}+\mye{0110}$\\ 
		&$3$&&& $-\mye{1010}+\mye{0110}+\mye{0011}$\\
		&$7$&&& $-\mye{1100}+\mye{0110}+\mye{0101}$\\
		&$9$&&& $\mye{0110}-\mye{0101}-\mye{0011}$\\ \cdashline{2-5}
		
		&$1$& $2$ &  $(-\imath\lambda, 0, 0, 0)$ & $\mye{1100}-\mye{1010}+\mye{0110}$\\ 
		&$3$&&& $-\mye{1010}+\mye{0110}-\mye{0011}$\\
		&$7$&&& $\mye{1100}+\mye{0110}+\mye{0101}$\\
\hline	
        \end{tabular}
	\caption{ {\bf Cases $\pmb{i=10}$ and $\pmb{p\notin S}$ (Part I):} Mixed real representatives corresponding to $\gamma_2$, $[z_k]\in H^1(Z)$, and $n_{10,2,r}$.}		
	\label{tabrepsME_NS_10_1}			
	\end{table}

	\begin{table}[H]
	\centering
	\begin{tabular}{c|c|c|c|c}
		$\pmb{i}$&$\pmb{r}$&$\pmb{k}$  & {\bf semisimple part } & {\bf nilpotent part }\\
		\hline
		$10$& $2,4,6,8,10,12$& $1$ & $(\imath\lambda, 0, 0, 0)$ & $-\mye{1010}+\mye{0110}$ \\ 
		&&&& $-\mye{1100}+\mye{0110}$ \\
		&&&& $\mye{0110}+\mye{0011}$ \\
		&&&& $\mye{0110}+\mye{0101}$ \\
		&&&& $-\mye{1100}-\mye{1010}$ \\
		&&&& $-\mye{1010}+\mye{0011}$ \\ \cdashline{3-5}
		
		&& $2$ &  $(-\imath\lambda, 0, 0, 0)$ & $-\mye{1010}+\mye{0110}$\\ 
		&&&& $\mye{1100}+\mye{0110}$ \\
		&&&& $\mye{0110}-\mye{0011}$ \\
		&&&& $\mye{0110}+\mye{0101}$ \\
		&&&& $\mye{1100}-\mye{1010}$ \\
		&&&& $-\mye{1010}-\mye{0011}$ \\ \cdashline{3-5}
		
		&& $3$ &  $(-\imath\lambda, 0, 0, 0)$ & $\mye{1010}+\mye{0110}$\\ 
		&&&& $-\mye{1100}+\mye{0110}$ \\
		&&&& $\mye{0110}+\mye{0011}$ \\
		&&&& $\mye{0110}-\mye{0101}$ \\
		&&&& $-\mye{1100}+\mye{1010}$ \\
		&&&& $\mye{1010}+\mye{0011}$ \\ \cdashline{3-5}
		
		&& $4$ &  $(-\imath\lambda, 0, 0, 0)$ & $-\mye{1010}-\mye{0110}$\\ 
		&&&& $-\mye{1100}-\mye{0110}$ \\
		&&&& $-\mye{0110}+\mye{0011}$ \\
		&&&& $-\mye{0110}+\mye{0101}$ \\
		&&&& $-\mye{1100}-\mye{1010}$ \\
		&&&& $-\mye{1010}+\mye{0011}$ \\
		\hline

		$10$& $5,11$& $1$ & $(\imath\lambda, 0, 0, 0)$ & $\mye{0110}$ \\ 
		&&&& $-\mye{1010}$ \\ \cdashline{3-5}
		&& $2$ & $(-\imath\lambda, 0, 0, 0)$ & $\mye{0110}$   \\ 
		&&&& $-\mye{1010}$ \\
		\hline
		
		$10$& $13$& $1$ & $(\imath\lambda, 0, 0, 0)$ & $\mye{0011}$ \\  
		&& $2$ &   $(-\imath\lambda, 0, 0, 0)$ & $\mye{0011}$ \\  
		\hline		
		\multicolumn{4}{l}{ {\bf Coefficients:} with respect to the basis $\{v_1,v_2,v_3,v_4\}$}\\
		\multicolumn{4}{l}{ {\bf Conditions:}   $\lambda$ up to $\Gamma_{\Pi_{10}}$-conjugacy, $\lambda\in\imath\R^\times$ }\\
		\hline	
\end{tabular}
	\caption{ {\bf Cases $\pmb{i=10}$ and $\pmb{p\notin S}$ (Part II):} Mixed real representatives corresponding to $\gamma_2$, $[z_k]\in H^1(Z)$, and $n_{10,2,r}$.}		
\label{tabrepsME_NS_10_2}	
\end{table}

	\begin{table}[H]
	\centering
	\begin{tabular}{c|c|c|c|p{4.5cm}|p{10.5cm}}
		$\pmb{i}$&$\pmb{j}$&$\pmb{r}$ & $\pmb{k}$  & {\bf  semisimple part } & {\bf nilpotent part }\\
		\hline
		
		$4$& $2$ & $1$ & $1$ & $(\imath\lambda_1, 0, 0, \lambda_4)$ & $\mye{1010}+\mye{0110}$\\ 
		&&& $2$ &   $(-\imath\lambda_1, 0, 0, -\lambda_4)$ & $\mye{1010}+\mye{0110}$\\  
		&&& $3$ &  $(-\imath\lambda_1, 0, 0, \lambda_4)$ & $-\mye{1010}+\mye{0110}$\\  
		&&& $4$ &  $(-\imath\lambda_1, 0, 0, \lambda_4)$ & $\mye{1010}-\mye{0110}$\\  
		&&& $5$ &  $(-\imath\lambda_1, 0, 0, \lambda_4)$ & $\mye{1010}+\mye{0110}$\\  
		&&& $6$ &  $(\imath\lambda_1, 0, 0, -\lambda_4)$ & $\mye{1010}+\mye{0110}$\\  
		&&& $7$ &  $(\imath\lambda_1, 0, 0, \lambda_4)$ & $-\mye{1010}+\mye{0110}$\\ 
		&&& $8$ &  $(\imath\lambda_1, 0, 0, \lambda_4)$ & $\mye{1010}-\mye{0110}$\\  
		\hline
		
		$4$& $2$ & $2$ & $1$ &  & $\mye{0110}+\mye{0101}$\\ 
		&&& $2$ &  & $\mye{0110}+\mye{0101}$\\
		&&& $3$ &  & $\mye{0110}-\mye{0101}$\\ 
		&&& $4$ &  & $-\mye{0110}+\mye{0101}$\\
		&&& $5$ & same as $r=1$  & $\mye{0110}+\mye{0101}$\\ 
		&&& $6$ &  & $\mye{0110}+\mye{0101}$\\
		&&& $7$ &  & $\mye{0110}-\mye{0101}$\\
		&&& $8$ &  & $-\mye{0110}+\mye{0101}$\\
		\hline
				
		$4$& $2$ & $3$ & $1$ & $(\imath\lambda_1, 0, 0, \lambda_4)$ & $\mye{0110}$\\ 
		&&& $2$ &   $(-\imath\lambda_1, 0, 0, -\lambda_4)$ & $\mye{0110}$\\  
		&&& $3$ &  $(-\imath\lambda_1, 0, 0, \lambda_4)$ & $\mye{0110}$\\  
		&&& $4$ &  $(-\imath\lambda_1, 0, 0, \lambda_4)$ & $-\mye{0110}$\\  
		\hline
		
		$4$& $2$ & $4$ & $1$ &  & $\mye{0101}$ \\ 
		&&& $2$ &  & $\mye{0101}$ \\ 
		&&& $3$ & same as $r=3$ & $-\mye{0101}$ \\ 
		&&& $4$ & & $\mye{0101}$ \\ 
		\hline
		\multicolumn{6}{l}{ {\bf Coefficients:}  with respect to the basis $\{t_1,t_4\}$}\\
		\multicolumn{6}{l}{ {\bf Conditions:}  $(\lambda_1,\lambda_4)$ up to $\Gamma_{\Pi_4}$-conjugacy, $\lambda_4\in\R^\times$, $\lambda_1\in\imath\R^\times$}  \\ \hline
		
	\end{tabular}
	\caption{ {\bf Cases $\pmb{i=4}$ and $\pmb{p\notin S}$ (Part I):} Mixed real representatives corresponding to $\gamma_j$, $[z_k]\in H^1(Z)$, and  $n_{4,j,r}$.}		
	\label{tabrepsME_NS_4_1}	        
	\end{table}

	\begin{table}[H]
	\centering
	\begin{tabular}{c|c|c|c|p{7.5cm}|p{6.5cm}}
	$\pmb{i}$&$\pmb{j}$&$\pmb{r}$ & $\pmb{k}$  & {\bf semisimple part } & {\bf nilpotent part }\\
	\hline
			
		$4$& $3$ & $1$ & $1$ & $(-\imath\lambda_1, 0, 0, \imath\lambda_4)$ & $\mye{1010}+\mye{0110}$\\ 
		&&& $2$ &   $(\imath\lambda_1, 0, 0, -\imath\lambda_4)$ & $\mye{1010}+\mye{0110}$\\ 
		&&& $3$ &   $(\imath\lambda_1, 0, 0, \imath\lambda_4)$ & $-\mye{1010}+\mye{0110}$\\ 
		&&& $4$ &   $(\imath\lambda_1, 0, 0, \imath\lambda_4)$ & $\mye{1010}-\mye{0110}$\\ 
		&&& $5$ &   $(\imath\lambda_1, 0, 0, \imath\lambda_4)$ & $\mye{1010}+\mye{0110}$\\ 
		&&& $6$ &   $(-\imath\lambda_1, 0, 0, -\imath\lambda_4)$ & $\mye{1010}+\mye{0110}$\\ 
		&&& $7$ &   $(-\imath\lambda_1, 0, 0, \imath\lambda_4)$ & $-\mye{1010}+\mye{0110}$\\
\hline
		
	\end{tabular}
	\caption{ {\bf Cases $\pmb{i=4}$ and $\pmb{p\notin S}$ (Part II):} Mixed real representatives corresponding to $\gamma_j$, $[z_k]\in H^1(Z)$, and $n_{4,j,r}$.}		
    \label{tabrepsME_NS_4_2}	        
    \end{table}

\newpage



\begin{thebibliography}{99}

\bibitem{real} A.\ Ac\'{\i}n, A\. Andrianov, L.\ Costa, E. Jan\'{e}, J.I.\
Latorre, R.\ Tarrach. Generalized Schmidt Decomposition and
Classification of Three-Quantum-Bit States. Phys.\ Rev.\ Lett.\ \textbf{85},
1560 (2000).

\bibitem{rebit-2}  A.\ Aleksandrova, V.\ Borish, W.K.\ Wootters. 
Real-vector-space quantum theory with a universal quantum bit. Phys.\ Rev.\ \textbf{A87}, 052106 (2013).

\bibitem{ADFMT} L.\ Andrianopoli, R.\ D'Auria, S.\ Ferrara, A.\ Marrani, M.\  Trigiante. Two-Centered Magical Charge Orbits. JHEP {\bf 2011}, 41 (2011).

\bibitem{two-rebit} J.\ Batle, A.R.\ Plastino, M.\ Casas, A.\ Plastino.
Understanding Quantum Entanglement: Qubits, Rebits and the Quaternionic
Approach. Optics and Spectroscopy \textbf{94} (5), 700-705 (2003).

\bibitem{STU2} K.\ Behrndt, R.\ Kallosh, J.\ Rahmfeld, M.\ Shmakova, W.\ K.\ Wong.
STU Black Holes and String Triality. Phys.\ Rev.\ {\bf D54}, 6293 (1996).

\bibitem{berhuy} G.\ Berhuy.
An Introduction to Galois Cohomology and its Applications. 
London Mathematical Society Lecture Note Series. Cambridge: Cambridge University Press (2010).

\bibitem{trivectors} M.\ Borovoi, W.A.\ de Graaf, H.\ V\^an L\^e.
Real graded Lie algebras, Galois cohomology and classification of trivectors in $\mathbb R^9$, \url{arXiv:2106.00246v1}.

\bibitem{trivectors2} M.\ Borovoi, W.A.\ de Graaf, H.\ V\^an L\^e.
Classification of real trivectors in dimension nine
\url{arXiv:2108.00790v1}.

\bibitem{Borsten} L.\ Borsten, D.\ Dahanayake, M.J.\ Duff, W.\ Rubens, H.\ Ebrahim. Freudenthal triple classification of three-qubit entanglement. Phys.\ Rev.\ \textbf{A80} 032326 (2009).

\bibitem{nilp} 
  L.\ Borsten, D.\ Dahanayake, M.J.\ Duff, A.\ Marrani, W.\ Rubens. Four-qubit entanglement classification from string theory. Phys.\ Rev.\ Lett.\ {\bf 105} 100507, 4 (2010).

\bibitem{Oct} L.\ Borsten, D.\ Dahanayake, M.J.\ Duff, H.\ Ebrahim, W.\ Rubens.
Black Holes, Qubits and Octonions. Phys.\ Rept.\ \textbf{471}, 113-219 (2009).

\bibitem{BH-1} L.\ Borsten, M.J.\ Duff, A.\ Marrani, W.\ Rubens. {On the
  Black-Hole/Qubit Correspondence}. Eur.\ Phys.\ J.\ Plus \textbf{126}, 37 (2011).

\bibitem{BH-2} L.\ Borsten, M.J.\ Duff, P.\ L\'{e}vay. The
black-hole/qubit correspondence: an up-to-date review. Class.\ Quant.\ Grav.\ \textbf{29}  224008 (2012).

\bibitem{15} G.\ Bossard. Octonionic black holes. JHEP {\bf 05} 113 (2012).

\bibitem{20} G.\ Bossard, C.\ Ruef. Interacting non-BPS black holes. Gen.\ Rel.\ Grav.\ {\bf 44}, 21 (2012).

\bibitem{real-QM} C.M.\ Caves, C.A.\ Fuchs, P.\ Rungta. {Entanglement of
formation of an arbitrary state of two rebits}. Found.\ Phys.\ Lett. \textbf{14}, 199-212 (2001).

\bibitem{Cayley} A.\ Cayley. {On the theory of linear transformations}. Camb.\ Math.\ J.\ \textbf{4}, 193-209 (1845).  

\bibitem{CFMY-Small} A.\ Ceresole, S.\ Ferrara, A.\ Marrani, A.\ Yeranyan. {Small Black Hole Constituents and Horizontal Symmetry}. JHEP \textbf{06} 078 (2011).  

\bibitem{28} D.D.K.\ Chow, G.\ Comp\`{e}re. Black holes in $\mathcal{N}=8$ supergravity from ${\rm SO}(4,4)$ hidden symmetries. Phys.\ Rev.\ {\bf D90} 2, 025029 (2014).

\bibitem{29} D.D.K.\ Chow, G.\ Comp\`{e}re. Seed for general rotating non-extremal black holes of $\mathcal{N}=8$ supergravity. Class.\ Quant.\ Grav.\ {\bf 31}, 022001 (2014).

\bibitem{rebit-1} N.\ Delfosse, P.A.\ Guerin, J.\ Bian, R.\ Raussendorf. {%
Wigner Function Negativity and Contextuality in Quantum Computation on Rebits%
}. Phys.\ Rev.\ \textbf{X5}, 021003 (2015).  

  \bibitem{SP} F Denef. {Supergravity flows and }$\mathit{D}$\textit{%
    -brane stability}. JHEP \textbf{0008}, 050 (2000).

  \bibitem{SP1}  
    F.\ Denef, B.R.\ Greene, M.\ Raugas. Split attractor flows and the
    spectrum of BPS $\mathit{D}$-branes on the quintic. JHEP \textbf{0105}, 012 (2001).

\bibitem{SP2}
    B.\ Bates, F.\ Denef. Exact solutions for supersymmetric stationary black hole composites. JHEP
\textbf{11} 127 (2011).
    
\bibitem{complex}
  H.\ Dietrich, W.A.\ de Graaf, A.\ Marrani, M.\ Origlia.
  \newblock Classification of four qubit states and their stabilisers under SLOCC operations. 
 J.\ Phys.\ A: Math.\ Theor.\ (doi.org/10.1088/1751-8121/ac4b13 (2022).
  

\bibitem{nilporb}
H.\ Dietrich, W.A.\ de Graaf, D.\ Ruggeri, M.\ Trigiante.
\newblock Nilpotent orbits in real symmetric pairs and stationary black holes.
\newblock Fortschr.\ Phys.\ {\bf 65}, 2, 1600118 (2017).

 


\bibitem{Duff-Cayley} M.J.\ Duff. {String triality, black hole entropy
  and Cayley's hyperdeterminant}. Phys.\ Rev.\ \textbf{D76}, 025017 (2007).

\bibitem{Duff-Ferrara-1} M.J.\ Duff, S.\ Ferrara. $E_{7}$ and the tripartite entanglement of seven qubits. Phys.\ Rev.\ \textbf{D76},  025018 (2007).

\bibitem{Duff-Ferrara-2} M.J.\ Duff, S.\ Ferrara.  $E_{6}$ and the bipartite entanglement of three qutrits. Phys.\ Rev.\ \textbf{D76}, 124023 (2007). 


\bibitem{STU1} M.J.\ Duff, J.T.\ Liu, J.\ Rahmfeld. Four-dimensional string/string/string triality., Nucl.\ Phys.\ {\bf B459}, 125 (1996).
  


\bibitem{AM2} S.\ Ferrara, G.W.\ Gibbons, R.\ Kallosh. Black Holes and
Critical Points in Moduli Space. Nucl.\ Phys.\ \textbf{B500}, 75 (1997).
  
\bibitem{AM1} S.\ Ferrara, R.\ Kallosh, A.\ Strominger. $\mathcal{N}\mathit{=2}$ extremal black holes. Phys.\ Rev.\ \textbf{D52}, 5412 (1995).


\bibitem{AM1-2}   S.\ Ferrara, R.\ Kallosh. Supersymmetry and attractors. Phys.\ Rev.\ \textbf{D54}, 1514 (1996).

\bibitem{AM1-3} S.\ Ferrara, R.\ Kallosh. Universality of
supersymmetric attractors. Phys.\ Rev.\ \textbf{D54}, 1525 (1996).


\bibitem{FMOSY} S.\ Ferrara, A.\ Marrani, E.\ Orazi, R.\ Stora, A.\ Yeranyan.
Two-Center Black Holes Duality-Invariants for $\mathit{STU}$ Model and its lower-rank Descendants. J.\ Math.\ Phys.\ \textbf{52}, 062302 (2011).

\bibitem{ProcStora} S.\ Ferrara, A.\ Marrani. On Symmetries of
Extremal Black Holes with One and Two Centers. Springer Proc.\ Phys.\ \textbf{144}, 345 (2013).

\bibitem{FMY-FI} S.\ Ferrara, A.\ Marrani, A.\ Yeranyan. On Invariant
Structures of Black Hole Charges. JHEP \textbf{02}, 071 (2012).
  
\bibitem{gap4}
  The GAP~Group, \emph{GAP -- Groups, Algorithms, and Programming, 
  Version 4.11.1},  \url{www.gap-system.org} (2021).

\bibitem{Kac-80} V.G.\ Kac. Some remarks on nilpotent orbits, J.\ Alg.\ \textbf{64}, 190-213 (1980).

\bibitem{Kallosh-Linde} R.\ Kallosh, A.D.\ Linde. Strings, black
holes, and quantum information.\ Phys.\ Rev.\ \textbf{D73}, 104033 (2006).

\bibitem{Levay-stringy} P.\ L\'{e}vay. Stringy black holes and the
geometry of entanglement. Phys.\ Rev.\ \textbf{D74}, 024030 (2006).  


\bibitem{Levay-0} P.\ L\'{e}vay. Stringy black holes and the geometry
of entanglement. Phys.\ Rev.\ \textbf{D74},  024030 (2006).

\bibitem{Levay-1} P.\ L\'{e}vay. A Three-qubit interpretation of BPS
and non-BPS STU black holes. Phys.\ Rev.\ \textbf{D76}, 106011 (2007).  

\bibitem{Levay} P.\ Levay. Two-Center Black Holes, Qubits and
  Elliptic Curves. Phys.\ Rev.\ \textbf{D84}, 025023 (2011).
  
\bibitem{Levay-3} P.\ L\'{e}vay. STU Black Holes as Four Qubit Systems. Phys.\ Rev.\ \textbf{D82}, 026003 (2010).  
\bibitem{Levay-2} P.\ L\'{e}vay, S.\ Szalay. The attractor mechanism
as a distillation procedure. Phys.\ Rev.\ \textbf{D82},  026002 (2010).


\bibitem{LT} J.-G.\ Luque, J.-Y.\ Thibon. The polynomial invariants of four qubits. Phys.\ Rev.\ \textbf{A67}, 042303, 1-5 (2003).
  
\bibitem{rebit-0} T.\ Rudolph, L.\ Grover. A 2 Rebit Gate Universal
for Quantum Computing. \url{arXiv:quant-ph/0210187}.
   
\bibitem{Trigiante2} D.\ Ruggeri, M.\ Trigiante. Stationary $D=4$ Black Holes in Supergravity: The Issue of Real Nilpotent
  Orbits. Fortsch.\ Phys.\ {\bf 65}, 5, 1700007 (2017).

\bibitem{serre} J.-P.\ Serre. Galois cohomology. Springer-Verlag (1997).

\bibitem{Stueck} E.C.G.\ Stueckelberg. Quantum Theory in Real Hilbert Space. Helv.\ Phys.\ Acta \textbf{33}, 727 (1960).  
  

\bibitem{AM1-1} A.\ Strominger. Macroscopic entropy of $\mathcal{N}\mathit{=2}$ extremal black holes. Phys.\ Lett.\ \textbf{B383}, 39 (1996).
  
\bibitem{Verstraete} F.\ Verstraete, J.\ Dehaene, B.\ de Moor, H.\ Verschelde.
Four qubits can be entangled in nine different ways. Phys.\ Rev.\ \textbf{A36}, 65,  051001 (2002).  

\bibitem{Vinberg-Weyl} E.B.\ Vinberg. The Weyl group of a graded Lie algebra. Math.\ USSR-Izv.\ \textbf{10}, 463-495 (1976).

\bibitem{rebit-3} W.K.\ Wootters. The rebit three-tangle and its relation to two-qubit entanglement. J.\ Phys.\ {\bf A47}, 424037 (2014).  

\end{thebibliography}
\end{document}